\documentclass{article}
\usepackage[utf8]{inputenc}
\usepackage{cancel} 

\usepackage[top=2.5cm,bottom=2.5cm,right=2.5cm,left=2.5cm]{geometry}

\usepackage{amsmath, amsthm, amsfonts, amssymb}
\usepackage{bbm} 
\usepackage{mathrsfs}
\usepackage[hidelinks]{hyperref}
\usepackage{multicol, tabularx}
\usepackage[inline]{enumitem}
\usepackage{thmtools}
\usepackage{algorithm, algpseudocode}
\usepackage{caption} 
\usepackage{subcaption}
\usepackage{graphicx,float}
\usepackage[
    nameinlink,
    noabbrev,
    capitalise
    ]{cleveref}

\usepackage[dvipsnames]{xcolor}

\newtheorem{theorem}{Theorem}[section]
\newtheorem{proposition}{Proposition}[section]
\newtheorem{corollary}{Corollary}[section]
\newtheorem{example}{Example}[section]
\newtheorem{remark}{Remark}[section]
\newtheorem{assumption}{Assumption}[section]
\newtheorem{definition}{Definition}[section]

\crefname{example}{Example}{Examples}
\Crefname{example}{Example}{Examples}
\crefname{lemma}{Lemma}{Lemmas}
\Crefname{lemma}{Lemma}{Lemmas}
\crefname{thm}{Theorem}{Theorems}
\Crefname{thm}{Theorem}{Theorems}
\crefname{corollary}{Corollary}{Corollaries}
\Crefname{corollary}{Corollary}{Corollaries}

\declaretheorem[
    name=Lemma,
    Refname={Lemma,Lemmas},
    numberwithin=section]{lemma}

\newlist{lemlist}{enumerate}{1}
\setlist[lemlist]{label=(\roman{lemlisti}), ref=\thethm(\roman{lemlisti})}

\crefname{lemlisti}{Lemma}{Lemmas}
\Crefname{lemlisti}{Lemma}{Lemmas}

\usepackage{tikz}
\usepackage{tkz-graph} 

\usetikzlibrary{graphs}
\usetikzlibrary{arrows.meta, chains, positioning, shapes.geometric}
\usetikzlibrary{quotes}
\usetikzlibrary{decorations.pathreplacing}
\usetikzlibrary{calc,fit}

\makeatletter
\newcommand{\mylabel}[2]{#2\def\@currentlabel{#2}\label{#1}}
\makeatother

 \newcommand*\diff{\mathop{}\!\mathrm{d}}

\algnewcommand\algorithmicforeach{\textbf{for each}}
\algdef{S}[FOR]{ForEach}[1]{\algorithmicforeach\ #1\ \algorithmicdo}

\DeclareFontFamily{U}{mathx}{\hyphenchar\font45}
\DeclareFontShape{U}{mathx}{m}{n}{
      <5> <6> <7> <8> <9> <10>
      <10.95> <12> <14.4> <17.28> <20.74> <24.88>
      mathx10
      }{}
\DeclareSymbolFont{mathx}{U}{mathx}{m}{n}
\DeclareMathSymbol{\bigtimes}{1}{mathx}{"91}

\usepackage{etoolbox}
\makeatletter
\patchcmd{\@bibitem}{\ignorespaces}{\label{bib-#1}\ignorespaces}{}{}
\makeatother

\newcommand{\mycite}[1]{[\ref{bib-#1}]}

\newcommand{\myciteintitle}[1]{\texorpdfstring{\cite{#1}}{\mycite{#1}}}

\newcommand{\R}{\mathbb{R}} 

\newcommand{\G}{{G}} 
\renewcommand{\O}{{O}} 
\newcommand{\V}{{V}} 
\newcommand{\E}{{E}} 

\newcommand{\yt}{\Tilde{y}} 
\newcommand{\ot}{\Tilde{o}} 
\newcommand{\Ot}{\Tilde{O}} 

\newcommand{\har}{\rightleftharpoons}
\newcommand{\nothar}{\hspace{0.1em} \har\hspace{-1em}/ \hspace{0.5em}}
\newcommand{\ra}{\rightarrow}
\newcommand{\notrightarrow}{\hspace{0.1em} \rightarrow \hspace{-1em}/ \hspace{0.5em}}
\newcommand{\la}{\leftarrow}

\newcommand{\X}{{ \mathbf{X} }} 
\newcommand{\U}{{ \mathbf{U} }} 

\newcommand{\x}{{ \mathbf{x} }} 
\renewcommand{\u}{{ \mathbf{u} }} 

\newcommand{\cwvparents}{{
c_{w, v | \parentsdown{v}{w}} }}
\newcommand{\cwvparentsbig}[3]{{\cwvparents
\big( #1 \, , \, #2 \, \big| \, #3 \big)}}

\newcommand{\PP}{{ \mathbb{P} }}

\newcommand{\modelcop}{{ \mathscr{C} }}
\newcommand{\modelcopwv}{{ \mathscr{C}_{w \ra v} }}

\newcommand{\thetavec}{{ \boldsymbol{\theta} }} 
\newcommand{\Thetavec}{{ \boldsymbol{\Theta} }} 
\newcommand{\thetavecwv}{{ \thetavec_{w \ra v} }}
\newcommand{\Thetavecwv}{{ \Thetavec_{w \ra v} }}
\newcommand{\cthetawv}{{
c_{\rule{0pt}{0.6em} \thetavecwv } }}
\newcommand{\cthetawvbig}[3]{{\cthetawv
\big( #1 \, , \, #2 \, \big| \, #3 \big)}}

\newcommand{\thetavectrue}{{ \thetavec^* }} 
\newcommand{\thetavecwvtrue}{{ \thetavec_{w \ra v}^* }}

\newcommand{\thetavechat}{{ \hat{\thetavec} }}
\newcommand{\thetavecwvhat}{{ \hat{\thetavec}_{w \ra v} }}

\newcommand{\D}{{\mathscr{D}}} 
\newcommand{\Dwvhat}{{ \hat{\mathscr{D}}_{w \ra v} }}

\newcommand{\phiwv}{{ \phi_{w \ra v} }}
\newcommand{\phiwvtilde}{{ \tilde{\phi}_{w \ra v} }}
\newcommand{\phitilde}{{ \tilde{\phi} }}

\newcommand{\Ovk}{O^k_v}
\newcommand{\BOvk}{{B(O^k_v)}}

\newcommand{\smallerbig}[1]{<_{\displaystyle #1}}

\newcommand{\ad}[1]{{ad}(#1)}
\newcommand{\pa}[1]{{pa}(#1)}
\newcommand{\ch}[1]{{ch}(#1)}

\newcommand{\an}[1]{{an}(#1)}
\newcommand{\de}[1]{{de}(#1)}

\newcommand{\parentsdown}[2]{{ pa (#1 \downarrow #2) }}

\newcommand{\parentsup}[2]{{ pa(#1 \uparrow #2) }} 

\newcommand{\dsepsymbol}{{ \operatorname{dsep} }}
\newcommand{\notdsepsymbol}{{ \cancel{\dsepsymbol} }}

\newcommand{\dsep}[3]{{ \dsepsymbol (#1, #2 | #3) }}
\newcommand{\dsepbig}[3]{{ \dsepsymbol \big(#1, #2 \, \big| \, #3 \big) }}
\newcommand{\dsepBig}[3]{{ \dsepsymbol \Big(#1, #2 \, \Big| \, #3 \Big) }}
\newcommand{\notdsep}[3]{{ \notdsepsymbol (#1, #2 | #3) }}
\newcommand{\notdsepbig}[3]{{ \notdsepsymbol \big(#1, #2 \, \big| \, #3 \big) }}

\newcommand{\TRAILSbig}[3]{{ {TRAILS} \big(#1, #2 \, \big| \, #3 \big) }}

\newcommand{\smallerTRAIL}{<_{TRAIL}}

\newcommand{\cdotslong}{\cdots \cdots \cdots}
\newcommand{\prop}{\mathfrak{P}}

\newcommand{\n}[2]{{n_{{#1}}({#2})}}

\newcommand{\PossCand}[1]{{{PossCand}(#1) }}
\newcommand{\PossCandInd}[1]{{{PossCandInd}(#1) }}
\newcommand{\PossCandIn}[1]{{{PossCandIn}(#1) }}
\newcommand{\PossCandOut}[1]{{{PossCandOut}(#1) }}

\title{Restrictions of PCBNs for integration-free computations}
\author{Alexis Derumigny\thanks{Department of Applied Mathematics, Delft University of Technology, Delft, The Netherlands.
E-mail address: a.f.f.derumigny@tudelft.nl
},
Niels Horsman\thanks{Department of Applied Mathematics, Delft University of Technology, Delft, The Netherlands},
Dorota Kurowicka\thanks{Department of Applied Mathematics, Delft University of Technology, Delft, The Netherlands.
E-mail address: d.kurowicka@tudelft.nl
}
}
\date{\today}

\begin{document}

\maketitle

\begin{abstract}
The pair-copula Bayesian Networks (PCBN) are graphical models composed of a directed acyclic graph (DAG) that represents (conditional) independence in a joint distribution. The nodes of the DAG are associated with marginal densities, and arcs are assigned with bivariate (conditional) copulas following a prescribed collection of parental orders. The choice of marginal densities and copulas is unconstrained. However, the simulation and inference of a PCBN model may necessitate possibly high-dimensional integration.

We present the full characterization of DAGs that do not require any integration for density evaluation or simulations. Furthermore, we propose an algorithm that can find all possible parental orders that do not lead to (expensive) integration. Finally, we show the asymptotic normality of estimators of PCBN models using stepwise estimating equations. Such estimators can be computed effectively if the PCBN does not require integration. A simulation study shows the good finite-sample properties of our estimators.

\medskip


\end{abstract}


\section{Introduction}

One of the main goals of statistics is to recover the unknown distribution of a random vector $\X$, often represented by its density $f_\X$ with respect to some dominating measure. Because of the curse of dimensionality, this is a hard task in general, and one of the way to make this feasible is to use Bayesian Networks (BNs).
BNs are composed of a direct acyclic graph (DAG) where the nodes correspond to each of the random variables of $\X$ and the arcs encode the dependence structure of these variables.
An extremely attractive feature of these models is their ability to represent complex dependencies in an intuitive way. This is especially important for practitioners, who can easily describe their problems and rely on a solid mathematical theory and many computer implementations of BNs.
These models have been applied in a wide variety of fields including medicine, finance, genetics, and forensic science (\cite{GBN_applications}).

\medskip

Each component of $\X$ is represented by an element $v$ in a set $\V$, often chosen to be $\{1, \dots, d\}$ where $\X = (X_1, \dots, X_d)$.
A key property of BNs is that the conditional independencies between components of $\X$, encoded by a graph $\G$ with node set $\V$ and arc set $\E$, allow for a factorization of the joint density into a product of conditional densities:
\begin{equation*}
    f_{\X}(\x) = \prod_{v\in \V} f_{v|\pa{v}}(x_v|\x_{\pa{v}}),
\end{equation*}
where $f_{v|\pa{v}}$ is the conditional density of a node $v \in \V$ given its parents $\pa{v}$, where a node $w$ is said to be a parent of node $v$ if the arc $w \rightarrow v$ is present in $\G$.
This factorization allows us to decompose the problem of estimating the (global) high-dimensional density $f_\X$ into a set of (local) lower-dimensional problems (a node given its parents).
BNs can be used to represent distributions that are purely discrete, purely continuous, or mixed (discrete and continuous with more restrictions, see \cite{Neil_hybrid}).

\medskip

In this paper we consider a particular type of statistical model based on Baysian Networks, which is called pair-copula Bayesian networks (PCBNs).
These models were introduced in \cite{kurowicka_2005}, and were further investigated in 
\cite{Czado_2016, Bauer_2012, Kurowicka2006, Kurowicka2008}.
In PCBNs, the conditional densities $f_{v|\pa{v}}$ are continuous with respect to Lebesgue's measure and decomposed as a product of bivariate (conditional) copulas. Recall that a copula is a distribution on the unit hypercube with uniform margins, and that, by Sklar's theorem, the joint density $f_\X$ can be decomposed as\
\begin{align*}
    f_\X(\x) = c_\X \big( (F_v(x_v))_{v \in \V} \big)
    \times \prod_{v \in \V} f_v(x_v),
\end{align*}
where $c_\X$ is the copula density of $\X$, $f_v$ is the marginal density of $X_v$ and $F_v$ is the marginal cumulative distribution function of $X_v$.
This allows us to separate the estimation of the marginal densities and the copula, which contains all the information about the dependencies between the components of $\X$.

\medskip

In a PCBN model, each arc $p \rightarrow v$ is assigned a continuous (conditional) bivariate copula representing the (conditional) dependence between the random variables $X_p$ and $X_v$; these copulas must be assigned in a specific manner.
If a node $v$ has more than one parent, then a total order $\smallerbig{v}$ is defined over the parental set $\pa{v}$. 
The parents of $v$ are ordered 
and copulas are then assigned as follows.
The arc from the first parent $p_1$ to $v$ is assigned the copula $C_{p_1, v}$, representing the dependence between $X_{p_1}$ and $X_v$.
Then, the arc from the second parent $p_2$ to $v$ is assigned the conditional copula $C_{p_2, v|p_1}$, representing the conditional dependence between $X_{p_2}$ and $X_v$ given $X_{p_1}$.
The arc from the third parent $p_3$ to $v$ is assigned the conditional copula $C_{p_3, v|p_1,p_2}$, representing the conditional dependence between $X_{p_3}$ and $X_v$ given $(X_{p_1}, X_{p_2})$ and so on.
Therefore, the conditional density $f_{v|\pa{v}}$ can be decomposed as
\begin{align}
    f_{v|\pa{v}}(x_v|\x_{\pa{v}})
    &= f_{v}(x_v) \cdot
    c_{p_1, v}(u_{p_1}, u_v) \cdot
    c_{p_2, v | p_1}(u_{p_2 | p_1}, u_{v | p_1} | X_{p_1} = x_{p_1}) 
    \label{eq:decomposition_f_v_pa_v} \\
    & \cdot \ldots \cdot
    c_{p_m, v | p_1, \dots, p_{m-1}}(
    u_{p_m | p_1, \dots, p_{m-1}},
    u_{v | p_1, \dots, p_{m-1}} \, | \,
    X_{p_1} = x_{p_1}, \dots,
    X_{p_{m-1}} = x_{p_{m-1}} ),
    \nonumber
\end{align}
where $m$ is the number of parents of $v$ in the graph
and $u_{w | S} := F(x_w | \X_S = \x_S)$
for any node $w \in \V$ and any set $S \subset \V$.
Thus, each arc $w\rightarrow v$ is assigned the conditional copula $C_{w, v | \parentsdown{v}{w}}$, where $\parentsdown{v}{w}$ is the set consisting of all parents of $v$, which are earlier than $w$ according to $\smallerbig{v}$.
It has been shown in \cite{kurowicka_2005} that such an assignment of copulas is consistent and provides us with a proper joint density $f_\X$, whose copula is given by
\begin{align*}
    c_\X(\u) = \prod_{v \in \V}
    \prod_{w \in \pa{v}}
    \cwvparentsbig
    {u_{w | \parentsdown{v}{w}}}
    {u_{v | \parentsdown{v}{w}}}
    {u_{\parentsdown{v}{w}}},
\end{align*}
for $\u \in [0,1]^d$, where $d$ is the dimension of $\X$.
Furthermore, if all copulas and marginal distributions in the PCBN are Gaussian, then the PCBN is equivalent to the Gaussian Bayesian Network, see \cite{Lauritzen2001,KollerFriedman,lauritzen_handbook,Scutari_2019}.

\medskip

Parental orders for all nodes are collected in the set $\O := (\smallerbig{v})_{v \in \V}$.
A PCBN includes the tuple $(\G,\O)$ where the graph determines (conditional) independencies between elements of the random vector, and the parental orders indicate (conditional) copula assignments.
Additionally, the copula types have to be determined and their parameters estimated as well as the marginal distributions of all nodes.
PCBNs are much more expensive computationally as compared to Gaussian Bayesian Networks, but they  can represent a much more flexible set of dependencies \cite{Bauer_2012}.

\medskip

To compute $c_\X(\u)$, the terms $u_{w|\parentsdown{v}{w}}$ and $u_{v|\parentsdown{v}{w}}$ are needed.
These conditional margins may require integration. 
\begin{figure}[H]
\centering
    \begin{tikzpicture}[scale=1, transform shape, node distance=1.2cm, state/.style={circle, draw=black}]
    \begin{scope}[every node/.style={circle, minimum size=0.8cm,thick,draw}]
        \node[state] (1) {$1$};
        \node[state, below left = of 1] (2) {$2$};
        \node[state, below right = of 1] (3) {$3$};
        \node[state, below right = of 2] (4) {$4$};
    \end{scope}
    \begin{scope}[>={Stealth[black]}]
        \path [->] (1) edge node[scale=0.8, below, sloped] {$1,2$} (2);
        \path [->] (1) edge node[scale=0.8, below, sloped] {$1,3$} (3);
        \path [->] (2) edge node[scale=0.8, below, sloped] {$2,4$} (4);
        \path [->] (3) edge node[scale=0.8, below, sloped] {$3,4|2$} (4);
    \end{scope}
    \end{tikzpicture}
\caption{Diamond PCBN on four nodes.}
\label{fig:PCBN_diamond}
\end{figure}
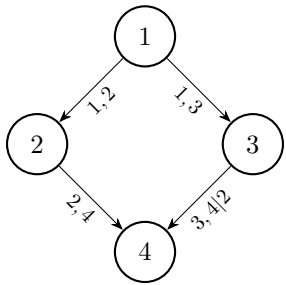

For example, in \cite{Czado_2016}, the graph presented in Figure~\ref{fig:PCBN_diamond} was found to require integration for any assignment of parental orders $\O$.
Note that for this graph we have two possible choices of orders for node 4:
$2 <_4 3$ and $3 <_4 2$.
When we pick $2 <_4 3$, as presented in Figure~\ref{fig:PCBN_diamond}, the copula density is
\begin{align*}
    c(u_1, u_2, u_3, u_4)
    = c_{1, 2}(u_1,u_2)
    \cdot c_{1, 3}(u_1, u_3)
    \cdot c_{2, 4}(u_2, u_4)
    \cdot c_{3, 4|2}(u_{3|2}, u_{4|2} | u_2). 
\end{align*}

The conditional margin $u_{3|2}$, which depends on $u_2$ and $u_3$ must in general be computed using integration. This is due to the fact that the marginalization in pair-copula based models cannot be performed analytically. We have that  
\begin{align*}
    c_{2, 3}(u_2, u_3)
    = \int_0^1 \int_0^1 c(u_1,u_2,u_3,u_4)du_1du_4 =
    \int_0^1 c_{1, 2}(u_1,u_2)
    \cdot c_{1, 3}(u_1, u_3)du_1
\end{align*}
which in general cannot be simplified any further. 
\medskip

In this paper the full characterization of graph structures that do not require integration in the density evaluation is presented. We show that if one restricts the structure of $\G$ by disallowing certain induced subgraphs then there exists a set of parental orders $\O$ for which the evaluation of density is free of integration.  We provide an algorithm that assigns copulas to the arcs of a restricted graph such that the joint density does not require integration.

\medskip

The rest of the paper is organized as follows. Section~\ref{sec:DAGs_and_BNs} presents the background information on DAGs and BNs.
Then in Section~\ref{sec:PCBN}, PCBNs are introduced and the conditions that lead to the need for integration in the density evaluation are studied. Section~\ref{sec:restricted_PCBN} details the restrictions of PCBNs that are sufficient and necessary for integration-free computations.
One of the main theorems of the paper, which guarantees that these restrictions are indeed sufficient for integration-free computations, 
is proved in Section~\ref{sec:proof:existence_possible_candidates},
while supporting lemmas can be found in the appendix.
Estimation of PCBN is studied in Section~\ref{sec:EstimationPCBN}.
The proposed methodology is implemented in the \texttt{R} package \texttt{PCBN} \cite{PCBN}.

\section{Background on DAGs and BNs}
\label{sec:DAGs_and_BNs}

This section contains the background information necessary in later parts of the paper.

\subsection{Directed graphs}
\label{sec:directed_graphs}
Let $\G = (V,E)$ be a directed graph with nodes $V$ and arcs $E$.
We consider only simple graphs without loops and multiple arcs.
Moreover, let $\G^*$ be the associated undirected graph called \textbf{skeleton} of $G$, obtained from $G$ by removing the directions of the arcs.
We say that $G'=(V',E')$ is a \textbf{subgraph} of $G$
if $V'\subseteq \V$ and  $E'\subseteq E$ and for all arcs $w \rightarrow v \in E'$ the nodes $w$ and $v$ are in $V'$.
If $E'$ contains all arcs in $G$ between nodes in $V'$, then $G'$ is said to be \textbf{induced} by $V'$. 
A \textbf{path} is a sequence of nodes $(v_1,v_2,\dots,v_n)$ such that $\{v_1,v_2,\dots,v_n\}\subseteq \V$ and $\{(v_1,v_2), (v_2,v_3), \dots, (v_{n-1},v_n) \} \subseteq \E$ for some integer $n > 0$ called the \textbf{length} of the path.
A \textbf{trail} is a sequence of nodes in $\G$ that forms an undirected path in $\G^*$, for which we use the notation $v_1 \har \cdots \har v_n$. Moreover, two nodes that are connected by an edge in $\G^*$ are called \textbf{adjacent}.
An arc between non-consecutive nodes in a given trail is referred to as a \textbf{chord}.
A path of the form $\{v_1,v_2,\dots,v_n,v_1\}$ is called a \textbf{cycle}.
We call $G$ \textbf{acyclic} if it does not contain any cycle. In this case $G$ is a directed acyclic graph (DAG).

\medskip

For each arc $w \rightarrow v\in \E$ the node $w$ is said to be the \textbf{parent} of $v$ and $v$ is said to be the \textbf{child} of $w$.
For a node $v\in \V$ the sets containing all its parents and children are denoted by $\pa{v}$ and $\ch{v}$, respectively.
If there exists a path from $w$ to $v$, then $w$ is said to be an \textbf{ancestor} of $v$ and $v$ is said to be a \textbf{descendant} of $w$.
For a node $v\in \V$ the sets containing all its ancestors and descendants are denoted by $\an{v}$ and $\de{v}$, respectively.

\medskip

If a node has at least two parents $v_1$ and $v_2$ then we say that $(v_1, v, v_2)$ is a \textbf{v-structure} at $v$ and
when $v$ has at least two children $v_1$ and $v_2$, $(v_1, v, v_2)$ is referred to as a \textbf{diverging connection}
Moreover, paths $(v_1, v_2, v_3 )$ or
$(v_3, v_2, v_1)$ in $\G$ will be called \textbf{serial connections}. 

\medskip

An important concept in graphical models and in particular in BNs whose qualitative part is represented by a directed graph, is that two subsets of nodes can be connected through trails.
These trails can be either blocked or activated given another subset of nodes. 

\begin{definition}[d-separation]\label{def:graph_dsep}
    Let $\G=(\V,\E)$ be a directed graph and
    let $S, Y, Z \subseteq \V$ be disjoint sets.
    Then, $Z$ is said to \textbf{d-separate} $S$ and $Y$ in $\G$, denoted by $\dsep{S}{Y}{Z}$, if every trail $v_1\har v_2 \har \cdots \har v_n$ with $v_1\in S$ and $v_n\in Y$ contains at least one node $v_i$ satisfying one of the following conditions:
        \begin{itemize}
            \item The trail forms a v-structure at $v_i$, i.e. $v_{i-1}\rightarrow v_{i}\leftarrow v_{i+1}$, and the set $\{v_i\}\sqcup de(v_i)$ is disjoint from $Z$.
            \item The trail does not contain a v-structure at $v_i$ and $v_i\in Z$.
        \end{itemize}
    If a trail satisfies one of the conditions above, it is said to be \textbf{blocked} by $Z$, else it is \textbf{activated} by $Z$.
    Furthermore, if $S$ and $Y$ are not d-separated by $Z$, we use the notation $\notdsep{S}{Y}{Z}$.
    Moreover, if a set $S$ or $Y$ is empty, then by convention $\dsep{S}{Y}{Z}$ holds.
\end{definition}

\subsection{Bayesian networks}

A graphical model is a representation of the distribution of a multivariate random vector $\X = (X_1, \dots, X_d)$ in terms of a graph.
Each node $v \in \V$ corresponds to a univariate random variable $X_v$, which is the $v$-th component of $\X$.
In this paper all random vectors are assumed to be absolutely continuous.
We denote by $f_v$ the probability density function (pdf) of $X_v$. 
For $K \subseteq \V$ we write 
$\X_{K} := (X_v)_{v\in K}$
and its pdf is denoted by $f_{K}$.
The cardinal of $K$ is denoted by $|K|$.
Furthermore, the pdf of a random variable $X_v$ conditional on $\X_K$ with $K\subseteq \V\setminus \{v\}$ is denoted by $f_{v|K}$ and the corresponding conditional cumulative distribution function is denoted by $F_{v|K}$.

\begin{definition}[Bayesian network]
    A \textbf{Bayesian network} (BN) is a graphical model composed of
    \begin{itemize}
        \item a DAG $\G = (\V,\E)$ where the nodes correspond to univariate random variables and the arcs describe the conditional independence through d-separation, in the sense that for any disjoints sets $S, Y, Z \subseteq \V$, $\dsep{S}{Y}{Z}$ implies that $\X_S$ and $\X_Y$ are independent given $\X_Z$;
        
        \item a sequence of conditional densities $(f_{v|\pa{v}})_{v \in \V }$.
    \end{itemize}
\end{definition}

The set of conditional independence statements given by $E$ allows for the decomposition of the joint density of $\X$ as a product of the specified conditional densities:
\begin{equation}
    f_\X(\x)
    = \prod_{v \in \V} f_{v|\pa{v}}(x_v|\x_{\pa{v}}).
    \label{eq:BN_markovian}
\end{equation}
Note that different graphical structures can induce the same set of conditional independence statements.
Such graphical structures are then called \textbf{equivalent}.

\medskip

As seen in the second point of the definition, a BN require the specification of all conditional densities $f_{v|\pa{v}}$.
The most popular BNs are discrete and Gaussian BNs, i.e. where each conditional density is either a density with respect to the counting measure (discrete case) or a Gaussian density with respect to Lebesgue's measure (Gaussian case).
Nevertheless, in practice it is rare that random variables follow Gaussian distributions, and it is necessary to have more flexible models that can adequately represent real-life distributions.
This is why, in this paper, we consider a copula-based Bayesian Network, which is presented next.

\section{Pair-Copula Bayesian Networks} \label{sec:PCBN}

\subsection{Introduction}

Each conditional density  $f_{v|pa(v)}$ in the density decomposition~\eqref{eq:BN_markovian}, can be rewritten as a product of the marginal density $f_v$ and the (conditional) bivariate copula densities as seen in \eqref{eq:decomposition_f_v_pa_v}:
the arc $w \rightarrow v$ is assigned the copula $c_{w, v | \parentsdown{v}{w}}$,
where sets $\parentsdown{v}{w}$ and  $\parentsdown{v}{w}$  defined below.

\begin{definition}[Parental order]
\label{def:graph_pa(v;w)}
    Let $\G=(\V, \E)$ be a directed graph and $v \in \V$ be a node.
    A \textbf{parental order} of $v$ is a total order on the set $\pa{v}$ denoted by $<_v$.
    For all $w\in \pa{v}$, the set of parents of $v$ \textbf{strictly up to} $w$
    (respectively, \textbf{after} $w$)
    is defined as
    $\parentsdown{v}{w} := \{z\in \pa{v};\, z<_v w\}\;$
    (respectively, $\;\parentsup{v}{w} := \{z\in \pa{v};\, w<_v z\}$).
\end{definition}

\medskip

The formal definition of the PCBN including the parental order for each node is presented.
\begin{definition}[Pair-copula Bayesian network]
\label{def:PCBN}
    A \textbf{pair-copula Bayesian network} is defined as a collection
    \({(\G, \O, (f_v)_{v \in \V},
    (C_{w, v|\parentsdown{v}{w}})_{w \ra v \in \E})}\)
    where
    \begin{itemize}
        \item the pair $(\G,\O)$, called the \textbf{structure of the PCBN} consists of
        a DAG $\G$ and a collection of orderings
        $\O = (\smallerbig{v})_{v \in \V}$,
        \item $(f_v)_{v \in \V}$ is a collection of univariate densities,
        \item 
        $(C_{w, v|\parentsdown{v}{w}})_{w \ra v \in \E}$
        is a collection of (conditional) copulas.
    \end{itemize}
\end{definition}

The set of conditional independencies allows for the decomposition of the joint density of a PCBN as a product of the marginal densities and copulas;
\begin{align}
    \label{eq:PCBN_Fac} f_{\X}(\x)
    = \prod_{v\in \V} f_v(x_v)\prod_{w\in \pa{v}} 
    \cwvparentsbig
    {F_{w|\parentsdown{v}{w}}(x_w|\x_{\parentsdown{v}{w}})}
    {F_{v|\parentsdown{v}{w}}(x_v|\x_{\parentsdown{v}{w}})}
    {\x_{\parentsdown{v}{w}}}.
\end{align}

To simplify the notation, for any node $v \in \V$ and for any set $S \subset \V$, we define $u_{v|S} := \PP(U_v \leq u_v | \U_S = \u_s)$.
In particular, we will denote 
$F_{v|\parentsdown{v}{w}}(x_w|\x_{\parentsdown{v}{w}})$
by $u_{v|\parentsdown{v}{w}}$.

\medskip
The graph in \cref{fig:vstruc_3_parents}, with order $1 <_4 2 <_4 < 3$, has copulas $c_{1,4}$, $c_{2, 4 | 1}$ and $c_{3, 4 | 1, 2}$ assigned to its arcs.
The corresponding copula density is as follows:
\begin{align*}
    c_{1,2,3,4}(u_1,u_2,u_3,u_4)
    = c_{1,4}(u_1,u_4) 
    \cdot c_{2, 4 | 1}(u_2,u_{4|1}|u_1)
    \cdot c_{3, 4 | 1, 2}(u_3,u_{4|1,2}|u_1,u_2).
\end{align*}
We see that $X_1$, $X_2$ and $X_3$ are mutually independent.

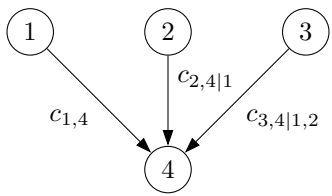
\begin{figure}[H]
    \centering
    \begin{tikzpicture}[scale=1, transform shape, node distance=1.2cm, state/.style={circle, draw=black}]
    
    \node[state] (X1) {$1$};
    \node[state, right = of X1] (X2) {$2$};
    \node[state, right = of X2] (X3) {$3$};
    \node[state, below = of X2] (X4) {$4$};
    
    \begin{scope}[>={Stealth[length=6pt,width=4pt,inset=0pt]}]
    
    \path [->] (X1) edge node[below left, scale=1] {$c_{1, 4}$} (X4);
    \path [->] (X2) edge node[above right, scale=1] {$c_{2, 4 | 1}$} (X4);
    \path [->] (X3) edge node[below right, scale=1] {$c_{3, 4 | 1, 2}$} (X4);
    
    \end{scope}
    \end{tikzpicture}
    \caption{PCBN where node $X_4$ has corresponding parental order $1 <_4 2 <_4 3$.}
    \label{fig:vstruc_3_parents}
\end{figure}

Note that we do not assume that the copula $C_{w, v|\parentsdown{v}{w}}$ is constant with respect to $\x_{\parentsdown{v}{w}}$.
Such assumption is known as the ``simplifying assumption'' (see e.g. \cite{derumigny2017tests} for a review), and is not needed in this paper.  All the results presented in this paper are valid in both cases whether the simplifying assumption is made or not.

\medskip

From \cref{eq:PCBN_Fac}, it is clear that the  density can be written as a product of all marginal densities of $\X$,
and a copula density.
This copula density $c_\X$  corresponding to $f_\X$ is then decomposed as a product of (conditional) copula densities
assigned to arcs in the graph, by
\begin{align}
\label{eq:PCBN_fac_unif}
    c_\X(\u) = \prod_{v\in \V \vphantom{\pa{v}}} \prod_{w\in \pa{v}}
    c_{w, v | \parentsdown{v}{w}}\big(
    u_{w | \parentsdown{v}{w}} ,
    u_{v | \parentsdown{v}{w}} 
    \, | \, \u_{\parentsdown{v}{w}}
    \big).
\end{align}

\medskip

The PCBN structure $(\G,\O)$ provides us with a collection of conditional copulas which are assigned to arcs in the graph:
$(C_{w, v|\parentsdown{v}{w}})_{w \ra v \in \E}$.
These copulas are said to be specified by the PCBN. 
Furthermore, the graph of a PCBN induces conditional independencies between random variables that can be read from the graph through the d-separation.
If two nodes are d-separated given set of nodes, then the conditional copula of variables corresponding to these nodes is also known; it is the independence copula.
Finally, adding conditionally independent variables to the conditioning set of an already specified copulas yields a copula that is still specified (see Figure~\ref{fig:example_PCBN_EE_specified}).
Hence, we formalize when $C_{w, v|K}$ is specified.

\begin{figure}[H]
    \centering
    \begin{tikzpicture}[scale=1, transform shape, node distance=1.2cm, state/.style={circle, draw=black}]
    
    \node[state] (X1) {$1$};
    \node[state, right = of X1] (X2) {$2$};
    \node[state, right = of X2] (X3) {$3$};
    \node[state, below = of X2] (X4) {$4$};
    \node[state, right = of X4] (X5) {$5$};
    
    \begin{scope}[>={Stealth[length=6pt,width=4pt,inset=0pt]}]
    
    \path [->] (X1) edge node[below left, scale=1] {$c_{1, 4}$} (X4);
    \path [->] (X2) edge node[above right, scale=1] {$c_{2, 4 | 1}$} (X4);
    \path [->] (X3) edge node[below right, scale=1] {$c_{3, 4 | 1, 2}$} (X4);
    
    \end{scope}
    \end{tikzpicture}
    \caption{PCBN where $C_{2,4|1,5} = C_{2,4|1}$ is specified.}
    \label{fig:example_PCBN_EE_specified}
\end{figure}
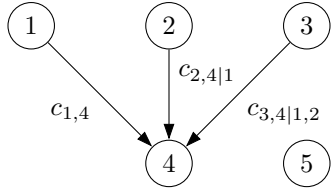

\begin{definition}\label{def:specified}
    Consider a PCBN with node set $\V$, and let $w, v \in \V$.
    We say that the following (conditional) copulas are {\bf specified}:
    \begin{enumerate}[label=(\roman*)]
        \item $C_{v, w | \parentsdown{v}{w}}$ and
        $C_{w, v | \parentsdown{v}{w}}$,
        if $w \rightarrow v$,
        since these copulas are explicitly specified in the PCBN.
        
        \item $C_{w, v | K}$,
        where $K \subseteq \V \setminus \{w, v\}$ and $\dsep{w}{v}{K}$. In this case, $C_{w, v | K}$ is known from the graph structure to be the independence copula.

        \item $C_{w, v | \parentsdown{v}{w} \sqcup J}$,
        where $w \rightarrow v$ and
        $J \neq \emptyset$ is such that
        $U_v$ is independent of 
        $\U_J$ given
        $(U_w, \U_{\parentsdown{v}{w}})$, i.e.
        $\dsepbig{v}{J}{\parentsdown{v}{w} \sqcup\{w\}}$.
        Then $C_{w, v | \parentsdown{v}{w} \sqcup J}$
        is equal to the explicitly specified copula
        $C_{v, w | \parentsdown{v}{w}}$.
    \end{enumerate}
\end{definition}

\begin{remark}
\label{remark:necessary_condition_Copula_specified}
    Note that a necessary condition for $C_{w, v | K}$ to be specified is that
    $v$ and $w$ are either adjacent or d-separated by $K$.
    Therefore, copulas of the form $C_{w, v | K}$ where $v$ and $w$ are neither adjacent nor d-separated by $K$ always require integration.
\end{remark}

\begin{remark}
\label{remark:cannot_remove_nodes_parent_v_w}
    Note that copulas $C_{w, v | K}$
    that are obtained from $C_{w, v | \parentsdown{v}{w}}$
    by removing nodes from the conditioning set $\parentsdown{v}{w}$ are never specified.
    They always must be computed by integration with respect to the nodes that need to be removed.
    More generally,
    \begin{equation}
        w \rightarrow v \text{ and }
        \exists o \in \parentsdown{v}{w}
        \text{ such that } o \notin K
    \label{eq:cannot_remove_nodes_parent_v_w}
    \end{equation}
    is a sufficient condition for $C_{w, v | K}$ to be not specified.
\end{remark}

\subsection{Problematic conditional margins for PCBNs}

Equation~\eqref{eq:PCBN_fac_unif} requires the computation of the conditional margins $u_{w | \parentsdown{v}{w}}$ and
$u_{v | \parentsdown{v}{w}}$, for $v \in V$ and $w \in \pa{v}$.
This means that $u_{w | \parentsdown{v}{w}}$ and
$u_{v | \parentsdown{v}{w}}$ must be computed in order to evaluate the density $f_\X$.
In this section, we will see two examples where the term
$u_{w | \parentsdown{v}{w}}$
cannot be computed without using integration.

\begin{example}    
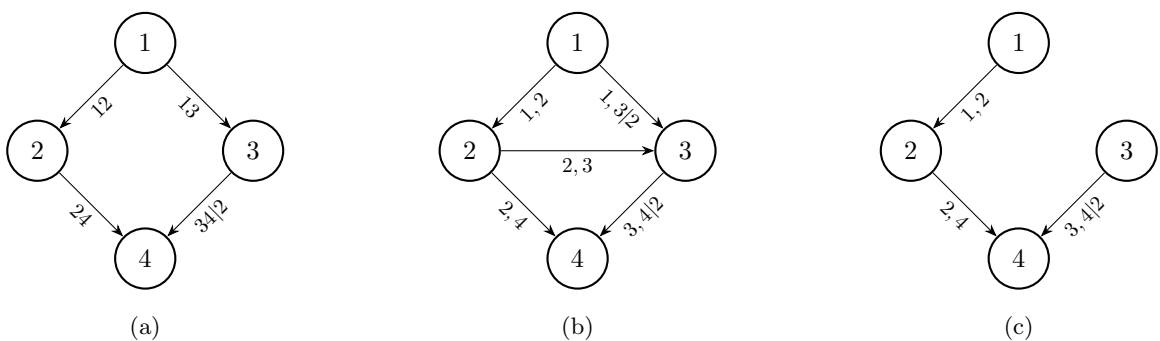
\begin{figure}[H]
\centering
\begin{subfigure}{.3\textwidth}
    \centering
    \begin{tikzpicture}[scale=1, transform shape, node distance=1.2cm, state/.style={circle, draw=black}]
    \begin{scope}[every node/.style={circle, minimum size=0.8cm,thick,draw}]
        \node[state] (1) {$1$};
        \node[state, below left = of 1] (2) {$2$};
        \node[state, below right = of 1] (3) {$3$};
        \node[state, below right = of 2] (4) {$4$};
    \end{scope}
    \begin{scope}[>={Stealth[black]}]
        \path [->] (1) edge node[scale=0.8, below, sloped] {$12$} (2);
        \path [->] (1) edge node[scale=0.8, below, sloped] {$13$} (3);
        \path [->] (2) edge node[scale=0.8, below, sloped] {$24$} (4);
        \path [->] (3) edge node[scale=0.8, below, sloped] {$34|2$} (4);
    \end{scope}
    \end{tikzpicture}
\caption{}
\label{fig:PCBN_integration_a}
\end{subfigure}%
\hfill
\begin{subfigure}{.3\textwidth}
    \centering
    \begin{tikzpicture}[scale=1, transform shape, node distance=1.2cm, state/.style={circle, draw=black}]
    \begin{scope}[every node/.style={circle, minimum size=0.8cm,thick,draw}]
        \node[state] (1) {$1$};
        \node[state, below left = of 1] (2) {$2$};
        \node[state, below right = of 1] (3) {$3$};
        \node[state, below right = of 2] (4) {$4$};
    \end{scope}
    \begin{scope}[>={Stealth[black]}]
        \path [->] (1) edge node[scale=0.8, below, sloped] {$1,2$} (2);
        \path [->] (1) edge node[scale=0.8, below, sloped] {$1,3|2$} (3);
        \path [->] (2) edge node[scale=0.8, below, sloped] {$2,4$} (4);
        \path [->] (2) edge node[scale=0.8, below, sloped] {$2,3$} (3);
        \path [->] (3) edge node[scale=0.8, below, sloped] {$3,4|2$} (4);
    \end{scope}
    \end{tikzpicture}
\caption{}
\label{fig:PCBN_integration_b}
\end{subfigure}
\hfill
\begin{subfigure}{.3\textwidth}
  \centering
     \centering
    \begin{tikzpicture}[scale=1, transform shape, node distance=1.2cm, state/.style={circle, draw=black}]
    \begin{scope}[every node/.style={circle, minimum size=0.8cm,thick,draw}]
        \node[state] (1) {$1$};
        \node[state, below left = of 1] (2) {$2$};
        \node[state, below right = of 1] (3) {$3$};
        \node[state, below right = of 2] (4) {$4$};
    \end{scope}
    \begin{scope}[>={Stealth[black]}]
        \path [->] (1) edge node[scale=0.8, below, sloped] {$1,2$} (2);
     
        \path [->] (2) edge node[scale=0.8, below, sloped] {$2,4$} (4);
     
        \path [->] (3) edge node[scale=0.8, below, sloped] {$3,4|2$} (4);
    \end{scope}
    \end{tikzpicture}
\caption{}
\label{fig:PCBN_integration_c}
\end{subfigure}
\caption{PCBNs where the computation of $u_{3|2}$ can be done without integrating for (b) and (c), but not for (a).
}
\label{fig:PCBN_integration}
\end{figure}

\noindent The copula density corresponding to the PCBN in \cref{fig:PCBN_integration_a} \footnote{This example was already shortly discussed in the Introduction.} is as follows:
\begin{align*}
    c(\u) = c_{1, 2}\big(u_1, u_2\big)
    \cdot c_{1, 3}\big(u_1, u_3\big)
    \cdot c_{2, 4}\big(u_2, u_4\big)
    \cdot c_{3, 4|2} \big(u_{3|2}, u_{4|2} \, | \, u_2 \big).
\end{align*}
To compute this copula density, we need the copulas $c_{1, 2}, c_{1, 3}, c_{2, 4}, c_{3, 4|2}$, which are already specified, but we also need to compute the conditional marginals $u_{3|2}$ and $u_{4|2}$.
By definition, $u_{3|2} = \frac{\partial C_{2,3}(u_2, u_3)}{\partial u_2}$.
Since the copula $c_{2,3}$ is not specified by the PCBN, it needs to be computed.
We know that
\begin{align*}
    c_{2,3}(u_2, u_3)
    = \int_{0}^1 c_{1,2,3}(w_1, u_2, u_3) dw_1
    = \int_0^1 c_{1, 2} \big(w_1, u_2\big)
    \cdot c_{1, 3} \big(w_1, u_3\big) \diff w_1.
\end{align*}
As a consequence, we have
\begin{align*}
    u_{3|2}
    &= \frac{\partial }{\partial u_2}
    \int_0^{u_2} \int_0^{u_3} c_{2,3}(w_2, w_3)  \diff  w_2  \diff w_3
    = \int_0^{u_3} c_{2,3}(u_2, w_3)  \diff w_3 \\
    &= \int_0^1 c_{1, 2} \big(w_1, u_2\big)
    \int_0^{u_3}  c_{1, 3} \big(w_1, w_3\big) dw_3 \diff w_1 \\
    &= \int_0^1 c_{1,2}(w_1,u_2) \cdot 
    \frac{\partial C_{1,3}(w_1, u_3)}{\partial u_1}
    \, \diff w_1
    = \int_0^1 c_{1,2}(w_1,u_2) \cdot
    h_{1, 3}(w_1,u_3)\, \diff w_1.
\end{align*}
Hence, integration is needed, since $c_{1,2}$ and $C_{1,3}$ can be specified arbitrarily.
Note that the conditional margin $u_{4|2}$ does not pose a problem since it can be computed using the copula $C_{2, 4}$, which is specified by the PCBN.

\medskip

On the contrary, for the PCBNs in
\cref{fig:PCBN_integration_b,fig:PCBN_integration_c}
there is no problem with computing the conditional margin $u_{3|2}$.
Indeed, for the PCBN in \cref{fig:PCBN_integration_b} the copula $c_{2, 3}$ is assigned to the arc $v_2 \rightarrow v_3$.
For the PCBN in \cref{fig:PCBN_integration_c}, $c_{2,3}$ is also specified because it is the independent copula, due to $\dsepbig{2}{3}{\emptyset}$.

\medskip

It is important to notice that even for the same graph structure, certain ordering $\O$ necessitate integration, while other does not.
For example, if the PCBN \cref{fig:PCBN_integration_b}, we instead of $2 <_3 1$ choose the order $1 <_3 2$, then the copula $c_{2,3}$ is not specified by the PCBN.
Indeed, in this case copulas $c_{1,3}$ and $c_{2,3|1}$ are specified.
Hence, the computation of conditional margin $u_{3|2}$ requires integration.

\end{example}

\begin{example}\label{ex:PCBN_interfering}
Consider the PCBN in \cref{fig:PCBN_interferin_vstruc}.
The orderings at nodes $v_4$ and $v_5$ have been chosen and we need to determine the order at node $v_3$.
The v-structures at nodes $v_4$ and $v_5$ require us to compute the conditional margins $u_{3|1}$ and $u_{2|3}$, respectively.
The former implies that $c_{1,3}$ must be specified, and thus we must have $1<_3 2$.
But, the latter requires $c_{2,3}$ to be specified, implying that $2 <_3 1$.
Since we cannot have both, there does not exist a suitable ordering for node $v_3$.    
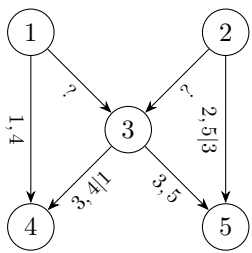
\begin{figure}[H]
  \centering
    \begin{tikzpicture}[scale=1, transform shape, node distance=1.2cm, state/.style={circle, draw=black}]

        \node[state] (3) at (3,0) {$3$};
        \node[state, above left = of 3] (1) {$1$};
        \node[state, above right = of 3] (2) {$2$};
        \node[state, below left = of 3] (4) {$4$};
        \node[state, below right = of 3] (5) {$5$};

        \begin{scope}[>={Stealth[black]}]
            \path [->] (1) edge node[scale=0.8, sloped, below] {?} (3);
            \path [->] (2) edge node[scale=0.8, sloped, below] {?} (3);
            \path [->] (1) edge node[scale=0.8, sloped, below] {$1,4$} (4);
            \path [->] (3) edge node[scale=0.8, sloped, below] {$3,4|1$} (4);
            \path [->] (3) edge node[scale=0.8, sloped, below] {$3,5$} (5);
            \path [->] (2) edge node[scale=0.8, sloped, below] {$2,5|3$} (5);
        \end{scope}
    \end{tikzpicture}
\caption{PCBN with interfering v-structures.}
\label{fig:PCBN_interferin_vstruc}
\end{figure}

\end{example}

If $\G$ contains three v-structures that interact in a similar fashion as in \cref{ex:PCBN_interfering}, then the joint density will require integration for any choice of $\O$.
Such problematic set of three v-structures will be called interfering v-structures (see \cref{def:interfering_v_structures} below).

\medskip

All copulas $C_{w_i,v|w_1,...,w_{i-1}}$, where $\pa{v}=\{w_1,...,w_s\}$ are specified by construction. Hence computation of $u_{v|\parentsdown{v}{w}}$ is in general not problematic. The main obstacle is in the computation of terms of the form $u_{w | \parentsdown{v}{w}}$. These conditional margins are computed recursively. Let $v\in\V$ and $K \subseteq \V \setminus \{v\}$ then
\begin{eqnarray}\label{recursions}
     u_{v|K}
   & =& \frac{\partial C_{k, v | K \setminus k}
    (u_{k | K \setminus k } ,
    u_{v | K \setminus k } \, | \, \u_{K \setminus k} )}{
    \partial u_{k | K \setminus k }
    }.
\end{eqnarray}
   
First we formally define such recursions, which will be used repeatedly throughout the paper. 

\begin{definition}
\label{def:recursion_cond_margin}
    Consider a PCBN with a node set $\V$.
    Let $v \in \V$ and $K \subseteq \V \setminus \{v\}$.
    $\mathcal{R}(v|K)$ is \textbf{the set of recursions} for the computation of $u_{v|K}$ in (\ref{recursions}) if
    \begin{enumerate}[label=(\roman*)]
        \item 
        $\mathcal{R}(v | \emptyset) := \{u_v\}$.
        \item 
        $\mathcal{R}(v|K) := \bigsqcup_{k \in K}
        \big\{ \big(
        u_{v|K},
        k,
        C_{v, k | K \setminus k},
        R_v,
        R_k \big):
        R_v \in \mathcal{R}(v|K \setminus k), 
        R_k \in \mathcal{R}(k|K \setminus k) \big\}.$
    \end{enumerate}
\end{definition}


\begin{remark}
\label{remark:recursion_around_v}
    From Definition~\ref{def:recursion_cond_margin} it follows that  any recursion $R \in \mathcal{R}(v|K)$ requires the choice of some $k \in K$, which determines a conditional copula $C_{v, k | K \setminus k}$, and a recursion $R_v \in \mathcal{R}(v|K \setminus k)$.
    From this branch of the recursion, there exists an ordering $(k_1, \dots, k_{|K|})$ of $K$ such that copulas
    \begin{equation*}
        \forall i \in \{1, \dots, |K|\}, \,
        C_{v, k_i | K \setminus \{k_i, \dots, k_{|K|}\} }
        = C_{v, k_i | k_1, \dots, k_{i-1}}
        \text{ appear in } R.
    \end{equation*}
    The ordering of $K$ depends on the choice of the recursion $R \in \mathcal{R}(v|K)$.
\end{remark}

\begin{definition}\label{def}
    We say that a recursion $R \in \mathcal{R}(v|K)$ is \textbf{proper} if all (conditional) copulas appearing in $R$ are specified. Moreover, $u_{v|K}$ {\bf does not require integration} if there exist a proper recursion $R \in \mathcal{R}(v|K)$.
\end{definition}

\begin{remark}
    From Definition~\ref{def} it follows that $u_{v|K}$ 
    does not require integration
    if there exists $k \in K$ such that
    \begin{enumerate}
    \item $C_{k, v | K \setminus k}$ is specified by the PCBN;
    
    \item $u_{k | K \setminus k }$ and
    $u_{v | K \setminus k }$ do not require integration.
    \end{enumerate}
\end{remark}

\begin{remark}
    If a recursion to compute $u_{v|\parentsdown{v}{w}}$ is proper, then the ordering of $\parentsdown{v}{w}$ presented in Remark~\ref{remark:recursion_around_v}  must be the one induced by $\smallerbig{v}$.
\end{remark}

Computation of $u_{v|\parentsdown{v}{w}}$ does not require integration. We show next that computation of the joint density is easy when $u_{w|\parentsdown{v}{w}}$ can be computed without integration.

\medskip

\begin{lemma}
\label{lemma:necess_suff_condition_integration-free}
    The joint copula density $c_\X(\u)$ (and therefore $f_\X(\x)$) of a PCBN can be computed without integration if and only if 
    \begin{align}
        \forall v \in \V, \,
        \forall w \in \pa{v}, 
        \text{ the conditional margin }
        u_{w|\parentsdown{v}{w}}
        \text{ does not require integration.}
    \label{eq:sufficient_no_integration}
    \end{align}
\end{lemma}

\begin{proof}
Assume that Condition~\eqref{eq:sufficient_no_integration} is not satisfied. Then one of the terms $u_{w|\parentsdown{v}{w}}$ in Equation~\eqref{eq:PCBN_fac_unif} must be computed with integration, and therefore $c_\X(\u)$ needs integration.
Assume now that Condition~\eqref{eq:sufficient_no_integration} is satisfied.
Since all pair-copula appearing in~\eqref{eq:PCBN_fac_unif} are specified, and all conditional margins $u_{w|\parentsdown{v}{w}}$ can be computed without integration by Condition~\eqref{eq:sufficient_no_integration},
there only remains to prove that all terms of the form $u_{v|\parentsdown{v}{w}}$ can be computed without integration.

\medskip

Note that only nodes $v$ such that $\big| \pa{v} \big| > 0$
appear in the factorization~\eqref{eq:PCBN_fac_unif}.
Let $v \in \V$ be such a node and
let $\pa{v} := (w_1, \dots, w_n)$
be ordered according to $\smallerbig{v}$.
We show by induction on $i \in \{1, \dots, n\}$ that the conditional margins $u_{v|\parentsdown{v}{w_i}}$ do not require integration.
For $i = 1$, the statement holds since
$u_{v|\parentsdown{v}{w_1}}
= u_v$.
Now, suppose that $u_{v|\parentsdown{v}{w_{i-1}}}$ can be computed without integration.
The conditional margin $u_{v|\parentsdown{v}{w_i}}$ is
\begin{align*}
    u_{v|\parentsdown{v}{w_i}}
    = \frac{
    \partial C_{w_{i-1} , v | \parentsdown{v}{w_{i-1}}}}{
    \partial u_{w_{i-1} | \parentsdown{v}{w_{i-1}} }}
    \big(
    u_{w_{i-1}|\parentsdown{v}{w_{i-1}} },
    u_{v|\parentsdown{v}{w_{i-1}} }
    \, | \, \u_{\parentsdown{v}{w_{i-1}}} \big)
\end{align*}
where the copula
$C_{w_{i-1} , v | \parentsdown{v}{w_{i-1}}}$
is specified by the PCBN.
The conditional margin
$u_{v | \parentsdown{v}{w_{i-1}}}$
does not require integration by the induction hypothesis.
Moreover,
$u_{w_{i-1} | \parentsdown{v}{w_{i-1}}}$ does not require integration by Condition~\eqref{eq:sufficient_no_integration}.
Hence, the conditional margin
$u_{v | \parentsdown{v}{w_i}}$ does not require integration.
\end{proof}

We now establish a convenient lemma, which gives a sufficient condition for~\eqref{eq:sufficient_no_integration} not to be satisfied, i.e. for $u_{w | \parentsdown{v}{w}}$ to require integration.
The condition require existence of a node $z$ for which none of the conditional copulas $C_{v,z|K}$ are specified.
Indeed, in such case none of the recursions can be proper, as a copula of this form necessarily appears at some point at these recursions.
\begin{lemma}
\label{lemma:zwO_not_specified}
    Consider a PCBN with node set $\V$,
    and let $v \in \V$ and $w \in \pa{v}$.
    The computation of the conditional margin $u_{w | \parentsdown{v}{w}}$ requires integration if
    \begin{align}
    \label{eq:suff_condition_integration}
        \exists z \in \parentsdown{v}{w}, \,
        \forall K \subseteq \parentsdown{v}{w} \setminus \{z\}, \,
        C_{z, w | K} \text{ is not specified}.
    \end{align}
\end{lemma}
\begin{proof}
    According to Remark~\ref{remark:recursion_around_v}, we know that a copula of the form $C_{z, w|K}$ for some $K \subseteq \parentsdown{v}{w} \setminus \{z\}$
    necessarily appears at some point in any recursion $R \in \mathcal{R}(w | \parentsdown{v}{w})$ to compute $u_{w | \parentsdown{v}{w}}$.
    So, $u_{w | \parentsdown{v}{w}}$ requires integration.
%
%
%
\end{proof}

\subsection{Active cycles}
\label{sec:active_cycles}

The PCBN in \cref{fig:PCBN_integration_a} is an example of a more general structure which we call active cycle.
In this diamond-type graph,
the v-structure $v_2 \rightarrow v_4 \leftarrow v_3$ is combined with a diverging connection at node $v_1$.
Together they form the cycle $v_4 - v_2 - v_1 - v_3 - v_4$ in the corresponding undirected graph.
Such undirected cycles always lead to a problematic conditional margin.
The diverging connection at $v_1$ can be replaced by a serial connection and the problem with computing the conditional margin remains.
The general definition of this problematic structure
is given below.

\begin{definition}
\label{def:active_cycle}
    Let $\G$ be a DAG.
    Consider a node $v \in \V$ with distinct parents $w, z\in \pa{v}$ which are connected by a trail $w \har x_1 \har \cdots \har x_n \har z$, $n \geq 1$, satisfying the following conditions:
    \begin{enumerate}[label = (\roman*)]
        \item $w \har x_1 \har \cdots \har x_n \har z$ consists of only diverging or serial connections.
        \item $v \leftarrow w \har x_1 \har \cdots \har x_n \har z \rightarrow v$ contains no chords.
    \end{enumerate}
    Then the trail $v \leftarrow w \har x_1 \har \cdots \har x_n \har z \rightarrow v$ is called an \textbf{active cycle} in $\G$.
    Furthermore, $\G$ is said to contain an active cycle.
\end{definition}

The presence of an active cycle in the graph necessitates integration.
This statement is proven in \cref{thm:active_cycles}.
%
%

\begin{theorem}
    Let $(\G,\O)$ be a PCBN. 
    If $\G$ contains an active cycle, then the computation of the joint density requires integration.
    \label{thm:active_cycles}
\end{theorem}
\begin{proof}
    Consider an active cycle in $\G$ of the form
    \begin{equation*}
    v \leftarrow w \rightleftharpoons x_1 \rightleftharpoons \dots \rightleftharpoons x_n \rightleftharpoons z \rightarrow v.
    \label{eq:active_cycle}
    \end{equation*}
    Since $w$ and $z$ are both parents of $v$, we have either $w \smallerbig{v} z$ or $z \smallerbig{v} w$.
    Let us assume that
    $z \smallerbig{v} w$.
    %
    %
    We want to prove that the margin $u_{w|\parentsdown{v}{w}}$ requires integration.
    Due to \cref{lemma:zwO_not_specified}, it is sufficient to show that for any $K\subseteq \parentsdown{v}{w}\setminus \{z\}$ the copula $C_{z,w|K}$ is not specified by the PCBN.
    Consider an arbitrary $K\subseteq \parentsdown{v}{w}\setminus \{z\}$.

    \medskip
    
    Note that $w$ and $z$ are not adjacent and the trail between $w$ and $z$ is not blocked by any subset of nodes in $pa(v)$, due to the existence of trail $w \rightleftharpoons x_1 \rightleftharpoons \dots \rightleftharpoons x_n \rightleftharpoons z $ without a cord.
    Since $K\subseteq\parentsdown{v}{w} \subseteq \pa{v}$, we have that $\notdsepbig{w}{z}{K}$.
    Thus, by Remark~\ref{remark:necessary_condition_Copula_specified},
    the copula $C_{z,w|K}$ is not specified.
\end{proof}

\subsection{Interfering v-structures}
\label{sec:interfering_vs}

In \cref{fig:PCBN_interferin_vstruc}, we showed another example of PCBN that requires integration for any choice of the parental ordering.
This structure is formally defined below.  

\begin{definition}
\label{def:interfering_v_structures}
    Consider a PCBN with node set $\V$.
    Assume that there exist nodes
    $v_1, v_2, v_3, v_4, v_5 \in \V$, satisfying the following conditions:
    \begin{itemize}
        \item $v_3 \ra v_4$, $v_3 \ra v_5$
        \item $v_1 \ra v_3$, $v_1 \ra v_4$, and
        $v_1 \notrightarrow v_5$
        \item $v_2 \ra v_3$, $v_2 \ra v_5$
        and $v_2 \notrightarrow v_4$.
    \end{itemize}
    Then, the nodes $v_1,\,v_2,\,v_3,\,v_4 \text{ and }v_5$ are said to be \textbf{interfering v-structures}.
    Moreover, $\G$ is said to contain interfering v-structures.
\end{definition}

Note that one or both of the arcs $v_1\to v_2$ or $v_2\to v_1$ can be added to the DAG in Figure \ref{fig:PCBN_interferin_vstruc}
and the interfering v-structure will remain.
However, this will not be the case if at least one of the arc $v_1 \to v_5$ or $v_2 \to v_4$ is added.
Removal of any of the arcs present in DAG in \cref{fig:PCBN_interferin_vstruc} alleviates the problem of the need to integrate. 

\medskip

We will prove next that for any graph containing interfering v-structures, the computation of the joint density requires integration for any choice of $\O$.
\begin{theorem}
    Consider a PCBN with DAG $\G$.
    If $\G$ contains interfering v-structures, then the computation of the joint density requires integration.
    \label{thm:interfering_vs}
\end{theorem}
\begin{proof}
    Let $v_1,v_2,v_3,v_4,v_5\in \V$ be nodes that form one of the interfering v-structures in $\G$.
    We have eight distinct cases concerning constraints on parental orderings of nodes $v_3$, $v_4$, and $v_5$, these are:
    \begin{align*}
        &v_1 \smallerbig{v_3} v_2,\, v_1 \smallerbig{v_4} v_3 \text{ and } v_2 \smallerbig{v_5} v_3, & v_2 \smallerbig{v_3} v_1,\, v_1 \smallerbig{v_4} v_3 \text{ and } v_2 \smallerbig{v_5} v_3, \\
        &v_1 \smallerbig{v_3} v_2,\, v_1 \smallerbig{v_4} v_3 \text{ and } v_3 \smallerbig{v_5} v_2, & v_2 \smallerbig{v_3} v_1,\, v_1 \smallerbig{v_4} v_3 \text{ and } v_3 \smallerbig{v_5} v_2, \\
        &v_1 \smallerbig{v_3} v_2,\, v_3 \smallerbig{v_4} v_1 \text{ and } v_2 \smallerbig{v_5} v_3, & v_2 \smallerbig{v_3} v_1,\, v_3 \smallerbig{v_4} v_1 \text{ and } v_2 \smallerbig{v_5} v_3, \\
        &v_1 \smallerbig{v_3} v_2,\, v_3 \smallerbig{v_4} v_1 \text{ and } v_3 \smallerbig{v_5} v_2, & v_2 \smallerbig{v_3} v_1,\, v_3 \smallerbig{v_4} v_1 \text{ and } v_3 \smallerbig{v_5} v_2.
    \end{align*}
    Since all cases are analogous, we only consider the case when: $v_1 \smallerbig{v_3} v_2$, $v_1 \smallerbig{v_4} v_3$ and $v_2 \smallerbig{v_5} v_3$. Then we have
    \begin{align*}
        v_1 &\in \parentsdown{v_4}{v_3} \text{ and } v_1 \notin \parentsdown{v_5}{v_3} \subseteq \pa{v_5}, \\
        v_2 &\in \parentsdown{v_5}{v_3} \text{ and } v_2 \notin \parentsdown{v_4}{v_3} \subseteq \pa{v_4}.
    \end{align*}    

    We  apply \cref{lemma:zwO_not_specified}, with $z=v_2$, $w=v_3$ and $v=v_5$ to find that the computation of
    $u_{v_3|\parentsdown{v_5}{v_3}}$
    requires integration.
   
    Let $K \subseteq \parentsdown{v_5}{v_3} \setminus \{v_2\}$. Since there is the arc $v_2 \rightarrow v_3$, the nodes $v_2$ and $v_3$ are not d-separated given any subset of $\V\setminus\{v_2, v_3\}$. Thus, the copula
    $C_{v_2, v_3|K}$ is not the independence copula.
    
    The arc $v_2 \rightarrow v_3$ has the assigned copula
    $C_{v_2, v_3 | \parentsdown{v_3}{v_2}}$.
    Note that $v_1 \in \parentsdown{v_3}{v_2}$,
    but $v_1 \notin \pa{v_5}$.
    Therefore, $v_1 \notin K$ as $K \subseteq \pa{v_5}$.
    Hence, the copula $C_{v_2, v_3|K}$ is also not specified by an arc.
\end{proof}

\section{Restricted PCBNs}
\label{sec:restricted_PCBN}

In the previous Sections~\ref{sec:active_cycles} and~\ref{sec:interfering_vs}, we have shown that PCBNs containing active cycles and/or interfering v-structures necessitate integration.
We now announce the main result, which is that these are the only graphical structures for which integration is needed.

\begin{theorem}\label{thm:PCBN_main_theorem}
    Let $\G$ be a DAG. There exists a collection of orderings $\O$ such that the computation of the joint density of the PCBN $(\G, \O)$ does not require integration
    if and only if
    $\G$ contains no active cycles nor interfering v-structures.
\end{theorem}
\begin{proof}
The sufficiency is proven using contraposition and applying \cref{thm:active_cycles,thm:interfering_vs}.
The necessity is proven by combining \cref{thm:validity_main_algorithm,thm:existence_possible_candidates}.
\end{proof}

To prove the necessity in \cref{thm:PCBN_main_theorem}, we will demonstrate that for any graph $\G$ that does not contain active cycles or interfering v-structures, we can find a collection of orderings $\O$ such that in computation of conditional margins integration is not needed.
Therefore, we construct an algorithm which is able to find a suitable $\O$ for any restricted DAG $\G$.

\subsection{Possible candidates and algorithm}

The algorithm  follows an arbitrary well ordering $(v_1, \dots, v_n)$ of nodes in $V$.
For any node $v \in \V$, a suitable ordering $\smallerbig{v}$ is chosen sequentially.
This means that when we arrive at a node $v$ we will have already chosen the order $\smallerbig{z}$ for all nodes in $z\in \pa{v}$.
The process of finding a suitable order $\smallerbig{v}$ involves growing an ordered set $\Ovk = (o_1,\dots,o_k)$, referred to as a partial order.
This should be interpreted as $o_1 \smallerbig{v} \cdots \smallerbig{v} o_k$.

\begin{definition}[Partial order]
    For a node $v\in \V$, an ordered subset of $k$ parents of $v$ will be referred to as a \textbf{partial order} denoted by $\Ovk$.
    Thus, we have $$
    \Ovk = (o_1,\dots,o_k)
    $$
    with $k\leq \big| \pa{v} \big|$ and $o_i\in \pa{v}$ for all $i\in \{1,\dots,k \}$.
    \label{def:partial_order}
\end{definition}

An initial state is $O^0_v = \emptyset$ to which at each iteration a node from the set $\pa{v}$ is added until we have found $O^{|\pa{v}|}_v$.
A node can be added to a partial order $\Ovk$ if it satisfies certain constraints.
Specifically, we can add a $w \in \pa{v}$ such that we can compute the conditional margin $u_{w|\Ovk}$ without integration.
This motivates the definition.

\begin{definition}
    The set of \textbf{possible candidates} for a partial order $\Ovk$ is defined by
    \begin{align*}
        \PossCand{\Ovk}
        &:= \big\{ w\in \pa{v}
        \text{ such that } u_{w | \Ovk}
        \text{ can be computed without integration}
        \big\}.
    \end{align*}
\end{definition}

Therefore, we propose the following  Algorithm~\ref{alg:finding_order_general}.

\begin{algorithm}[H]
\caption{Finding a suitable $O$.}
\label{alg:finding_order_general}
\begin{algorithmic}[1]
\Require restricted DAG $\G$
\Ensure set of orderings $\O$ for which we will not require integration
\ForEach{node $v$ in $\V$ according to a well-ordering}\label{alg:well-ord}
    \State $O_{v}^0 \gets \emptyset$
    \ForEach{$k = 0, \dots, |\pa{v}| - 1$}
        \State Compute $\PossCand{\Ovk}$.

        \State Choose one element $w \in Poss.Cand(\Ovk)$

        \State Define $O_v^{k+1} := (\Ovk, w)$.
    \EndFor
    \State Set $\smallerbig{v}$ according to $O_v^{|\pa{v}|}$
\EndFor
\State \Return{$\O := \{\smallerbig{v};\, v\in\V \}$}
\end{algorithmic}
\end{algorithm}

Algorithm~\ref{alg:finding_order_general} allows finding all sets of orders $\O$ that do not result in integration. This algorithm seems very simple. The main issue here is how to find the set of possible candidates and whether it is non-empty. This will be discussed next. 

\begin{theorem}
\label{thm:validity_main_algorithm}
    Let $\G$ be a DAG containing no active cycles nor interfering v-structures.
    The joint density of any PCBN with DAG $\G$ does not requires integration if and only if its set of orders $\O$ is one of the possible outputs of Algorithm~\ref{alg:finding_order_general}.
\end{theorem}

\begin{proof}
    This is a direct consequence of Lemma~\ref{lemma:necess_suff_condition_integration-free}.
    Indeed, for any node $v$ and any $w \in \pa{v}$, there are $k = |\parentsdown{v}{w}|$ elements smaller than $w$ (with respect to $\smallerbig{v}$) in $\pa{v}$.
    Remark that $\parentsdown{v}{w} = \Ovk$.
    By definition, $u_{w | \parentsdown{v}{w}} = u_{w | \Ovk}$ does not require integration if and only if $w \in \PossCand{\Ovk}$, i.e. if and only if it is chosen by the algorithm.
\end{proof}

Before characterizing the set of possible candidates, we announce one important result. Its proof is difficult and is delayed to Section~\ref{sec:proof:existence_possible_candidates}.

\begin{theorem}
\label{thm:existence_possible_candidates}
    Let $\G$ be a DAG containing no active cycles nor interfering v-structures.
    The set $\PossCand{\Ovk}$ computed in line 4 of Algorithm~\ref{alg:finding_order_general} is never empty.
\end{theorem}

This result implies that we will never encounter a case where there is no possible candidate to be added.
Hence, the algorithm will never terminate prematurely nor get ``stuck'' and will always return a suitable set of orders $\O$.

\subsection{B-sets}\label{sec:B_sets}

The construction of B-sets is motivated by observation made in \cref{sec:interfering_vs} that the children of a node $v$ with a v-structure can constrain the order of parents of $v$.
This happens because $v$ and its children can have common parents.
In the PCBN represented in \cref{fig:PCBN_integration_b}, we have $\pa{v_3}=\{v_1,v_2\}$ and $\pa{v_4}=\{v_2,v_3\}$.
Hence $v_2$ is their common parent and it has to be put first in the parental order of $v_3$.
Similarly, for the PCBN in \cref{fig:PCBN_interferin_vstruc} the common parent of $v_3$ and $v_4$ is $v_1$, hence it should be put as first in the parental order of $v_3$.
However, at the same time the common parent of $v_3$ and $v_5$ is $v_2$, which we would need to put first in the parental order of $v_3$.       
To formalize these observations we introduce the concept of B-sets.

\begin{definition}[B-set]
    Let $\G = (\V, \E)$ be a DAG.
    For $v, w \in \V$ such that $v \rightarrow w$, we denote by
    \begin{align*}
        B(v, w) := \pa{v} \cap \pa{w}
    \end{align*}
    We say that $B(v, w)$ is
    a \textbf{B-set} of $v$.
\end{definition}
The B-sets will provide us with clear restrictions an order must abide to, so that the joint density does not require integration.
It should be noted that a node has as many B-sets as it has children.
Some of them can be empty and not all of them have to be distinct.
In the lemma below we prove that the B-sets of a node in DAG $\G$ are ordered by inclusion if and only if $\G$ does not contain interfering v-structures.

\begin{lemma}
\label{Bsets_eqcond}
    Let $\G=(\V,\E)$ be a DAG. The following two statements are equivalent:
    \begin{lemlist}[label=(\roman*)]
        \item $\G$ does not contain interfering v-structures.\label{lemma:Bsets_eqcondi}
        \item For all $v_1\in \V$ and $v_2,v_3\in \ch{v_1}$ we have $B(v_1,v_2)\subseteq B(v_1,v_3)$ or $B(v_1,v_3)\subseteq B(v_1,v_2)$\label{lemma:Bsets_eqcondii}.
    \end{lemlist}
\end{lemma}
\begin{proof}
    If $\G$ contains interfering v-structures, then \ref{lemma:Bsets_eqcondii} is violated.
    For example, in \cref{fig:PCBN_interferin_vstruc} we have $\{v_1\}=B(v_3,v_4)\, \cancel{\subseteq}\,B(v_3,v_5)=\{v_2\}$ and $B(v_3,v_5) \,\cancel{\subseteq}\, B(v_3,v_4)$.
    
    If \ref{lemma:Bsets_eqcondii} is violated, then we can find $v_4 \in B(v_1, v_2) \setminus B(v_1, v_3)$
    and $v_5 \in B(v_1, v_3) \setminus B(v_1, v_2)$.
    In this case the nodes $v_1,\dots,v_5$ are exactly interfering v-structures.
\end{proof}

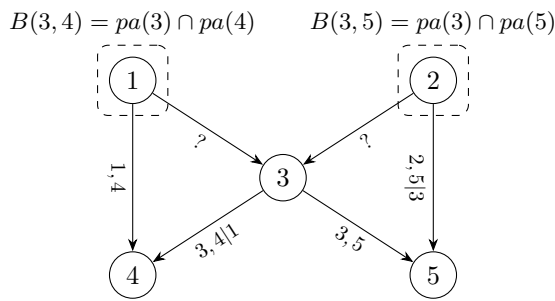
\begin{figure}[H]
  \centering
    \begin{tikzpicture}[scale=1, transform shape, node distance=1.2cm, state/.style={circle, draw=black}]

        \node[state] (3) at (3,0) {$3$};
        \node[state, above left = of 3, xshift = -2em] (1) {$1$};
        \node[state, above right = of 3, xshift = 2em] (2) {$2$};
        \node[state, below left = of 3, xshift = -2em] (4) {$4$};
        \node[state, below right = of 3, xshift = 2em] (5) {$5$};

        \begin{scope}[>={Stealth[black]}]
            \path [->] (1) edge node[scale=0.8, sloped, below] {?} (3);
            \path [->] (2) edge node[scale=0.8, sloped, below] {?} (3);
            \path [->] (1) edge node[scale=0.8, sloped, below] {$1,4$} (4);
            \path [->] (3) edge node[scale=0.8, sloped, below] {$3,4|1$} (4);
            \path [->] (3) edge node[scale=0.8, sloped, below] {$3,5$} (5);
            \path [->] (2) edge node[scale=0.8, sloped, below] {$2,5|3$} (5);
        \end{scope}

        \node (box_v1) [
            draw, rounded corners, dashed, inner sep=1ex, 
            fit=(1),
            label={[node font=\small\itshape\bfseries,rotate=0]above: $B(3,4) = \pa{3} \cap \pa{4}$}
            ]
            {};
            
        \node (box_v2) [
            draw, rounded corners, dashed, inner sep=1ex, 
            fit=(2),
            label={[node font=\small\itshape\bfseries,rotate=0]above: $B(3,5) = \pa{3} \cap \pa{5}$}
            ]
            {};
    \end{tikzpicture}
\caption{PCBN with interfering v-structures: the B-sets of $3$ are not ordered by $\subseteq$.}
\label{fig:PCBN_interferin_vstruc_B_sets}
\end{figure}

Thus, if graph $\G$ does not contain interfering v-structures, we can order the B-sets corresponding to each node $v$ according to the inclusion order $\subseteq$. The sorted sequence of these subsets with respect to the inclusion relation is referred to as the \textbf{B-sets} of $v$ and  it determines a partial order of the parental set of $v$.   

\begin{definition}[B-sets]
    Consider a DAG with no interfering v-structures and let $v \in \V$ with
    \begin{equation*}
        Q = Q(v)
        := \big| \{B(v,w); \, w\in \ch{v} \}
        \big|
    \end{equation*}
    the number of distinct B-sets corresponding to $v$.
    We denote by $B_1(v),\dots,B_Q(v)$ the sorted sequence of $\{B(v,v_2);\, v_2\in ch(v) \}$ in increasing order with respect to $\subseteq$.
    We also define $B_0(v):=\emptyset$ and $B_{Q+1}:=\pa{v}$.
    The sequence $\big( B_q(v) \big)_{q=0,\dots ,Q+1}$ is referred to as the \textbf{B-sets} of $v$.

    \medskip
    
    Furthermore, for each B-set $B_q$ with $q<Q(v)+1$, we denote by $b_q$ an arbitrary node such that $B_q=B(v,b_q)$.
    This node $b_q$ may not be unique.
    \label{def:bsets}
\end{definition}

The B-sets introduce the restriction that all nodes in $B_q$ must be smaller than nodes in $B_{q+1}\setminus B_q$ with respect to the order $\smallerbig{v}$.
We denote this by $B_q \smallerbig{v} B_{q+1}\setminus B_q$.
Hence, we must have 
$$
B_1 \smallerbig{v} B_2\setminus B_1 \smallerbig{v} B_3\setminus B_2 \smallerbig{v} \cdotslong \smallerbig{v} B_{Q+1}\setminus B_Q.
$$
We now state the following definition.

\begin{definition}[Abiding by the B-sets]
    \label{def:abide}
    Let $(\G,\O)$ be a PCBN where $\G$ contains no interfering v-structures.
    A parental order $\smallerbig{v}$ is said to \textbf{abide} by the B-sets if
    $$
    B_1
    \smallerbig{v} B_2\setminus B_1
    \smallerbig{v} B_3\setminus B_2
    \smallerbig{v} \cdotslong
    \smallerbig{v} B_{Q+1}\setminus B_Q.
    $$
    Similarly, a set of orders $\O:= \{\smallerbig{v};\, v\in \V \}$ abides by the B-sets if all its parental orders $\smallerbig{v}$ abide by the B-sets.
\end{definition}

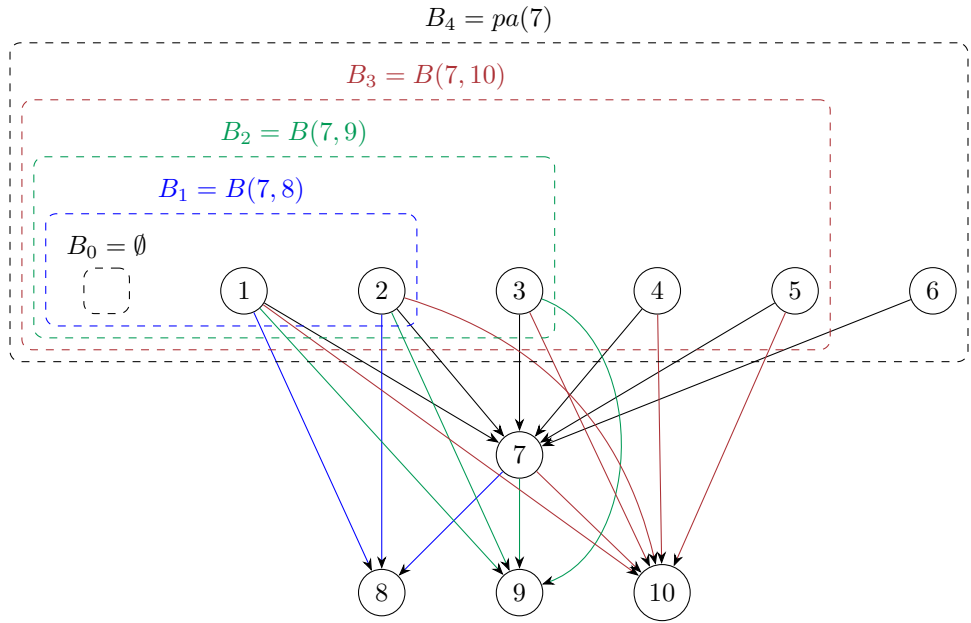
\begin{figure}[H]
  \centering
    \begin{tikzpicture}[scale=1, transform shape, node distance=1.2cm, state/.style={circle, draw=black}]

        \node[state] (1) {$1$};
        \node[state, right = of 1] (2) {$2$};
        \node[state, right = of 2] (3) {$3$};
        \node[state, right = of 3] (4) {$4$};
        \node[state, right = of 4] (5) {$5$};
        \node[state, right = of 5] (6) {$6$};
        
        \node[state, below = of 3, yshift = -1em] (7) {$7$};
        
        \node[state, below = of 7] (9) {$9$};
        \node[state, left = of 9] (8) {$8$};
        \node[state, right = of 9] (10) {$10$};
        
        \begin{scope}[>={Stealth[black]}]
            \path [->] (1) edge (7);
            \path [->] (2) edge (7);
            \path [->] (3) edge (7);
            \path [->] (4) edge (7);
            \path [->] (5) edge (7);
            \path [->] (6) edge (7);
            
            \path [->] (7) edge[color = blue] (8);
            \path [->] (7) edge[color = ForestGreen] (9);
            \path [->] (7) edge[color = Maroon] (10);
            
            \path [->] (1) edge[color = blue] (8);
            \path [->] (2) edge[color = blue] (8);
            
            \path [->] (1) edge[color = ForestGreen] (9);
            \path [->] (2) edge[color = ForestGreen] (9);
            \path [->] (3) edge[bend left=70, color = ForestGreen] (9);
            
            \path [->] (1) edge[color = Maroon] (10);
            \path [->] (2) edge[color = Maroon, bend left=30] (10);
            \path [->] (3) edge[color = Maroon] (10);
            \path [->] (4) edge[color = Maroon] (10);
            \path [->] (5) edge[color = Maroon] (10);
        \end{scope}
        
        \node (box_bset_0) [
            draw, rounded corners, dashed, inner sep=2ex, 
            left = of 1,
            label={[name = lab_bset_0]above:
            $B_0 = \emptyset$}
            ]
            {};
            
        \node (box_bset_1) [
            draw, rounded corners, dashed, inner sep=1ex, 
            color = blue,
            fit=(1) (2) (box_bset_0) (lab_bset_0),
            label={[name = lab_bset_1, color = blue]above:
            $B_1 = B(7,8)$}
            ]
            {};
            
        \node (box_bset_2) [
            draw, rounded corners, dashed, inner sep=1ex, 
            color = ForestGreen,
            fit=(1) (2) (3) (box_bset_1) (lab_bset_1),
            label={[name = lab_bset_2, color = ForestGreen]above:
            $B_2 = B(7,9)$}
            ]
            {};
            
        \node (box_bset_3) [
            draw, rounded corners, dashed, inner sep=1ex,
            color = Maroon,
            fit=(1) (2) (3) (4) (5)
            (box_bset_1) (box_bset_2) (lab_bset_2),
            label={[name = lab_bset_3, color = Maroon]above:
            $B_3 = B(7,10)$}
            ]
            {};
            
        \node (box_bset_4) [
            draw, rounded corners, dashed, inner sep=1ex,
            fit=(1) (2) (3) (4) (5) (6)
            (box_bset_1) (box_bset_2) (lab_bset_2) (box_bset_3) (lab_bset_3),
            label={above:
            $B_4 = \pa{7}$}
            ]
            {};
    \end{tikzpicture}
\caption{PCBN with ordered B-sets for the node $v = 7$.
Colors represent the different B-sets.
Any order $\smallerbig{7}$ that abides by the B-sets must start with $1$, $2$ in any order; then $3$; then $4$, $5$ in any order, and finally $6$. This means that the only possible orders $\smallerbig{7}$ that abide by the B-sets are
$(1, 2, 3, 4, 5, 6)$,
$(1, 2, 3, 5, 4, 6)$,
$(2, 1, 3, 4, 5, 6)$ and
$(2, 1, 3, 5, 4, 6)$.}
\label{fig:B_sets_ordered}
\end{figure}

Any PCBN whose set of orders does not abide by the B-sets will require integration.

\begin{lemma}\label{lemma:abide_to_bsets}
    Let $(\G,\O)$ be a PCBN where $\G$ contains no active cycles nor interfering v-structures.
    If $\O$ does not abide by the B-sets, then the computation of the joint density requires integration.
\end{lemma}
\begin{proof}
    By assumption, there is a node $v$ in $\V$ such that $\smallerbig{v}$ does not abide by the B-sets.
    Hence, there exist $s, w \in \pa{v}$ and a node $z\in \ch{v}$, $z\in \ch{s}$ and $w\notin \pa{z}$ with $w \smallerbig{v} s$. This means that $\G$ contains the subgraph below.
    $$
    \begin{tikzpicture}[scale=1, transform shape, node distance=1.2cm, state/.style={circle, font=\Large, draw=black}]

    \node (v) {$v$};
    \node[above left = of v] (s) {$s$};
    \node[above right = of v] (w) {$w$};
    \node[below left = of v] (z) {$z$};
    
    \begin{scope}[>={Stealth[length=6pt,width=4pt,inset=0pt]}]
    
    \path [->] (s) edge node[scale=0.8, sloped, below] {
    } (v);
    \path [->] (w) edge node[scale=0.8, sloped, below] {
    } (v);
    \path [->] (s) edge node {} (z);
    \path [->] (v) edge node {} (z);
    
    \end{scope}
    
    \end{tikzpicture}
    $$

    The factorization of the joint density requires the computation of the conditional margins $u_{v|\parentsdown{z}{v}}$ and $u_{s|\parentsdown{z}{s}}$.
    Both $s \smallerbig{z} v$ or $v \smallerbig{z} s$ are possible but since they are analogous, we only consider the case when $s \smallerbig{z} v$, which implies that $s \in \parentsdown{z}{v}$.
    To compute the joint density, we need the conditional margin $u_{v|\parentsdown{z}{v}}$, and we show that this margin requires integration.
    
    To achieve this \cref{lemma:zwO_not_specified} is applied.
    Let us consider an arbitrary
    $K \subseteq \parentsdown{z}{v} \setminus \{s\}$. It is sufficient to show that the copula $C_{s,v|K}$ is not specified by the PCBN.
    Note that since $w \smallerbig{v} s$ then $w \in \parentsdown{v}{s}$. However, we have $w \notin \parentsdown{z}{v}$, which means $w\notin K$.
    Therefore, by Remark~\ref{remark:cannot_remove_nodes_parent_v_w}, the copula $C_{s,v|K}$ is not specified.
    %
\end{proof}

By Lemma~\ref{lemma:abide_to_bsets} a parental order $\smallerbig{v}$ must abide to the B-sets in the sense of Definition~\ref{def:abide} to prevent integration.
Therefore, we obtain the following corollaries.

\begin{corollary}
    All parental orders determined by the algorithm abide by the B-sets.
    \label{cor:previous_orders_abide_by_bsets}
\end{corollary}


    

\begin{corollary}
    Let $x, o, w, q$ be nodes such that $x \smallerbig{o} w$
    and $w \in B(o, q)$, where the order is determined by the algorithm. 
    Then we have $x \in B(o, q)$ and in particular $x \ra q$.
    \label{cor:B_sets_implies_rightarrow}
\end{corollary}

\begin{proof}
    By Corollary~\ref{cor:previous_orders_abide_by_bsets} all previously determined parental orders chosen by the algorithm abide by the B-sets in the sense of Definition~\ref{def:abide},
    so, since $x \smallerbig{o} w$,
    we deduce that $x$ belongs to a smaller B-set than $B(o, q)$.
    Since the DAG does not contain interfering v-structures, the B-sets are ordered by inclusion.
    Therefore, we must have $x \in B(o, q)$.
\end{proof}

\subsection{Explicit characterization of the set of possible candidates}

In Algorithm~\ref{alg:finding_order_general}, if $w$ is added to the partial order, then $\parentsdown{v}{w}$ will be $\Ovk$. Moreover, the ability to compute $u_{w|\Ovk}$ without integration is a necessary and sufficient condition for the joint density to be computed without integration, by Lemma~\ref{lemma:necess_suff_condition_integration-free}.
Hence, we must start the process with a node from the smallest possible B-set.
That is, we only incorporate a node from $B_i(v)$ if all nodes from $B_{i-1}(v)$ are already included in $\Ovk$. Then, the elements of the smallest B-set larger than $\Ovk$ (denoted as $\BOvk$) are added.
\begin{definition}\label{def:smallest_bset_larger_O}
    The smallest B-set strictly larger than a partial order $\Ovk$ with $k<\big| \pa{v} \big|$ is denoted by $\BOvk$. Thus, $\BOvk:=B_{\tilde{q}}(v)$ with $$
    \tilde{q}:=\min \{q\in \{1,\dots,Q(v)+1\};\, \Ovk\subsetneq B_q(v))\}.
    $$
\end{definition}
Remark that such a $\tilde{q}$ always exists, since by definition $B_{Q(v)+1}(v)=\pa{v}$.
Furthermore, note that that $B(O^{|\pa{v}|}_v)$ is not defined, since at that moment the complete order $\smallerbig{v}$ on $\pa{v}$ has been determined already.

\begin{remark}
    In Definition~\ref{def:smallest_bset_larger_O}, the set $B(O^{|\pa{v}|}_v)$ is not defined because it is not possible to find a strictly larger B-set than $O^{|\pa{v}|}_v$.
    Furthermore, the largest B-set $B_{Q(v)+1}$ is equal to $\pa{v}$, by definition.
    Consequently, we always have $\Ovk \subsetneq \BOvk$.
\label{rem:Okv_strictly_subset_B(Okv)}
\end{remark}

It is important to note that including a node $w\in \BOvk\setminus \Ovk$ ensures that this order abides by the B-sets.
Therefore, the only allowed additions to a partial order $\Ovk$ are nodes in $\BOvk\setminus \Ovk$ for which we can compute $u_{w|\Ovk}$ without integration.
This yields the following result.

\begin{lemma}
\label{lemma:first_characterization_PossCand}
       $
        \PossCand{\Ovk}
        = \big\{ w\in \BOvk\setminus \Ovk
        \text{ such that } u_{w | \Ovk}
        \text{ can be computed without integration}
        \big\}.
    $
\end{lemma}

Lemma~\ref{lemma:first_characterization_PossCand} already narrows down the set of possible candidates by restricting the choice of $w$ from the whole set of parents of $v$ to the set $\BOvk \setminus \Ovk$.
Next we fully characterize the set of possible candidates.
It can be divided into three subsets, dependending on the local structure of node $w$ in the graph.
The most elementary case is when $w\in \BOvk$ is such that $\dsepbig{w}{\Ovk}{\emptyset}$.
Here, we have $u_{w|\Ovk} = u_w$ which obviously does not require any integration.
We will refer to these nodes as \textbf{possible candidates by independence}.
Note that this is in particular the case when $\Ovk = \emptyset$, i.e. $k = 0$.

\bigskip

Assume now that $w$ is not a possible candidate by independence.
The goal is to find conditions so $u_{w|\Ovk}$ can be computed without integration.
First the largest possible set of nodes in $\Ovk$ will be removed using conditional independence.
Let us define $J$ to be the largest set $J \subset \Ovk$
such that 
$\dsepbig{w}{J}{\Ovk \setminus J}$.
Then we have $u_{w|\Ovk} = u_{w | \Ovk \setminus J}$.
Since $w$ is not a possible candidate by independence, we have $\Ovk \setminus J \neq \emptyset$. $u_{w | \Ovk \setminus J}$ can be computed without integration if there exists a proper recursion.
By Definition~\ref{def:recursion_cond_margin}, such a recursion needs to start with a specified conditional copula
of the form $C_{w, o | \Ovk \setminus (J \sqcup \{o\})}$
where $o$ is a node in $\Ovk \setminus J$.

Since $J$ maximal then this copula cannot be specified to be the independence copula.
Therefore, there must be $w \rightarrow o$ or $o \rightarrow w$.
This means that a copula $C_{w, o | \Ovk \setminus  \{o\}}
= C_{w, o | \Ovk \setminus (J \sqcup \{o\})}$ must be equal to
\begin{itemize}
    \item $C_{w, o | \parentsdown{o}{w}}$, if $w \rightarrow o$, or,

    \item $C_{o, w | \parentsdown{w}{o}}$, if $o \rightarrow w$.
\end{itemize}
In the first case, we say that $w$ is a \textbf{possible candidate by incoming arc}. In the second case, we say that $w$ is a \textbf{possible candidate by outgoing arc}.
Such copulas have already been specified in the PCBN as they concern $w$ and $o$, which are both parents of node $v$, hence they appear earlier in the well order of nodes.
These conditions can be rewritten using d-separation, giving an explicit characterization of the set of possible candidates.
Note that $w$ is a possible candidate by incoming arc
if $w \rightarrow o$
and
\begin{align}
\label{eq:cond_incoming_arc}
    \parentsdown{o}{w} \subseteq \Ovk 
    \textnormal{ and }
    \dsepbig{w}{
    \Ovk \setminus \big(\parentsdown{o}{w} \sqcup \{o\} \big)
    }{\parentsdown{o}{w} \sqcup \{o\}}.
\end{align}
Indeed, the condition $\parentsdown{o}{w} \subseteq \Ovk$ is necessary
by Remark~\ref{remark:cannot_remove_nodes_parent_v_w}.
The second part of \eqref{eq:cond_incoming_arc} ensures that
\begin{align*}
    C_{w, o | \Ovk \setminus o}
    = C_{w, o | \parentsdown{o}{w}},
\end{align*}
precisely because the nodes that can be removed from the conditioning set due to  d-separation of $w$ given the remaining nodes $\parentsdown{o}{w} \sqcup \{o\}$ are
$\Ovk \setminus \big(\parentsdown{o}{w} \sqcup \{o\}$.

\medskip

Similary we get that $w$ is a possible candidates by outgoing arc if $o \rightarrow w$ and
\begin{align}
\label{eq:cond_outgoing_arc}
    \parentsdown{w}{o} \subseteq \Ovk 
    \textnormal{ and }
    \dsepbig{w}{\Ovk \setminus
    (\parentsdown{w}{o} \sqcup \{o\}) }{
    \parentsdown{w}{o} \sqcup \{o\} }.
\end{align}
The second part of \eqref{eq:cond_outgoing_arc} ensures that
\begin{align*}
    C_{o, w | \Ovk \setminus o}
    = C_{o, w | \parentsdown{w}{o}},
\end{align*}
as nodes in
$\Ovk \setminus (\parentsdown{w}{o} \sqcup \{o\})$
can be removed from the conditioning set due to the d-separation.

\medskip

Note that we have just proved the following result.
\begin{proposition}
\label{prop:characterization_poss_candidates}
    The set of possible candidates is composed
    \begin{equation*}
        \PossCand{\Ovk} = \PossCandInd{\Ovk} \sqcup \PossCandIn{\Ovk} \sqcup \PossCandOut{\Ovk},
    \end{equation*}
    where the three disjoint sets are defined as
    \begin{align*}
        \PossCandInd{\Ovk}
        &:= \big\{ w\in \BOvk\setminus \Ovk;\, \dsepbig{w}{\Ovk}{\emptyset}\big\},\\
        \PossCandIn{\Ovk}
        &:= \big\{w\in \BOvk\setminus \Ovk;\, \exists o\in \Ovk;\, w\ra o \textnormal{ and }
        \eqref{eq:cond_incoming_arc}
        \big\}
        ,\\
        \PossCandOut{\Ovk}
        &:= \big\{w\in \BOvk\setminus \Ovk;\, \exists o\in \Ovk;\, o\ra w \textnormal{ and }
        \eqref{eq:cond_outgoing_arc} \big\}.
    \end{align*}
\end{proposition}

This proposition gives an explicit way of finding all possible candidates by looping over all nodes in $\BOvk$ and testing whether they satisfy one of the conditions to be a possible candidates. This algorithm is implemented in the function \texttt{possible\_candidates} of the \texttt{R} package \texttt{PCBN} \cite{PCBN}.

\medskip

In the example in the next section we illustrate the process of finding the parental orders for a specific DAG.

\subsection{Example}
\label{ex:parental_orders}

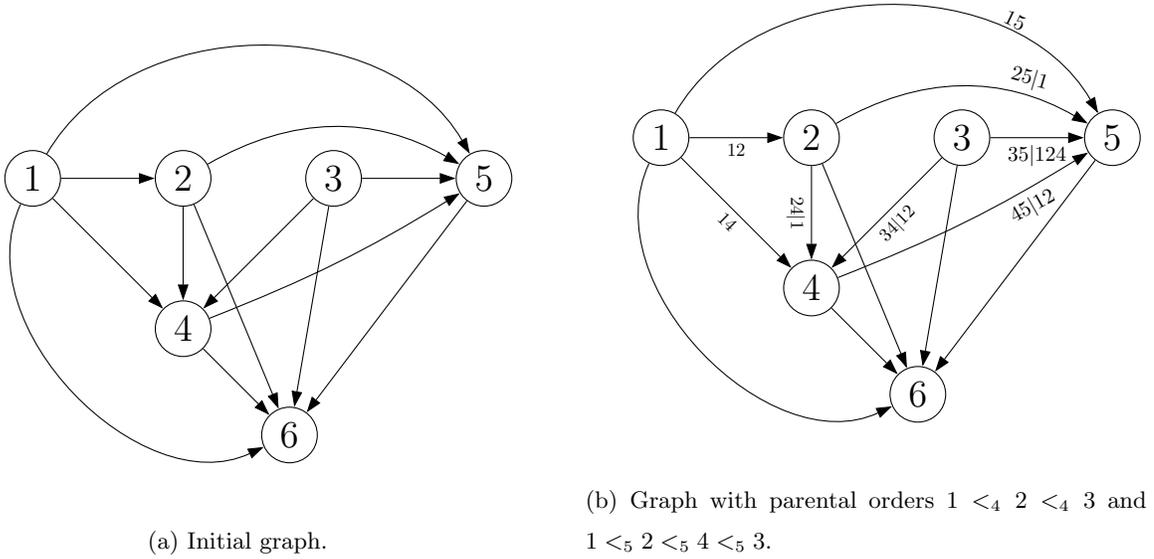
\begin{figure}[htb]
\centering
\hfill
\begin{subfigure}{.45\textwidth}
   \centering
\begin{tikzpicture}[scale=1, node distance= 2cm, transform shape, state/.style={circle, font=\Large,draw=black}]
\node[state] (1) {$1$};
\node[state, right of= 1] (2) {$2$};
\node[state, below of=2] (4) {$4$};
\node[state, right of=2] (3) {$3$};
\node[state, right of=3] (5) {$5$};
\node[state, below right of=4] (6) {$6$};

\begin{scope}[>={Stealth[length=6pt,width=4pt,inset=0pt]}]
\path [->] (1) edge node[scale=0.7, below, sloped] {} (2);
\path [->] (1) edge node[scale=0.7, below, sloped] {} (4);
\path [->] (2) edge node[scale=0.7, below, sloped] {} (4);
\path [->] (3) edge node[scale=0.7, below, sloped] {} (4);

\path [->] (1) edge[bend left=60] node[scale=0.8, sloped, above, near end] {} (5);
\path [->] (2) edge[bend left=30] node[scale=0.8, sloped, above, near end] {} (5);
\path [->] (3) edge[bend left=0] node[scale=0.8, sloped, below] {} (5);
\path [->] (4) edge[bend right=5] node[scale=0.8, sloped, below, near end] {} (5);

\path [->] (1) edge[bend right=70] node[scale=0.7, above left] {} (6);
\path [->] (2) edge[bend left=0] node[scale=0.8, above left] {} (6);
\path [->] (3) edge[bend left=0] node[scale=0.8, above left] {} (6);
\path [->] (4) edge node[scale=0.8, above left] {} (6);
\path[->] (5) edge node {} (6);
\end{scope}
\end{tikzpicture}
\caption{Initial graph.}
\end{subfigure}
 \hfill
\begin{subfigure}{.45\textwidth}
  \centering
\begin{tikzpicture}[scale=1, node distance= 2cm, transform shape, state/.style={circle, font=\Large,draw=black}]
\node[state] (1) {$1$};
\node[state, right of= 1] (2) {$2$};
\node[state, below of=2] (4) {$4$};
\node[state, right of=2] (3) {$3$};
\node[state, right of=3] (5) {$5$};
\node[state, below right of=4] (6) {$6$};

\begin{scope}[>={Stealth[length=6pt,width=4pt,inset=0pt]}]
\path [->] (1) edge node[scale=0.7, below, sloped] {$14$} (4);
\path [->] (1) edge node[scale=0.7, below, sloped] {$12$} (2);
\path [->] (2) edge node[scale=0.7, below, sloped] {$24|1$} (4);
\path [->] (3) edge node[scale=0.7, below, sloped] {$34|12$} (4);

\path [->] (1) edge[bend left=60] node[scale=0.8, sloped, above, near end] {$15$} (5);
\path [->] (2) edge[bend left=30] node[scale=0.8, sloped, above, near end] {$25|1$} (5);
\path [->] (3) edge[bend left=0] node[scale=0.8, sloped, below] {$35|124$} (5);
\path [->] (4) edge[bend right=5] node[scale=0.8, sloped, below, near end] {$45|12$} (5);

\path [->] (1) edge[bend right=70] node[scale=0.7, above left] {} (6);
\path [->] (2) edge[bend left=0] node[scale=0.8, above left] {} (6);
\path [->] (3) edge[bend left=0] node[scale=0.8, above left] {} (6);
\path [->] (4) edge node[scale=0.8, above left] {} (6);
\path[->] (5) edge node {} (6);
\end{scope}
\end{tikzpicture}
\caption{Graph with parental orders $1 <_4 2 <_4 3$ and $1 <_5 2 <_5 4 <_5 3$.}
\end{subfigure}
\caption{DAG on six nodes without active cycles and interfering v-structures.}
\label{fig:alg_ex_incoming_main}
\end{figure}

Let us consider the graph in Figure~\ref{fig:alg_ex_incoming_main} a). 
We choose to use the well-order $(1,2,3,4,5,6)$ of nodes in this graph; note that it is not unique.
The process for node $1, 2$ and $3$ is simple as $1$ and $3$ have no parents and $2$ has just one parent.
We start our consideration at node $4$.

\medskip

\noindent \textbf{1. Finding $\smallerbig{4}$.}
To find the parental order at $4$ the algorithm is initiated with  $O^0_4=\emptyset$. Since node $4$ has two children the B-sets are computed and we get that 
$B(4,5)= \{1,2,3\}$ and $B(4,6)=\{1,2,3\}$. This  means that we have just one $B_1(4)=\{1,2,3\}$. Any of these nodes can be added as first to $O^0_4$. Since $B(O^0_4)\setminus O^0_4=\{1,2,3\}\setminus \emptyset=\{1,2,3\}$ and by convention  $\dsepbig{1}{O^0_4}{\emptyset}$, $\dsepbig{2}{O^0_4}{\emptyset}$ and $\dsepbig{3}{O^0_4}{\emptyset}$ each of these nodes is a possible candidate by independence.  If we choose the node $1$, then  $O^1_4=(1)$. 

Now  $B(O^1_4)\setminus O^1_4=\{1,2,3\}\setminus\{1\}=\{2,3\}$. Since there is an arc $1\to 2$ and the copula
$C_{1, 2 | \parentsdown{4}{1}} = C_{1, 2}$
is specified then node $2$ is a possible candidate by incoming arc. Since $\dsepbig{3}{O^1_4}{\emptyset}$
then the node $3$ is a possible candidate by independence.
If we choose node $2$ and fix $O^2_4=(1, 2)$, the last choice is $O^3_4=(1, 2, 3)$,
since as before the node $3$ is a possible candidate by independence. 

\medskip

\noindent \textbf{2. Finding $\smallerbig{6}$.}
Let us now assume that we have followed the process and fixed the orders $<_4$ and $<_5$ to be:
$1 <_4 2 <_4 3$ and $1 <_5 2 <_5 4 <_5 3$. The graph together with the parental orderings for nodes $4$ and $5$ is presented in Figure~\ref{fig:alg_ex_incoming_main} b). 
Our objective now is to choose a suitable ordering $<_6$ by growing a partial order $O_6^k$.
Note that node $6$ does not have any corresponding B-sets, as it has no children.

\medskip
\noindent \textbf{2.1. Finding $O^1_6$.} Any node in $\pa{6}$ can be added to  $O^0_6=\emptyset$  by independence. 
Suppose that we choose node $4$ and get $O^1_6=(4)$. 

\medskip

\noindent \textbf{2.2. Finding $O^2_6$.}
None of the nodes in
$B(O^1_6) \setminus O^1_6 = \{1, 2, 3, 5\}$
are d-separated from $4$ by the empty set.
Therefore, we must use one of the copulas corresponding to an arc connected to $4$. The only suitable arc is the incoming arc $1\ra 4$, since its corresponding copula $C_{1, 4}$ allows computation of the margin $u_{1|O^1_6} = u_{1|4}$.

\medskip

\noindent \textbf{2.3. Finding $O^3_6$.}
Next, we consider as possible candidates nodes in the set
$B(O^2_6) \setminus O^2_6
= \{1,2,3,4,5\} \setminus \{4,1\}
= \{2,3,5\}.$
%
%
There are two incoming arcs to node $4$, i.e. $2\ra 4$ and $3\ra 4$.
Node $2$ is a possible candidate since the margin $u_{2|O^2_6}=u_{2|14}$ can be computed with the copula $C_{2,4|\parentsdown{4}{2}} = C_{2,4|1}$.

Node $3$ is not a possible candidate by incoming arc $3\ra 4$.
The copula corresponding to this arc,
$C_{3,4|\parentsdown{4}{3}} = C_{3, 4 | 1, 2}$,
contains node $2$ in its conditioning set whereas
$2 \notin O^2_6 = \{1, 4\}$.
Hence, computing $u_{3|O^2_6} = u_{3|14}$
from the copula $C_{3, 4 | 1, 2}$
requires integration with respect to node $2$:
\begin{align*}
    u_{3|14} = \int_0^1 \frac{\partial C_{3, 4 | 1, 2}
    \big(u_3, u_4 \,|\, u_1, w_2\big)}{\partial u_4} \, \diff w_2.
\end{align*}

Note that
$
\parentsdown{4}{3} = \{z \in \pa{4}; \, z <_4 3\} = \{1, 2\}
\nsubseteq \{1, 4\} = O^2_6 $ 
and this condition is necessary for a node $w$ to be considered as a possible candidate for $\Ovk$ by incoming arc $w \ra o$.
It must be that $\parentsdown{o}{w} \subseteq \Ovk$, hence that parents of node $o$ earlier in the ordering than $w$ must have already been included in  $\Ovk$.  This condition, however, is not sufficient (as shown below). Indeed, if $\Ovk$ contains nodes that are not in $\parentsdown{o}{w}$, then these elements should be ``removable'' from the conditioning set of a copula by d-separation.
Hence, another required condition is
$\dsepbig{w}{\Ovk \setminus
\big( \parentsdown{o}{w} \sqcup \{o\} \big)}{
\parentsdown{o}{w} \sqcup \{o\}}$.
 We will examine this condition more closely below.

In this example, node $2$ is the only node satisfying both of these conditions.
Indeed, for $w=2$ and $o=4$, we have $$
\dsepbig{w}{\Ovk \setminus \big( \parentsdown{o}{w} \sqcup \{o\} \big)}{\parentsdown{o}{w} \sqcup \{o\}}
= \dsepbig{2}{\{1, 4\} \setminus \{1, 4\}}{\{1, 4\}}
= \dsepbig{2}{\emptyset}{\{1, 4\}}
$$
which is satisfied by convention.
Therefore, we add it to the partial order $O^2_6$,
and obtain $O^3_6 = (4,1,2)$
and move on to the next iteration of the algorithm.

\medskip

\noindent \textbf{2.4. Finding $O^4_6$.}
Now we have two options: node $3$ by incoming arc  $3 \ra 4$ and node $5$ by the outgoing arc $4 \ra 5$.  Let us consider first the incoming arc $3\ra 4$.
The condition
$\parentsdown{4}{3}
= \{1, 2\} \subseteq O^3_6 = \{4, 1, 2\}$ is satisfied.
Indeed, we can compute the conditional margin
$u_{3|O^3_6} = u_{3 | 1, 2, 4}$ with the copula
$C_{3, 4 | \parentsdown{4}{3}} = C_{3, 4 | 1, 2}$. 
Remark that the second condition is also satisfied since
\begin{align*}
    \dsepbig{3}{O^3_6 \setminus \big( \parentsdown{4}{3} \sqcup \{4\} \big)}{\parentsdown{4}{3} \sqcup \{4\}}
    &= \dsepbig{3}{\{1, 2, 4\}\setminus \{1, 2, 4\}}{\{1, 2, 4\}} \\
    &= \dsepbig{3}{\emptyset}{\{1,2,4\}}
\end{align*}
holds by convention.
Hence, node $3$ is a possible candidate by incoming arc $3\ra 4$ at this iteration. 
However, rather than adding node $3$, we add node $5$ by the outgoing arc $4 \ra 5$. The conditions are satisfied: 
$$\parentsdown{5}{4}
= \{1,2\} \subseteq O^3_6 = \{4,1,2\} \;\;\mbox{  and }$$
\begin{align*}
    \dsepbig{5}{O^3_6 \setminus \big(\parentsdown{5}{4} \sqcup \{4\} \big)}{\parentsdown{5}{4} \sqcup \{4\}}
    &= \dsepbig{5}{\{1,2,4\}\setminus \{1,2,4\}}{\{1,2,4\}} \\
    &= \dsepbig{5}{\emptyset}{\{1,2,4\}}
\end{align*}
holds by convention.

\medskip

\noindent \textbf{2.5. Finding $O^5_6$.}
After the addition of node $5$, the only remaining node is node $3$. There are two incoming arcs $3 \ra 5$ and $3 \ra 4$.
Note that in the previous iteration, it was possible to use the incoming arc $3 \ra 4$.
We now explain why we cannot use the incoming arc $3\ra 4$ anymore,
even though the condition $\parentsdown{4}{3} = \{1, 2\}
\subseteq \{1, 2, 4, 5\} = O^4_6$ is satisfied. Note that we need the conditional margin
$u_{3|O^4_6} = u_{3|1, 2, 4, 5}$
but the arc $3 \ra 4$ corresponds to the copula
$C_{3, 4|\parentsdown{4}{3}} = C_{3, 4|1, 2}$
from which $u_{3|1, 2, 4}$ can be computed. To remove node $5$ from the conditioning set in $u_{3|1, 2, 4, 5}$ the d-separation is required
$$
\dsepbig{3}{5}{\{1,2,4\}}
= \dsepBig{3}{
O^4_6 \setminus \big(\parentsdown{4}{3} \sqcup \{4\} \big)
}{\parentsdown{4}{3} \sqcup \{4\}}
$$
which does not hold due to the arc $3\ra 5$. 
%
The node $3$ is a possible candidate by incoming arc $3 \ra 5$
as it satisfies 
$$
\dsepBig{3}{
\{1, 2, 4, 5 \} \setminus \{1, 2, 4, 5\}
}{\{1, 2, 4, 5\}}
= \dsepBig{3}{\emptyset
}{\{1, 2, 4, 5\}}.
$$
and
$$
\parentsdown{5}{3} = \{1,2,4\}
\subseteq \{1, 2, 4, 5\} = O^4_6.
$$
Therefore, we can add node $3$ by incoming arc $3\ra 5$
to obtain $O^5_6 = (4,1,2,5,3)$,
giving us the order $4 <_6 1 <_6 2 <_6 5 <_6 3.$

\section{Proof of Theorem~\ref{thm:existence_possible_candidates}}
\label{sec:proof:existence_possible_candidates}

In this section, we prove Theorem~\ref{thm:existence_possible_candidates}.
We need to prove that for every node $v$, at each step of our algorithm with current partial order $\Ovk$ with $k<\big| \pa{v} \big|$, the following property hold: $\PossCand{\Ovk} \neq \emptyset$.

%

\medskip

We will assume that the arguments of copulas assigned to arcs of the BN up to the current point of the algorithm (copulas assigned to arcs pointing to a node earlier in the well-ordering than node $v$ and copulas assigned to arcs from nodes in $\Ovk$ to $v$) do not require integration.
This means that the following copulas have been assigned by our algorithm upon the arrival at $\Ovk$:
\begin{itemize}
    \item $C_{x, y | \parentsdown{y}{x}}$ with $x \ra y\in \E$ and $y < v$.
    \item $C_{o_j, v | O^{j-1}_v}$ with $j\leq k$.
\end{itemize}



\medskip

The proof requires many additional results regarding properties of trails, B-sets, partial orders and possible candidates.
These results can be found in the Appendix.

\medskip

We will show that at any point of the algorithm, we are able to extend the current order $\Ovk$ with a node $w\in \PossCand{\Ovk}$.
Thus, we must prove that there exists a node $w\in \PossCand{\Ovk}$.

\medskip

If $\PossCandInd{\Ovk}$ is not empty, then the proof is complete.
Therefore, in the rest of the proof we assume that $\PossCandInd{\Ovk} = \emptyset$.
Consequently, we can apply Lemma~\ref{lemma:adset_not_empty} to find that $ad(O^k_{v})\cap B(O^k_{v}) \neq \emptyset$. 
Thus, there exist a $w_1\in \BOvk\setminus \Ovk$ and $o_1\in \Ovk$
such that $w_1 \ra o_1$ or $o_1 \ra w_1$.

\medskip

In what follows we will show that the existence of an arc $w_1\ra o_1$ implies that $\PossCandIn{\Ovk}$ is not empty.
Subsequently we will assume that no arc of the form 
 $w_1\ra o_1$ exists and we prove that this together with the existence of an arc $o_1 \ra w_1$ implies that $\PossCandOut{\Ovk}$ is not empty, concluding the proof.
The cases of the existence of the arcs $w_1\ra o_1$ and $o_1\ra w_1$ are considered separately.

\subsection{First case: \texorpdfstring{$w_1 \ra o_1$}{w1 -> o1}}


If $w_1$ can be added to $\Ovk$ by the incoming arc $w_1 \ra o_1$,
then $w_1 \in \PossCandIn{\Ovk}$, completing the proof.
Thus, we assume that $w_1 \notin \PossCandIn{\Ovk}$, for every $w_1$ which is a parent of $o_1$ and which belongs to $\BOvk\setminus \Ovk$.
Formally, this means
\begin{align}
\label{eq:first_case_empty_intersection}
    \pa{o_1}
    \cap \big(\BOvk\setminus \Ovk\big)
    \cap \PossCandIn{\Ovk}
    = \emptyset.
\end{align}
%
To get a contradiction our strategy is to apply the lemma below.
This lemma states that under the assumptions above, the arc $w_1\ra o_1$ implies the existence of another pair of nodes $w_2\in \BOvk\setminus \Ovk$ and $o_2\in \Ovk$ such that $w_2\ra o_2$ and $o_1\ra o_2$.
It is the case that $o_2\neq o_1$, but $w_1$ and $w_2$ may be the same node.

\begin{lemma}
\label{lemma:sequence_of_os}
    Assume that there exist $w_1 \in \BOvk \setminus \Ovk$ and $o_1 \in \Ovk$ such that $w_1 \ra o_1$.
    If \eqref{eq:first_case_empty_intersection} holds, then
    there exist
    $w_2 \in \BOvk \setminus \Ovk$ and $o_2\in \Ovk$
    such that $w_2 \ra o_2$ and $o_1 \ra o_2$.
\label{lemma:w1o1_to_w2o2}
\end{lemma}
Applying Lemma~\ref{lemma:w1o1_to_w2o2} iteratively, we obtain an infinite sequence
$
o_1 \ra o_2 \ra o_3 \ra \cdots
$
of connected nodes of $\G$.
Since the graph $\G$ is acyclic and has a finite number of nodes,
such a sequence cannot exist.
Therefore \eqref{eq:first_case_empty_intersection} cannot be true, which means that 
$\PossCandIn{\Ovk} \neq \emptyset$, completing the proof.
It remains to prove Lemma~\ref{lemma:w1o1_to_w2o2}.


\subsubsection{Proof of Lemma \ref{lemma:w1o1_to_w2o2}}

Without loss of generality, we can assume that $w_1$ is the smallest element with respect to $\smallerbig{o_1}$ in $\pa{o_1}\cap \big(\BOvk\setminus \Ovk\big)$.
From \eqref{eq:first_case_empty_intersection}, we know that $w_1 \notin \PossCandIn{\Ovk}$, hence (at least) one of the two following restrictions must be violated:
\begin{enumerate}
    \item $\parentsdown{o_1}{w_1}\subseteq \Ovk$,
    \item 
    $\dsepbig{w_1}{\Ovk \setminus \big( \parentsdown{o_1}{w_1} \sqcup \{o_1\} \big)}
    {\parentsdown{o_1}{w_1} \sqcup \{o_1\}}$.
\end{enumerate}

\medskip

The first restriction is satisfied by the lemma below.
\begin{lemma}\label{lemma:x}
    Let $w_1$ be the smallest element in $\BOvk\cap pa(o_1)$ with respect to  $\smallerbig{o_1}$. 
    Then, $\parentsdown{o_1}{w_1}\subseteq \Ovk$.
\end{lemma}
\begin{proof}[Proof of Lemma \ref{lemma:x}]
    Suppose that there exists an $x\in \parentsdown{o_1}{w_1}\setminus \Ovk$.
    That is, $x\in pa(o_1)\setminus \Ovk$ and $x \smallerbig{o_1} w_1$.
    Since $w_1$ is the smallest node in $\BOvk\cap \pa{o_1}$, we have $x \notin \BOvk$.
    
    Remark that $B(o_1, v) = \pa{o_1} \cap \pa{v}$.
    Since $w_1 \ra v$ and $w_1 \ra o_1$, we know that $w_1 \in B(o_1, v)$.
    Since $x \smallerbig{o_1} w_1$, by Corollary~\ref{cor:B_sets_implies_rightarrow}, we obtain $x \ra v$.
    This mean that $\G$ contains the subgraph below.

    $$
    \begin{tikzpicture}[scale=1, transform shape, node distance=1.2cm, state/.style={circle, font=\Large, draw=black}]

    \node (o1) {$o_1$};
    \node[above left = of o1] (w1) {$w_1$};
    \node[above right = of o1] (x) {$x$};
    \node[below = of w1, yshift = -2em] (v) {$v$};
    
    \begin{scope}[>={Stealth[length=6pt,width=4pt,inset=0pt]}]
    
    \path [->] (w1) edge (o1);
    \path [->] (x) edge (o1);
    \path [->] (w1) edge (v);
    \path [->] (o1) edge (v);
    \path [->] (x) edge[bend left=30] (v);
    
    \end{scope}
    
    \end{tikzpicture}
    $$

    Let $q$ be a node such that $\BOvk = B(v, q)$.
    Then we must have $v \ra q$ and $w_1 \ra q$.
    Since the B-sets are ordered by inclusion and $o_1 \in \Ovk \subset \BOvk = B(v, q)$, we get that $o_1 \ra q$.
    So $w_1 \in B(o_1, q)$. By Corollary~\ref{cor:B_sets_implies_rightarrow}, we obtain $x \ra q$.
    So $x \in B(o_1, q) = \BOvk$, which is a contradiction.
    %
    %
\end{proof}

\medskip

Therefore, the second restriction must be violated.
We consider two possible cases.
Assume that $\parentsup{o_1}{w_1}\cap \Ovk\neq \emptyset$.
Lemma~\ref{lemma:o2_if_higher_os} immediately implies that there exists an $o_2\in \Ovk$ as desired
(where $w_1$ is $w$, $o_1$ is $o_i$ and $o_2$ is $\ot$ in the notation of Lemma~\ref{lemma:o2_if_higher_os})
and we set $w_2 := w_1$, completing the proof in this case.

\medskip

Assume now that $\parentsup{o_1}{w_1} \cap \Ovk = \emptyset$.
We can apply Lemma~\ref{lemma:sequence_of_os_no_converging} (with $w = w_1$ and $o = o_1$ in the notation of Lemma~\ref{lemma:sequence_of_os_no_converging}) to find that there exists a trail between $w_1$ and a node
$\ot \in \Ovk \setminus \big( \parentsdown{o_1}{w_1} \sqcup \{o_1\} \big)$ which is activated by $\parentsdown{o_1}{w_1} \sqcup \{o_1\}$ containing no converging connections.
Let us pick a shortest such trail.
$$
w_1 \har x_1 \har \cdots \har x_n \har \ot.
$$
Remark that this trail is a shortest trail activated by the empty set between $w_1 \in \BOvk$ and $\ot \in \Ovk \subseteq \BOvk$ consisting of nodes in $\V \setminus \big( \parentsdown{o_1}{w_1} \sqcup \{o_1\} \big)$.
Therefore, we can combine Lemma~\ref{lemma:generalisation_to_K} (with $K = \V \setminus \big( \parentsdown{o_1}{w_1} \sqcup \{o_1\} \big)$) and Lemma~\ref{lemma:bset_btrails} to find that $\{x_1, \dots, x_n\} \subseteq \BOvk$.
In particular we obtain $x_n \in \BOvk$.

Furthermore, $x_n$ cannot be contained in $\Ovk$.
Otherwise, the trail from $w$ to $x_n$ would be an even shorter active trail from $w$ to a node in $\Ovk$.
Hence, $x_n\in \BOvk\setminus \Ovk$.

Now, the trail
$$
o_1 \la w\har x_1 \har \cdots \har x_n \har \ot
$$
is an active trail between two nodes in $\Ovk$ given the empty set consisting of nodes not in $\Ovk$. 
Thus, by Lemma~\ref{lemma:adjacent_os}, $o_1$ and $\ot$ must be adjacent.
By the assumption that $\parentsup{o_1}{w_1}\cap \Ovk=\emptyset$, we have 
$\ot\notin \parentsup{o_1}{w_1}$.
Remark that $\ot \notin \parentsdown{o_1}{w_1} \sqcup \{o_1\}$ (by definition of $\ot$) and that
$\pa{o_1} = \parentsdown{o_1}{w_1} \sqcup \parentsup{o_1}{w_1} \sqcup \{w_1\}$.
This shows that
$\ot\notin \pa{o_1}$.
Therefore, we must have $o_1\in \pa{\ot}$; this means that we have the subgraph below.
\begin{figure}[H]
    \centering
\begin{tikzpicture}[scale=1, transform shape]
\node (v1) {$o_1$};
\node[right= 1cm of v1] (w1) {$w_1$};
\node[right= 1cm of w1] (x1) {$x_1$};
\node[right= 1cm of x1] (x2) {$x_2$};
\node[right= 2cm of x2] (xn-1) {$x_{n-1}$};
\node[right= 1cm of xn-1] (xn) {$x_n$};
\node[right= 1cm of xn] (v2) {$\ot$};

\begin{scope}[>={Stealth[length=6pt,width=4pt,inset=0pt]}]
\path [->] (w1) edge node {} (v1);
\path [->] (xn) edge node {} (v2);
\path [->] (v1) edge[bend left=30] node {} (v2);

\draw[transform canvas={yshift=0.21ex},-left to,line width=0.25mm] (w1) -- (x1);
\draw[transform canvas={yshift=-0.21ex},left to-,line width=0.25mm] (w1) -- (x1);

\draw[transform canvas={yshift=0.21ex},-left to,line width=0.25mm] (x1) -- (x2);
\draw[transform canvas={yshift=-0.21ex},left to-,line width=0.25mm] (x1) -- (x2);

\draw[loosely dotted, line width=0.5mm] (x2) -- (xn-1);

\draw[transform canvas={yshift=0.21ex},-left to,line width=0.25mm] (xn-1) -- (xn);
\draw[transform canvas={yshift=-0.21ex},left to-,line width=0.25mm] (xn-1) -- (xn);
\end{scope}
\end{tikzpicture}
\end{figure}
Clearly, $\ot$ is our desired node $o_2$ and $x_n$ is our desired node $w_2$.
Indeed, we have $o_1\ra \ot$ and $x_n\ra \ot$, with $\ot\in \Ovk$ and $x_n\in \BOvk\setminus \Ovk$.
    

\medskip

\subsection{Second case: \texorpdfstring{$o_1 \ra w_1$}{o1 -> w1}}


\medskip

First, we remark that if $\E$ contains arcs of the form $w\ra o$ with $w\in \BOvk\setminus \Ovk$ and $o\in \Ovk$, then by the previous case we have that $\PossCandIn{\Ovk}\neq \emptyset$. Therefore, we can assume without loss of generality that there are no such arcs in $\E$.

\medskip

If $w_1 \in \PossCandOut{\Ovk}$, then the proof is complete.
Assume now that $w_1\notin \PossCandOut{\Ovk}$. We will show that $\PossCandOut{\Ovk}\neq \emptyset$ with the lemma below.
The lemma states that under the assumptions above, the arc $o_1\ra w_1$ implies the existence of another pair of nodes $w_2\in \BOvk\setminus \Ovk$ and $o_2\in \Ovk$ such that $w_1$ and $w_2$ are connected by a trail where all arcs point in direction of $w_1$. We can repeat this argument and construct a sequence of nodes.

\begin{lemma}
    Let $w_1\in \BOvk \setminus \Ovk$ and $o_1 \in \Ovk$ such that $o_1 \ra w_1 \in \E$.
    Assume that $w_1\notin \PossCandOut{\Ovk}$, and that
    \begin{align}
        \E \text{ does not contain any arc from }
        \BOvk
        \text{ to }
        \Ovk.
    \label{eq:no_arc_BOvk_to_Ovk}
    \end{align}
    Then, there exist $w_2\in \BOvk \setminus \Ovk$ and $o_2 \in \Ovk$ such that $o_2\ra w_2$, and $w_1$ and $w_2$ are connected by a trail of the form$$
    w_1 \la x_1 \la \cdots \la x_n \la w_2. 
    $$\label{lemma:sequence_of_ws}
\end{lemma}
Applying Lemma~\ref{lemma:sequence_of_ws} iteratively, we obtain a sequence$$
w_1\la \cdots \la w_2 \la \cdots \la w_3 \la \cdots.
$$
Since the graph $\G$ is acyclic and has a finite set of nodes, this sequence must be finite.
Let $w^*$ be the last element of the longest sequence that can be constructed starting from $w_1$.
Then, $w^*$ must belong to $\PossCandOut{\Ovk}$, otherwise, it would not be the last.
Hence, we have $\PossCandOut{\Ovk}\neq \emptyset$, completing the proof.
It remains to prove Lemma~\ref{lemma:sequence_of_ws}.

\subsubsection{Proof of Lemma~\ref{lemma:sequence_of_ws}}


Without loss of generality, we can assume that $o_1$ is the largest element in $\Ovk\cap \pa{w_1}$ with respect to $\smallerbig{w_1}$.
We consider two cases.

\medskip

\noindent
\textbf{First case: } $\parentsdown{w_1}{o_1}\setminus \Ovk\neq \emptyset$.

\medskip

Let $x_1\in \parentsdown{w_1}{o_1}\setminus \Ovk$. 
That is, $x_1\in \pa{w_1} \setminus \Ovk$ and $x_1 \smallerbig{w_1} o_1$.
This means that $G$ contains the subgraph below.
$$
    \begin{tikzpicture}[scale=1, transform shape, node distance=1.2cm, state/.style={circle, font=\Large, draw=black}]

    \node (w1) {$w_1$};
    \node[above left = of w1] (x1) {$x_1$};
    \node[above right = of w1] (o1) {$o_1$};
    \node[below = of w1] (v) {$v$};
    
    \begin{scope}[>={Stealth[length=6pt,width=4pt,inset=0pt]}]
    
    \path [->] (o1) edge (w1);
    \path [->] (x1) edge (w1);
    \path [->] (w1) edge (v);
    \path [->] (o1) edge (v);
    
    \end{scope}
    
    \end{tikzpicture}
    $$

Note that $o_1 \in \pa{w_1} \cap \pa{v} = B(w_1, v)$, with $x_1 \smallerbig{w_1} o_1$
By Corollary~\ref{cor:B_sets_implies_rightarrow}, we get $v \ra x_1$.

\medskip

Let $q$ be a node such that $\BOvk = B(v, q)$. Then we know that
$v \ra q$,
$w_1 \ra q$ and
$o_1 \ra q$.
Therefore, $x_1 \in \pa{w_1} \cap \pa{v} = B(w_1, q)$, with $x_1 \smallerbig{w_1} o_1$.
By Corollary~\ref{cor:B_sets_implies_rightarrow}, we get $q \ra x_1$.

\medskip

We have shown that $x_1 \in \pa{v}$ and $x_1 \in \pa{q}$, so $x_1 \in B(v, q) = \BOvk$.

\medskip

By the assumption that $\PossCandInd{\Ovk} = \emptyset$, we have $\notdsepbig{x_1}{\Ovk}{\emptyset}$.
Let us pick a shortest trail from $x_1$ to $\Ovk$ 
\begin{equation}
    x_1 \har x_2 \har \cdots \har x_n \har \ot
    \label{eq:shortest_trail_seq_of_os}
\end{equation}
activated by the empty set with $\ot \in \Ovk$.
Note that $x_1$ and $\ot$ are both included in the B-set $\BOvk$.
Therefore, by Lemma~\ref{lemma:bset_btrails}, we have that $x_i\in \BOvk$ for all $i=1,\dots,n$.
In particular, $x_n\in \BOvk$.
Furthermore, $x_n$ does not belong to the set $\Ovk$.
Otherwise, the trail 
$$
x_1 \har x_2 \har \cdots \har x_n
$$
would be a shorter trail from $x_1$ to $\Ovk$ activated by the empty set than \eqref{eq:shortest_trail_seq_of_os}, which is a contradiction.
Therefore, $x_n$ must be in $\BOvk \setminus \Ovk$.

By \eqref{eq:no_arc_BOvk_to_Ovk}, we know that there is no arc pointing from a node in $\BOvk\setminus \Ovk$ to a node in $\Ovk$.
This means that the arc $x_n \ra \ot$ is not possible.
Consequently, the trail \eqref{eq:shortest_trail_seq_of_os} must contain the arc $x_n \la \ot$.
Since the trail is activated by the empty set it contains no converging connections by Lemma~\ref{lemma:trail_activ_emptyset}.
Therefore, \eqref{eq:shortest_trail_seq_of_os} must be of the form 
$
w_1\la x_1 \la \cdots \la x_n \la \ot
$
with $x_n\in \BOvk\setminus \Ovk$ and $\ot \in \Ovk$. 
Hence, $x_n$ is our desired node $w_2$ and $\ot$ is our desired node $o_2$.
    

\medskip

\noindent
\textbf{Second case:} $\parentsdown{w_1}{o_1}\setminus \Ovk= \emptyset$.

\medskip

By assumption, we have $w_1 \notin \PossCandOut{\Ovk}$.
Since $\parentsdown{w_1}{o_1} \setminus \Ovk = \emptyset$,
we must have $\parentsdown{w_1}{o_1}\subseteq \Ovk$.
Therefore, using the definition of $\PossCandOut{\Ovk}$ (see Equation~\eqref{eq:cond_outgoing_arc}),
we must have
$$
    \notdsepbig{w_1}{\Ovk\setminus (\parentsdown{w_1}{o_1} \sqcup \{o_1\}) }{\parentsdown{w_1}{o_1} \sqcup \{o_1\}}.
$$
Therefore, there exists a trail between $w_1$ and a node in $\Ovk\setminus (\parentsdown{w_1}{o_1} \sqcup \{o_1\})$ activated by $\parentsdown{w_1}{o_1} \sqcup \{o_1\}$. 
By Lemma~\ref{lemma:sequence_of_ws_no_converging}, there exists such a trail containing no converging connections.
Thus, we can pick a shortest trail from $w$ to
$\Ovk \setminus \parentsdown{w_1}{o_1} \sqcup \{o_1\}$
activated by $\parentsdown{w_1}{o_1} \sqcup \{o_1\}$
containing no converging connections:
\begin{equation}
    \label{trail:sequence_of_ws_case2}
    w_1 \har x_1 \har \cdots \har x_n \har \ot
\end{equation}
with $\ot \in \Ovk \setminus \parentsdown{w_1}{o_1} \sqcup \{o_1\}$.

First, we show that \eqref{trail:sequence_of_ws_case2} must be of length $n > 1$.
If $n=0$, then we would have that $w_1 \har \ot$.
This arc must point to the left, since the arc $w_1 \ra \ot$ is an arc from a node in $\BOvk \setminus \Ovk$ to a node in $\Ovk$ which cannot be present by Condition~\eqref{eq:no_arc_BOvk_to_Ovk}. Because there is the arc $w_1 \la \ot$, we know that $\ot \in \pa{w_1}$.
Moreover, by definition the node $\ot$ does not belong to $\parentsdown{w_1}{o_1} \sqcup \{o_1\}$, and thus $\ot \in \parentsup{w_1}{o_1}$.
This means that $o_1 \smallerbig{w_1} \ot$.
Since we picked $o_1$ to be largest element in $\Ovk \cap \pa{w_1}$ according to $\smallerbig{w_1}$ then $o_1 \smallerbig{w_1} \ot$ is not possible, proving that $n>1$.

Now, we show that for all $i=1,\dots,n$, $x_i \notin \Ovk$.
Suppose that for some $i$,  we have that $x_i \in \Ovk$.
The set $\Ovk$ can be rewritten as
$\Ovk=\big( \Ovk\setminus (\parentsdown{w_1}{o_1} \sqcup \{o_1\}) \big) \sqcup (\parentsdown{w_1}{o_1} \sqcup \{o_1\})$.
Since $x_i \in \Ovk$, it must be
in $\Ovk\setminus (\parentsdown{w_1}{o_1} \sqcup \{o_1\})$
or in $\parentsdown{w_1}{o_1} \sqcup \{o_1\}$.
If $x_i \in \Ovk\setminus (\parentsdown{w_1}{o_1} \sqcup \{o_1\})$, then the $w_1 \har x_1 \har \cdots \har x_i$ would be a shorter trail between $w_1$
and $\Ovk\setminus (\parentsdown{w_1}{o_1} \sqcup \{o_1\})$
activated by $\parentsdown{w_1}{o_1} \sqcup \{o_1\}$
than trail~\eqref{trail:sequence_of_ws_case2}.
This is a contradiction, because we picked the shortest such trail.
Hence, $x_i$ must be in $\parentsdown{w_1}{o_1} \sqcup \{o_1\}$.
However, in this case the trail \eqref{trail:sequence_of_ws_case2} would be blocked by $\parentsdown{w_1}{o_1} \sqcup \{o_1\}$ which is also a contradiction.
Therefore, $x_i$ cannot be in $\Ovk$ proving the claim.

The trail \eqref{trail:sequence_of_ws_case2} is the shortest trail activated by the empty set between two nodes in $\BOvk$ ($w$ and $\ot$), and thus by Lemma~\ref{lemma:bset_btrails}, $x_i \in \BOvk$ for all $i=1,\dots,n$.
This means that for all $i=1,\dots,n$, $x_i \in \BOvk \setminus \Ovk$, in particular $x_n \in \BOvk \setminus \Ovk$.
Therefore, the arrow $x_n \ra \ot$ is forbidden by Condition~\eqref{eq:no_arc_BOvk_to_Ovk}
and so we must have  $x_n \la \ot$.
As a consequence, using the fact that \eqref{trail:sequence_of_ws_case2} has no converging connection,
we deduce that \eqref{trail:sequence_of_ws_case2} takes the form
$$
w_1 \la x_1 \la \cdots \la x_n \la \ot
$$
with $x_n\in \BOvk\setminus \Ovk$ and $\ot\in \Ovk$ and $n>0$.
So, $x_n$ is our desired node $w_2$ and $\ot$ is our desired node $o_2$.
    

\medskip

This completes the proof of Lemma \ref{lemma:w1o1_to_w2o2}. The proof of Theorem~\ref{thm:existence_possible_candidates} is completed. 

\section{Parameter estimation of PCBN models via estimating equations}
\label{sec:EstimationPCBN}

\subsection{Methodology and results}

%
%
%



In this section, we focus on the estimation of the (conditional) copulas in PCBN model. Indeed, estimation of the marginal densities can be done using classical univariate techniques (e.g. using parametric methods such as maximum likelihood or method of moments, or non-parametric methods such as kernel smoothing).

We will assume that the marginal distributions have been estimated nonparametrically via ranking. The results can be adapted in the same way for parametric margins.


 Due to the product structure of the joint copula density in PCBNs it is necessary to estimate $c_{wv|\parentsdown{v}{w}}$, for each $v \in \V$ and $w \in \pa{v}$. We focus on parametric conditional copulas which  are of the form
\begin{equation*}
    \cwvparentsbig{
    u_{w | \parentsdown{v}{w}}}{
    u_{v | \parentsdown{v}{w}}}{
    u_{\parentsdown{v}{w}}}
    = \cthetawvbig{
    u_{w | \parentsdown{v}{w}}}{
    u_{v | \parentsdown{v}{w}}}{
    u_{\parentsdown{v}{w}}}
\end{equation*}
where $u_{w | \parentsdown{v}{w}} \in [0,1],$
$u_{v | \parentsdown{v}{w}} \in [0,1],$
$u_{\parentsdown{v}{w}} \in [0,1]^{|\parentsdown{v}{w}|},$
and $\thetavecwv \in \Thetavecwv$, for a given parametric family of (conditional) copulas
$\modelcopwv
= \{c_\theta; \, \theta \in \Thetavecwv\}$.

\medskip

By definition, $c_{w v | \parentsdown{v}{w}}$ is the conditional copula of
$u_{w | \parentsdown{v}{w}}$ and
$u_{v | \parentsdown{v}{w}}$ given
$u_{\parentsdown{v}{w}}$.
From Theorem~\ref{thm:PCBN_main_theorem}, we know that a PCBN $(\G, \O)$ with neither active cycles nor interfering v-structures does not necessitate integration to compute these conditional margins.
Moreover, we have shown that to prevent integration the assignment of copulas, $\O$, must be determined by Algorithm~\ref{alg:finding_order_general}.
Therefore, we impose both restrictions on the class of PCBNs.

\begin{definition}[PCBN model]
\label{def:restricted_PCBN_model}
    Let $(\G,\O)$ be a PCBN.
    Let 
    $\modelcopwv = \big\{ \cthetawv ;\, \thetavecwv \in \Thetavecwv \big\}$ be a collection of bivariate (conditional) pair-copula densities for each arc $w \ra v \in \E$.
    Let $\Thetavec := \bigtimes_{v \in \V} \bigtimes_{w \in \pa{v}} \Thetavecwv$, and for $\thetavec = (\thetavecwv)_{v \in \V, w \in \pa{v}} \in \Thetavec$, let
    \begin{equation}
    \label{eq:model_fac}
       c_{\thetavec}(\u_{\V})
       :=  \prod_{v\in \V} \prod_{w\in \pa{v}} 
       \cthetawvbig{
        u_{w | \parentsdown{v}{w}}}{
        u_{v | \parentsdown{v}{w}}}{
        u_{\parentsdown{v}{w}}},
    \end{equation}
    where $u_{w | \parentsdown{v}{w}}$,
    $u_{v | \parentsdown{v}{w}}$,
    $u_{\parentsdown{v}{w}}$ have been computed using the recursion of h-functions with the parameter $\thetavec$.
    Then, the collection of densities
    $\modelcop = \{c_\thetavec; \, \thetavec \in \Thetavec\}$
    is a \textbf{PCBN model} on $[0,1]^{|V|}$ corresponding to the PCBN $(\G,\O)$.
    We say that this is a \textbf{restricted PCBN model} if $(\G,\O)$ does not require integration.
\end{definition}


\medskip

Let us assume that we observe $N>1$ i.i.d. observations $\D = (\X^1, \dots, \X^N)$ from the density $f_V$, whose copula is assumed to belong to $\modelcop$ with parameter $\thetavectrue
= (\thetavecwvtrue)_{v \in \V, \, w \in \pa{v}}
\in \Thetavec$.
Let $\hat\D := (\hat\U^1, \dots, \hat\U^N)$
be the dataset after applying the marginal empirical cdfs component-wise.

\medskip

The estimation problems for the PCBN model will be similar as the ones encountered for the estimation of vine copulas models \cite{haff2013parameter}.
A naive way to estimate $\thetavectrue$ is by maximizing the pseudo-log-likelihood
\begin{align*}
    \ell(\thetavec ; \, \D) 
    &= \log \left(\prod_{i=1}^n \prod_{v \in \V} \prod_{w \in \pa{v}} 
    \cthetawvbig{
    \hat U_{w | \parentsdown{v}{w}}^i}{
    \hat U_{v | \parentsdown{v}{w}}^i}{
    \hat U_{\parentsdown{v}{w}}^i}\right) \\
    &=: \sum_{v\in \V} \sum_{w\in \pa{v}}
    \ell_{w \ra v}(\thetavecwv ; \, \Dwvhat),
\end{align*}
where $\Dwvhat := (
\hat U_{w | \parentsdown{v}{w}}^i,
\hat U_{v | \parentsdown{v}{w}}^i,
\hat U_{\parentsdown{v}{w}}^i
)_{i = 1, \dots, N}$.
This is difficult since $\Dwvhat$ depends on $\thetavec$, implicitly, via the recursion of h-functions. Therefore, as in the vine copula models~\cite{haff2013parameter}, we propose to estimate $\thetavec$ using a stepwise procedure:
for each $v \in \V$ and for every $w \in \pa{v}$, we estimate the copula $\cwvparents$ using the dataset $\Dwvhat$.
This dataset can be obtained easily, without integration in an iterative way assuming that the PCBN is restricted, and the stepwise procedure is done in the `right' order, i.e. using a well-ordering on the node set $V$
and the parental orderings $\O$.

\medskip

We follow \cite[Section 3.1]{haff2013parameter} and apply the framework presented in \cite{tsukahara2005semiparametric}.
To construct convergent and asymptotically normal estimators for parameters of PCBN models we propose to use stepwise estimating equations.
For every arc $w \ra v$, let $\phiwv$ be an $\R^{\dim(\Thetavecwv)}$-valued function on $[0,1]^2 \times [0, 1]^{|\parentsdown{v}{w}|} \times \Thetavecwv$. We estimate $\thetavecwvtrue$ by the rank approximate Z-estimator $\thetavecwvhat$ defined as a solution of
\begin{align}
    \sum_{i = 1}^N \phiwv(
    \hat U_{w | \parentsdown{v}{w}}^i,
    \hat U_{v | \parentsdown{v}{w}}^i,
    \hat U_{\parentsdown{v}{w}}^i,
    \thetavecwv) = 0.
    \label{eq:estimating_equation_thetavecwvhat}
\end{align}
where $\Dwvhat$ depends on the previously estimated parameters. We get that there exists a function
$\phiwvtilde: [0, 1]^{|V|} \times \R^{\dim(\Thetavec)}
\ra \R^{\dim(\Thetavecwv)}$ such that
\begin{align*}
    \phiwvtilde(\hat\U^i, \, \thetavec)
    = \phiwv(
    \hat U_{w | \parentsdown{v}{w}}^i,
    \hat U_{v | \parentsdown{v}{w}}^i,
    \hat U_{\parentsdown{v}{w}}^i,
    \thetavecwv).
\end{align*}
Moreover, this function $\phiwvtilde$ only depends on the $\hat U_j^i$ for $j$ earlier than $v$ in the well-ordering and parental ordering up to $w$.
This holds also for the components of $\thetavec$.
Therefore, the stepwise rank approximate Z-estimator
$\thetavechat = (\thetavecwvhat)_{v \in \V, \, w \in \pa{v}}$
is the solution of the estimating equations
\begin{align*}
    \sum_{i = 1}^N \phitilde(
    \hat U_1^i, \dots, \hat U_n^i, 
    \thetavec) = 0,
\end{align*}
where $\phitilde$ is the function obtained by concatenating all the outputs of $\phiwvtilde$, for $v \in \V$ and $w \in \pa{v}$.
The estimation procedure is summarized in \cref{alg:estimation_thetavecwvhat}.

\begin{algorithm}[H]
\caption{Computation of the stepwise rank approximate Z-estimator of $\thetavec$}
\label{alg:estimation_thetavecwvhat}
\begin{algorithmic}
\Require PCBN $(\G, \O)$, PCBN model $\modelcop$,
data $\D = (\X^1, \dots, \X^N)$
\State Compute the pseudo-observations $\hat\U^1, \dots, \hat\U^N$
\ForEach{node $v$ in $\V$ according to a well-ordering}
\ForEach{node $w$ in $\pa{v}$ according the order $\smallerbig{v}$}
    \State Compute the conditional margins
    $\hat U_{w | \parentsdown{v}{w}}^i$ and
    $\hat U_{v | \parentsdown{v}{w}}^i$, for $i = 1, \dots, N$, via the recursion \State of h-functions,
    using the previously estimated copula parameters.
    \State Compute $\thetavecwvhat$ as the solution of the Estimating Equation \eqref{eq:estimating_equation_thetavecwvhat}.
\EndFor
\EndFor
\State
\Return{$\thetavechat = (\thetavecwvhat)_{v \in \V, \, w \in \pa{v}}$}
\end{algorithmic}
\end{algorithm}

The following result is a direct application of Theorem~1 in~\cite{tsukahara2005semiparametric}.

\begin{theorem}
\label{thm:asymptotic_normality}
    Under classical conditions (A1)--(A5) in \cite{tsukahara2005semiparametric} on $\phitilde$, $\thetavectrue$ and $\PP$, there exists a positive definite matrix $A$ of size $\dim(\Theta)^2$ such that
    $\sqrt{n} (\thetavechat - \thetavectrue)$ converges in distribution, as $N \to \infty$, to a multivariate normal distribution with mean $0$ and covariance matrix $A$.
\end{theorem}

Note that \cref{thm:asymptotic_normality} also holds if the PCBN is not restricted. In this case, however, computation of the pseudo-observations becomes much more computationally costly, since integration on potentially high-dimensional spaces will be required.

\medskip

\cref{thm:asymptotic_normality} shows the consistency and asymptotic normality of the stepwise pseudo-maximum likelihood estimator, using the estimating functions
\begin{align*}
    \phiwv(\hat\U^i, \, \thetavecwv )
    = \nabla_\thetavecwv \cthetawvbig{
    \hat U_{w | \parentsdown{v}{w}}^i}{
    \hat U_{v | \parentsdown{v}{w}}^i}{
    \hat U_{\parentsdown{v}{w}}^i} /
    \cthetawvbig{
    \hat U_{w | \parentsdown{v}{w}}^i}{
    \hat U_{v | \parentsdown{v}{w}}^i}{
    \hat U_{\parentsdown{v}{w}}^i}
\end{align*}
This holds under usual conditions for pseudo-maximum likelihood estimators (domination condition on the derivatives of $\phitilde$, integrability condition of $\phitilde$, identifiability of the model and existence of a nonsingular Fisher information matrix). Other estimation techniques are also included in this framework, such as estimation by inversion of Kendall's tau.

\medskip

To get practical insights about the performance of these estimation techniques, a simulation example is presented in the next section.




\subsection{Small simulation study}
\label{sec:Simulation_Estimation}

In this section, we show that Algorithm~\ref{alg:estimation_thetavecwvhat}
can accurately estimate the parameters given a data set generated from a known PCBN.
We study a particular PCBN with graphical structure and assignment of copulas as in Figure~\ref{fig:Estimation_True_PCBN} and parameters from Table~\ref{tab:Estimation_True_Copulas}. 

\begin{minipage}{0.45\textwidth}
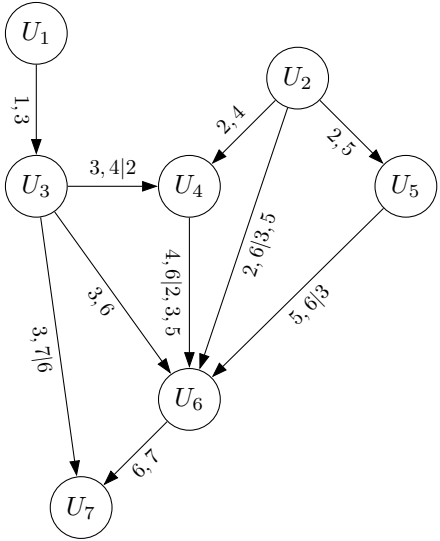
\begin{figure}[H]
    \centering
    \begin{tikzpicture}[scale=1, transform shape, node distance=1.2cm, state/.style={circle, draw=black}]

    \node[state] (X1) {$U_1$};
    \node[state, below = of X1] (X3) {$U_3$};
    \node[state, right = of X3] (X4) {$U_4$};
    \node[state, above right = of X4] (X2) {$U_2$};
    \node[state, below right = of X2] (X5) {$U_5$};
    \node[state, below = of X4, yshift = -0.8cm] (X6) {$U_6$};
    \node[state, below left = of X6] (X7) {$U_7$};
    
    \begin{scope}[>={Stealth[length=6pt,width=4pt,inset=0pt]}]
    
    \path [->] (X1) edge node[below, sloped, scale=0.8] {$1,3$} (X3);
    
    \path [->] (X2) edge node[below, sloped, above, scale=0.8] {$2,4$} (X4);
    \path [->] (X2) edge node[below, sloped, scale=0.8] {$2,5$} (X5);
    \path [->] (X2) edge node[below, sloped, scale=0.8] {$2,6|3,5$} (X6);

    \path [->] (X3) edge node[below, sloped, above, scale=0.8] {$3,4|2$} (X4);
    \path [->] (X3) edge node[below, sloped, scale=0.8] {$3,6$} (X6);
    \path [->] (X3) edge node[below, sloped,  scale=0.8] {$3,7|6$} (X7);

    \path [->] (X4) edge node[below, sloped, scale=0.8] {$4,6|2,3,5$} (X6);

    \path [->] (X5) edge node[below, sloped, scale=0.8] {$5,6|3$} (X6);

    \path [->] (X6) edge node[below, sloped, scale=0.8] {$6,7$} (X7);
    \end{scope}
    \end{tikzpicture}\hspace{5mm}
    \caption{PCBN used for the simulation studies.}
    \label{fig:Estimation_True_PCBN}
\end{figure}
\end{minipage}\hfill
\begin{minipage}{0.45\textwidth}
\begin{table}[H]
    \centering
    \begin{tabular}{cccc}
         Arc & Copula & Family & Kendall's $\tau$ \\
         \hline
        $1 \ra 3$ & $c_{1,3}$ & Gumbel & 0.6 \\
        $2 \ra 4$ & $c_{2,4}$ & Joe & 0.8 \\
        $3 \ra 4$ & $c_{3,4|2}$ & Gumbel & 0.6 \\
        $2 \ra 5$ & $c_{2,5}$ & Frank & 0.7 \\
        $3 \ra 6$ & $c_{3,6}$ & Joe & 0.9 \\
        $5 \ra 6$ & $c_{5,6|3}$ & Frank & 0.6 \\
        $2 \ra 6$ & $c_{2,6|3,5}$ & Frank & 0.85 \\
        $4 \ra 6$ & $c_{4,6|2,3,5}$ & Gumbel & 0.75 \\
        $6 \ra 7$ & $c_{6,7}$ & Gumbel & 0.65 \\
        $3 \ra 7$ & $c_{3,7|6}$ & Joe & 0.55 \\
    \end{tabular}
    \caption{The copula families and parameters of the PCBN in Figure~\ref{fig:Estimation_True_PCBN}.}
    \label{tab:Estimation_True_Copulas}
\end{table}
\hspace{1cm}
\end{minipage}

\medskip

We study the influence of the sample size $n$, with possible values 
$10 , 20 , 50 , 80 ,
100 , 200, 500, 800,
1000$ on the mean square error of the estimated parameters.
Since the copula families that are studied are different, we reparametrize them by their Kendall's tau.

\medskip

We study two possible estimation methods for the parameter of the copula: by maximum likelihood and by inversion of Kendall's tau. In both cases, we distinguish between known margins, and unknown margins, estimated non-parametrically via their rank (by the function \textit{pobs()} from the \textit{VineCopula} package \cite{RVineCopula}).
Since this article does not focus on model selection, we assume that the graph structure, orders and copula families are known. Model selection for PCBN will be treated in a future work.

\medskip

The mean-square error (MSE) of an estimator is then defined as the average squared difference between the true value of the parameter and its estimate. We do $100$ replications to estimate the MSE.
The estimation results are presented on Figure~\ref{fig:MSE_n}.

\begin{figure}[p]
    \centering
    \includegraphics[width=0.9\linewidth]{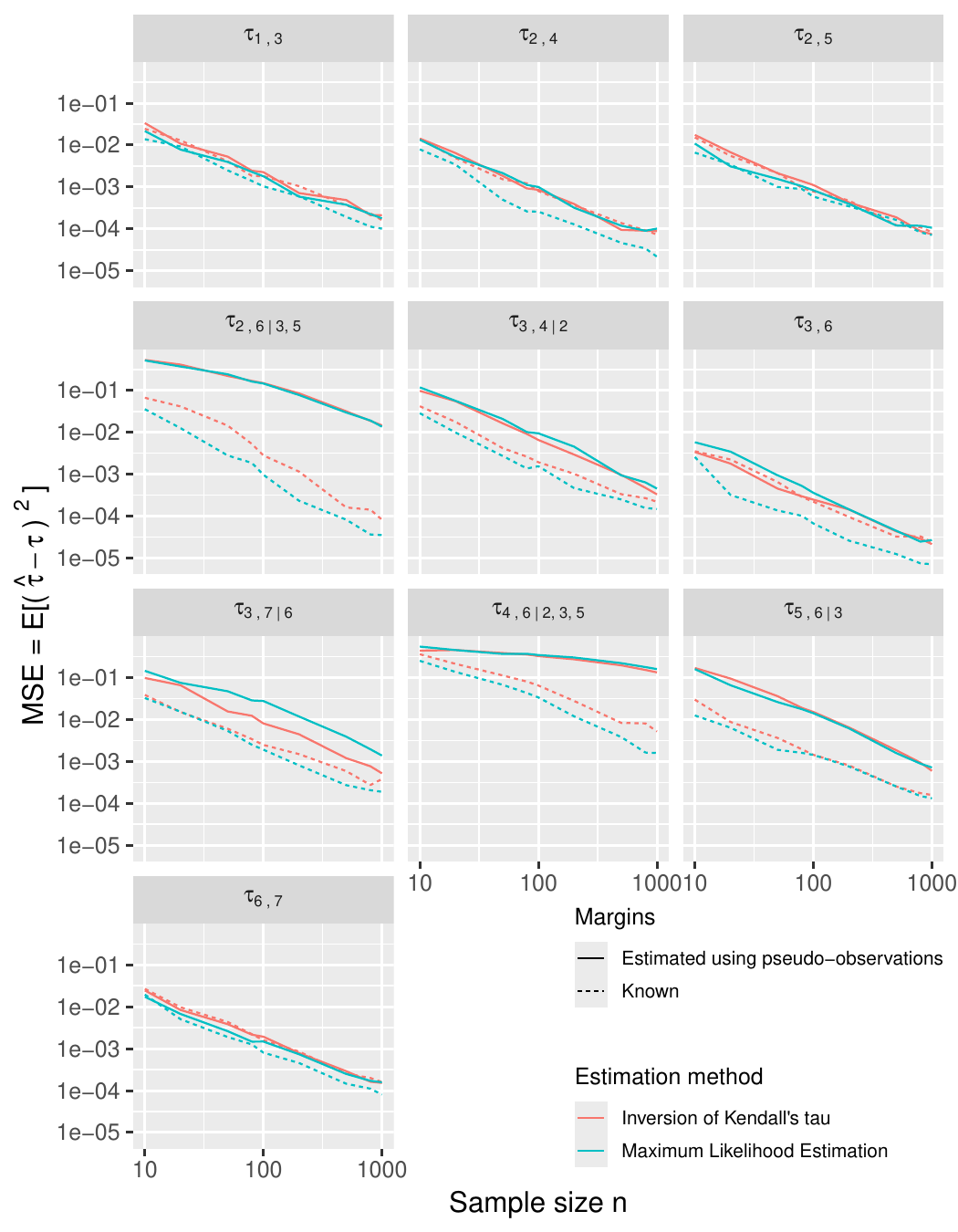}
    \caption{MSE as a function of the sample size, for different nodes of the PCBN.}
    \label{fig:MSE_n}
\end{figure}

\medskip

We can observe that our estimators are all converging (eventually) at the rate $1 / \sqrt{n}$.
But the time to reach this regime and the constant multiplicative factor can be very different.

\medskip

Estimating unconditional copulas is  relatively easy and the estimate converge fast to the true values.
However, for larger conditioning sets we can see a slower convergence of the estimates. This happens  even in the case where the simplifying assumption is made, both in the model specification and in the estimation procedure.
The observation is in line with the recent results of \cite{gauss2024asymptotics} about the asymptotics of statistical models with diverging number $p$ of parameters (meaning that $p = p_n \to \infty$).
Indeed, when the number of parameters increases (i.e., here when we are adding arcs to our graph), the sample size necessary to reach a certain accuracy needs to be larger.

\medskip

Interestingly, the knowledge of the margins seems to have a stronger influence on the MSE for the estimation of the copulas with larger conditioning sets.

\bibliographystyle{abbrv}
\bibliography{main}{}

\newpage
\appendix
\renewcommand{\thesection}{\Alph{section}}

\section{Results from \myciteintitle{derumigny2025MinimalTrails}}

In this paper, we required some results concerning properties of minimal trails in restricted DAGs, Such results are presented and proved in the paper \cite{derumigny2025MinimalTrails}. We summarize these results here to keep the paper self-contained. 

\medskip

\begin{theorem}
\label{lemma:po_active_cycle1}
    Let $\G$ be a DAG with no active cycles and let $v_1,v_2\in \V$ such that $v_1\ra v_2$.
    Suppose that 
      \begin{equation}\label{trail:subtrail2}
        v_1\la x_1 \har \cdots \har x_n \har v_2 
    \end{equation}
    is a shortest trail activated by the empty set starting with an arc $v_1 \la x_1$.
    Assume that $n \geq 1$.
    Then, for all $i \in \{1, \dots, n\}$,
    $x_i \ra x_{i+1}$ with the convention that $x_{n+1} := v_2$,
    and for all $i \in \{2, \dots, n\}$,
    $v_1 \ra x_{i}$ .
    
    This means that $\G$ contains the subgraph below.
    $$
    \begin{tikzpicture}[scale=0.8, transform shape, node distance=1.2cm, state/.style={circle, font=\Large, draw=black}]
    \node (v1) {$v_1$};
    \node[right= of v1] (v2) {$v_2$};
    \node[above left= of v1] (x2) {$x_2$};
    \node[above right= of v2] (xn-1) {$x_{n-1}$};
    \node[left = of x2] (x1) {$x_1$};
    \node[right = of xn-1] (xn) {$x_n$};

    \begin{scope}[>={Stealth[length=4pt,width=3pt,inset=0pt]}]
    \path [->] (v1) edge node {} (v2);
    \path [->] (v1) edge[bend right=0] node {} (x2);
    \path [->] (v1) edge node {} (xn-1);
    \path [->] (v1) edge node {} (xn);
    \path [->] (xn) edge node {} (v2);
    \path [->] (x1) edge node {} (v1);
    \path [->] (x1) edge node {} (x2);
    \path [->] (xn-1) edge node {} (xn);
    \draw[loosely dotted, line width=0.4mm] (x2) -- (xn-1);
    \end{scope}
    \end{tikzpicture}
    $$ 
    Furthermore, the theorem also holds for shortest trails activated by the empty set and of the form 
    \begin{equation}
    \label{trail:subtrail1}
        v_1\la x_1 \har \cdots \har x_n \ra v_2
    \end{equation}
    with $n\geq 1$.
\end{theorem}

\begin{theorem}\label{lemma:v1_trail_to_v2_parents_of_v3}
    Let $\G$ be a DAG with no active cycles
    and let $v_1,v_2,v_3 \in \V$ such that $v_1, v_2\in \pa{v_3}$. 
    Suppose that $v_1$ and $v_2$ are connected by a trail
    \begin{equation}\label{trail:trail_start}
        v_1 \har x_1 \har \cdots \har x_n \ra v_2
    \end{equation}
    activated by the empty set with $\{x_i \}_{i=1}^n \cap \pa{v_3}=\emptyset$ and $n \geq 1$.
    If this is a shortest such trail, then $\G$ contains the subgraph below, with the convention $x_0:=v_1$.
    
   $$
    \begin{tikzpicture}[scale=0.8, transform shape, node distance=1.2cm, state/.style={circle, font=\Large, draw=black}]
    \node (v3) {$v_3$};
    \node[above left = of v3] (v1) {$v_1$};
    \node[above right = of v3] (v2) {$v_2$};
    \node (help) at ($(v1)!0.5!(v2)$) {};
    \node[above = of help] (xm) {};
    \node[left  = of xm] (xm-1) {};
    \node[left = of xm-1] (x2) {$x_2$};
    \node[left = of x2] (x1) {$x_1$};
    \node[right = of xm] (xm+1) {};
    \node[right = of xm+1] (xn-1) {$x_{n-1}$};
    \node[right = of xn-1] (xn) {$x_n$};

    \begin{scope}[>={Stealth[length=4pt,width=3pt,inset=0pt]}]
        \path [->] (v1) edge node[scale=0.8, above left] {} (v3);
         \path [->] (v2) edge node[scale=0.8, above left] {} (v1);
        \path [->] (v2) edge node[scale=0.8, above left] {} (v3);
        \path [->] (v2) edge node[scale=0.8, above left] {} (x2);
        \path [->] (v2) edge node[scale=0.8, above left] {} (xn-1);
        \path [->] (xn) edge node[scale=0.8, above left] {} (v2);
        \path [->] (x2) edge node {} (x1);
        \path [->] (xn) edge node {} (xn-1);
        \path [->] (x1) edge node {} (v1);

        \draw[loosely dotted, line width=0.4mm] (x2) -- (xn-1);
    \end{scope}
    \end{tikzpicture}
    $$
\end{theorem}

\medskip

We will discuss trails between nodes, e.g. $x_0$ and $x_{n+1}$, for which all nodes on the trail are included in a certain subset $K \subseteq \V$.
In this case we say that the trail consists only of elements of $K$.
This does not include the end-points ($x_0$ and $x_{n+1}$), i.e. these end-points may or may not be in $K$.

\begin{definition}
    Let $\G=(\V, \E)$ be a DAG, let $K \subseteq \V$,
    and let 
    $x_0 \har x_1 \har \cdots \har x_{n+1}$
    be a trail.
    We say that the trail consists only of elements of $K$ if $\forall i = 1, \dots, n$, $x_i \in K$.
\end{definition}

We discuss the case when a shortest trail satisfying a certain property also satisfies a second property.
Let us first formalize what is meant by a property of a trail.

\begin{definition}[Trail property]
\label{def:trail_property}
    Let $\G$ be a DAG containing a trail $x_0 \har x_1 \har \cdots \har x_{n+1}$.
    A \textbf{property} $\prop := \prop(x_0, \dots, x_{n+1})$ specifies the existence of certain arcs between the nodes on the trail.
    Here, we mean that $\prop$ states that $\E$ contains a certain set of arcs $\{x_i \ra x_j; \, i \in I, j \in J\}$ with $I,J \subseteq \{0, 1, \dots,  n+1\}$.

    For instance, the following are regarded as trail properties:
    \begin{itemize}
        \item The first arc of the trail points to the left; $x_0 \la x_1$.
        \item The $i$-th and $j$-th node on the trail are adjacent; $x_i \har x_j$.
        \item The trail is of the form $x_0 \la x_1 \ra x_2 \ra \cdots \ra x_{n-1} \ra x_n$, and we have that $x_0 \ra x_{n-1}$. 
    \end{itemize}
\end{definition}

\begin{lemma}
    \label{lemma:generalisation_to_K}
    For a DAG $\G$ in $\mathcal{G}$,
    for a trail
    \begin{align}
        x_0 \har x_1 \har \cdots \har x_{n+1},
        \label{trail:lemma:generalisation_to_K}
    \end{align}
    let $\prop_1(x_0, x_1, \dots, x_{n+1})$ 
    and $\prop_2(x_0, x_1, \dots, x_{n+1})$
    be two properties.
    Let $\mathcal{G}$ be a set of DAGs such that
    \begin{itemize}
        \item for any DAG $\G = (\V, \E) \in \mathcal{G}$, for any $x_0, x_{n+1} \in \V$,
        and for any shortest trail \eqref{trail:lemma:generalisation_to_K} between $x_0$ and $x_{n+1}$ that satisfies $\prop_1$,
        the property $\prop_2$ holds.

        \item if $G$ belongs to $\mathcal{G}$, then any graph obtained by removing vertices from $G$ also belong to $\mathcal{G}$.
    \end{itemize}

    Let $\G = (\V, \E)$ be a DAG in $\mathcal{G}$, let $K \subseteq \V$. 
    Then for any shortest trail between $x_0$ and $x_{n+1}$ that satisfies $\prop_1$ and that consists only of elements of $K$, the property $\prop_2$ still holds.
\end{lemma}

We start with a simple lemma about trails activated by the empty set. 

\begin{lemma} \label{lemma:trail_activ_emptyset}
    A trail is activated by the empty set if and only if it does not contain a converging connection.
\end{lemma}

The lemma below states that if $\G$ contains a shortest trail $x_0 \har x_1 \har \cdots \har x_n \har x_{n+1}$ activated by the empty set for which $x_0 \ra v$ and $x_{n+1} \ra v$ for some node $v \in \V$, then for all $i=1,\dots,n$, $x_i \ra v$.

\begin{lemma}\label{lemma:ac_all_parents}
    Let $\G$ be a DAG with no active cycles 
    and let
    \begin{align}
        x_0 \har x_1 \har \cdots \har x_n \har x_{n+1}
        \label{trail:lemma:ac_all_parents}
    \end{align}
    be a trail in $\G$ for some $n \geq 0$.
    If this is a shortest trail between $x_0$ and $x_{n+1}$ activated by the empty set, then
    \begin{enumerate}[label=(\roman*)]
        \item $\ch{x_0} \cap \ch{x_{n+1}} \subseteq
        \bigcap_{i = 1}^{n} \ch{x_i}$,

        \item $\forall i = 1, \dots, n,$
        $x_i \notin \ch{x_0} \cap \ch{x_{n+1}}$.
    \end{enumerate}
\end{lemma}

In \cite{derumigny2025MinimalTrails}, the set of trails was defined as follows. 

\begin{definition}
    Let $X,Y,Z \subseteq \V$ be disjoint subsets of $\V$.
    We define $\TRAILSbig{X}{Y}{Z}$ to be the set of trails from $X$ to $Y$ activated by $Z$. 
\end{definition}

Moreover, the subtrails as well as the partial order  $\smallerTRAIL$ on such sets are considered.
\begin{definition}[Subtrails]
\label{def:subtrails}
    Let $T$ be a trail in $\TRAILSbig{X}{Y}{Z}$.
    Suppose that $T$ takes the form
    $$
    x \har t^{0}_1 \har \cdots \har t^{0}_{\n{t}{0}} \ra 
    c_1 \la t^{1}_1 \har \cdots \har t^{1}_{\n{t}{1}} \ra
    c_{2} \la \cdotslong \ra
    c_C \la t^{C}_1 \har \cdots \har t^{C}_{\n{t}{C}} \har y.
    $$
    The following are referred to as the \textbf{subtrails} of $T$:
    \begin{align*}  
        &x \har t^{0}_1 \har \cdots \har t^{0}_{\n{t}{0}} \ra c_1,\\
        &c_i \la t^{i}_1 \har \cdots \har t^{i}_{\n{t}{i}} \ra c_{i+1}, \text{ with } i \in \{1,\dots,C-1 \},\\
        &c_C \la t^{C}_1 \har \cdots \har t^{C}_{\n{t}{C}} \har y.
    \end{align*}
    The nodes on the subtrails are denoted by the symbol ``$t$'' where a superscript $i$ indicates that $t^i_j$ lies in between $c_i$ and $c_{i+1}$ with the conventions $c_0 := x$ and $c_{C+1}:=y$.
    The subscript indicates its location on the subtrail.
    The length of a subtrail is formally denoted by $\n{t}{i}$, but we will often simply write $n:=\n{t}{i}$.
    
    Furthermore, we use the conventions $c_{0} := x$, $c_{C+1} := y$, $t^i_0 := c_i$ and $t^i_{n + 1} := c_{i+1}$.
\end{definition}

The 'minimal' according to $\smallerTRAIL$ trail in $\TRAILSbig{X}{Y}{Z}$ has the following properties: 

\begin{enumerate}[label = C\arabic*.]
    \item It is a trail from $X$ to $Y$ activated by $Z$. \label{cond:active_XYZ}
    \item It contains a smaller number of converging nodes not contained in $Z$. \label{cond:least_converging_not_Z}
    \item Under the restrictions above, it contains fewer converging connections. \label{cond:least_converging}
    \item Under the restrictions above, the paths from converging nodes not contained in $Z$ to its closest descendants are shorter.\label{cond:shortest_descendant_paths}
    \item Under the restrictions above, it is a shorter such trail. \label{cond:shortest}
\end{enumerate}

We also define the notion of a closest descendant in a trail.

\begin{definition}[Closest descendant]
\label{def:closest_descendant}
    Let $T$ be a trail in $\TRAILSbig{X}{Y}{Z}$ and $i \in \{1,\dots, C(T) \}$.
    If $c_i \notin Z$, then its \textbf{closest descendant} in $Z$ is a node $Z(c_i) \in Z$ such that there exist a shortest path
    $$
    c_i \ra d^i_1 \ra \cdots \ra d^i_{\n{Z}{i}} \ra Z(c_i)
    $$
    with $d^i_j \notin Z$ for all $j = 1, \dots, \n{Z}{i}$.

    Such a path is referred to as a \textbf{descendant path} of $c_i$.
    Its nodes on the descendant path are denoted by the symbol ``$d$'' where a superscript $i$ indicates that $d^i_j$ lies on the descendant path of $c_i$, and the subscript $j$ indicates that it is the $j$-th node on this path.
    The length of the descendant path is formally denoted by $\n{Z}{i}$, but we will often simply write $n:=\n{Z}{i}$.
    If $c_i \in Z$, we also say that $c_i = Z(c_i)$.
    Finally, we use the conventions $d^i_0 := c_i$ and $d^i_{n + 1} := Z(c_i)$.
\end{definition}

The following assumptions are often used below.

\begin{assumption}
\label{assumpt:good_graphs}
    Let $\G=(\V,\E)$ be a DAG. 
    The following conditions are assumed to be satisfied:
    \begin{enumerate}
        \item $\G$ does not contain any active cycles, nor interfering v-structures.
        \item $<$ is a well-ordering corresponding to $\G$.
        \item $v$ is a node in $\V$ with $\mid \pa{v} \mid>0$.
        \item All previous orders, i.e. $<_w$ with $w<v$, have already been determined by our algorithm.
        \item $\Ovk$ is a partial order determined by our algorithm with $k< \big| \pa{v} \big|$.
    \end{enumerate}
\end{assumption}

\begin{theorem}
\label{thm:minimal_trails}
    Let $X,Y,Z \subseteq \V$ be three disjoint subsets.
    Assume that $\TRAILSbig{X}{Y}{Z} \neq \emptyset$ and
    \begin{equation}
        x \har t^{0}_1 \har \cdots \har t^{0}_{\n{t}{0}} \ra 
        c_1 \la 
        \la \cdotslong \ra
        c_C \la t^{C}_1 \har \cdots \har t^{C}_{\n{t}{C}} \har y.
    \label{trail:lemma_XYZ_Cs}
    \end{equation}
    be a minimal element of $\TRAILSbig{X}{Y}{Z} $
    with respect to the order $\smallerTRAIL$.

    Then, the following properties hold:
    \begin{enumerate}[label = (\roman*)]
        \item
        For all $i,j$, $t^{i}_j \notin X \sqcup Y \sqcup Z$
        and $d^{i}_j \notin X \sqcup Y \sqcup Z$.
        \label{prop:nodes_along_subtrails}

        \item
        For all $i=1,\dots,C$, the trails $c_i \ra d^i_1 \ra \cdots \ra d^i_n \ra Z(c_i)$ and $t^i_1 \har \cdots \har t^i_n$ do not contain a chord.
        Furthermore, the trails $x \har t^0_1 \har \cdots \har t^0_n$ and
        $t^C_1 \har \cdots \har t^C_n \har y$ do not contain a chord.
        \label{prop:no_chords}

        \item
        If $c_i \ra c_{i+1}$ and $c_{i+1} \in Z$, then $c_i \in Z$.
        \label{prop:rightarrow_in_Z}
        
        \item
        If $c_i \la c_{i+1}$ and $c_{i} \in Z$, then $c_{i+1} \in Z$.
        \label{prop:leftarrow_in_Z}

        \item
        For all $i=1,\dots,C-1$, the $i$-th subtrail is a shortest trail between $c_i$ and $c_{i+1}$ 
        starting with a leftward pointing arrow,
        ending with rightward pointing arrow,
        consisting of nodes in $\V \setminus Z$
        and with no converging connection.
        The $C$-th subtrail is a shortest trail between $c_C$ and $y$
        starting with a leftward pointing arrow,
        consisting of nodes in $\V \setminus Z$
        and with no converging connection.
        \label{prop:shortest_subtrails}
    \end{enumerate}
\end{theorem}

\begin{definition}
    \label{def:propadj}
    Let $\G$ be a DAG and $K$ a subset of $\V$.
    We say that $K$ has \textbf{local relationships} if
    for all $ v_1, v_2 \in K$ such that there exists a trail
    $$
    v_1 \har x_1 \har \cdots \har x_n \har v_2
    $$
    with $x_i \notin K$ for all $i=1, \dots, n$ and no converging connections, then $v_1$ and $v_2$ are adjacent.
\end{definition}

\begin{theorem}
\label{thm:minimal_trails_local_relationship}
    Let $X,Y,Z \subseteq \V$ be three disjoint subsets and $Y \sqcup Z$ has local relationships (\cref{def:propadj}).
    Assume that $\TRAILSbig{X}{Y}{Z} \neq \emptyset$ and let $T$ a trail of the form (\ref{trail:lemma_XYZ_Cs}) be a minimal element of $\TRAILSbig{X}{Y}{Z}$
    with respect to the order $\smallerTRAIL$. Then, the following properties hold.
    \begin{enumerate}[label = (\roman*)]
        \item
        The final converging node $c_C$ is in $Z$.
        \label{prop:final_c_in_Z}

        \item
        For all $i=1, \dots, C - 1$, we have $c_i \in Z$ or $c_{i+1} \in Z$.
        \label{prop:most_ci_in_Z}
        
        \item
        For all $i=1,\dots,C$, the nodes $c_i$ and $c_{i+1}$ are adjacent.
        \label{prop:adjacency}
        
        \item
        \label{prop:subgraphs}
        If this trail contains a total of $C>0$ converging connections, then $\G$ contains the subgraph below.
    
        \begin{figure}[H]
        \centering
        \begin{tikzpicture}[scale=1, transform shape, node distance=1cm, state/.style={circle, font=\Large, draw=black}]

        \node (c1) {$c_{1}$};
        \node[right= of c1] (c2) {$c_{2}$};
        \node[right = 1.5cm of c2] (cC-1) {$c_{C-1}$};
        \node[right= of cC-1] (cC) {$c_{C}$};
        \node[right = of cC] (y) {$y$};

        \node[left = of c1] (t0n) {$t^0_n$};
        \node[left = of t0n] (t01) {$t^0_1$};
        \node[left = of t01] (x) {$x$};

        \begin{scope}[>={Stealth[length=6pt,width=4pt,inset=0pt]}]
            \path [-] (c1) edge[bend left=45] node {} (c2);
            \path [-] (cC-1) edge[bend left=45] node {} (cC);
            \path [-] (cC) edge[bend left=45] node {} (y);
        
            \draw[transform canvas={yshift=0.21ex},-left to,line width=0.25mm] (c1) -- (c2);
            \draw[transform canvas={yshift=-0.21ex},left to-,line width=0.25mm] (c1) -- (c2);
        
            \draw[transform canvas={yshift=0.21ex},-left to,line width=0.25mm] (cC-1) -- (cC);
            \draw[transform canvas={yshift=-0.21ex},left to-,line width=0.25mm] (cC-1) -- (cC);
        
            \draw[transform canvas={yshift=0.21ex},-left to,line width=0.25mm] (cC) -- (y);
            \draw[transform canvas={yshift=-0.21ex},left to-,line width=0.25mm] (cC) -- (y);
        
            \draw[loosely dotted, line width=0.5mm] (c2) -- (cC-1);

            \path [->] (t0n) edge node {} (c1);
            \draw[loosely dotted, line width=0.5mm] (t01) -- (t0n);
            \draw[transform canvas={yshift=0.21ex},-left to,line width=0.25mm] (x) -- (t01);
            \draw[transform canvas={yshift=-0.21ex},left to-,line width=0.25mm] (x) -- (t01);
            
        \end{scope}
        \end{tikzpicture}
        \end{figure}
        Here, the curved lines represent one of the following two subgraphs.
        
        \begin{figure}[H]
        \centering
        \begin{subfigure}{0.45\linewidth}
        \centering
        \begin{tikzpicture}[scale=0.8, transform shape, node distance=1cm, state/.style={circle, font=\Large, draw=black}]
        \node (v1) {$c_i$};
        \node[right= of v1] (v2) {$c_{i+1}$};
        \node[above left= of v1] (x1) {$t^i_1$};
        \node[above right= of v2] (xn) {$t^i_n$};
        \node[above = of x1] (x2) {$t^i_2$};
        \node[above = of xn] (xn-1) {$t^i_{n-1}$};
        \begin{scope}[>={Stealth[length=4pt,width=3pt,inset=0pt]}]
        \path [->] (v1) edge node {} (v2);
        \path [->] (v1) edge[bend right=0] node {} (x2);
        \path [->] (v1) edge node {} (xn-1);
        \path [->] (v1) edge node {} (xn);
        \path [->] (xn) edge node {} (v2);
        \path [->] (x1) edge node {} (v1);
        \path [->] (x1) edge node {} (x2);
        \path [->] (xn-1) edge node {} (xn);
        \draw[loosely dotted, line width=0.4mm] (x2) -- (xn-1);
        \end{scope}
        \end{tikzpicture}
                \end{subfigure}
                \begin{subfigure}{0.45\linewidth}
                \centering
        \begin{tikzpicture}[scale=0.8, transform shape, node distance=1cm, state/.style={circle, font=\Large, draw=black}]
        \node (v1) {$c_i$};
        \node[right= of v1] (v2) {$c_{i+1}$};
        \node[above left= of v1] (x1) {$t^i_1$};
        \node[above right= of v2] (xn) {$t^i_n$};
        \node[above = of x1] (x2) {$t^i_2$};
        \node[above = of xn] (xn-1) {$t^i_{n-1}$};
        \begin{scope}[>={Stealth[length=4pt,width=3pt,inset=0pt]}]
        \path [->] (v2) edge node {} (v1);
        \path [->] (v2) edge[bend right=0] node {} (x2);
        \path [->] (v2) edge node {} (xn-1);
        \path [->] (v2) edge node {} (x1);
        \path [->] (xn) edge node {} (v2);
        \path [->] (x1) edge node {} (v1);
        \path [->] (x2) edge node {} (x1);
        \path [->] (xn) edge node {} (xn-1);
        \draw[loosely dotted, line width=0.4mm] (x2) -- (xn-1);
        \end{scope}
        \end{tikzpicture}
        \end{subfigure}
        \end{figure}
    \end{enumerate}
\end{theorem}

\begin{corollary}
\label{cor:all_cs_in_Z}
    Let us consider the setting of \cref{thm:minimal_trails_local_relationship}.
    \begin{enumerate}[label = (\roman*)]
        \item If the trail
        $c_1 \har \cdots \har  c_C$
        takes the form
        $c_1 \ra \cdots \ra  c_C$,
        then $\forall i=1,\dots,C$, $c_i \in Z$.
        \label{prop:rightarroww_all_in_Z}
        
        \item If $c_1 \in Z$ and the trail
        $c_1 \har \cdots \har c_C$
        takes the form
        $c_1 \la \cdots \la  c_C$,
        then $\forall i=1,\dots,C$, $c_i \in Z$.
        \label{prop:leftarrow_all_in_Z}
        
        \item Let $i \in \{2, \dots, C-1\}$.
        If the trail $c_{i-1} \har c_i \har c_{i+1}$ is not a converging connection, then $c_i \in Z$.
        \label{prop:no_convCon_c_i_in_Z}
    \end{enumerate}
\end{corollary}

\section{Properties of B-sets and possible candidates}
\label{sec:properties_of_bsets_and_posscand}



Informally, the lemma below states that two nodes in a B-set $B_q$ are either d-separated given the empty set or any shortest trail activated by the empty set between them must be contained in $B_q$.

\begin{lemma}\label{lemma:bset_btrails}
    Under Assumption~\ref{assumpt:good_graphs},  let $v\in \V$.
    Let $q \in \{1, \dots, Q(v) + 1\}$ and let $w_1, w_2 \in B_q(v)$.
    Then, $\dsepbig{w_1}{w_2}{\emptyset}$ or any shortest trail activated by the empty set joining $w_1$ and $w_2$ must consist entirely of nodes contained in $B_q$.
\end{lemma}

\begin{proof}
    If $\dsepbig{w_1}{w_2}{\emptyset}$, or if $w_1 = w_2$, then the proof of this lemma is completed. Therefore we can assume that they are not d-separated by the empty set and different from each other.
    Thus $\notdsepbig{w_1}{w_2}{\emptyset}$; let 
    \begin{align}
        w_1 \har x_1 \har \cdots \har x_n \har w_2
        \label{trail:lemma:bset_btrails}
    \end{align}
    be a shortest trail between $w_1$ and $w_2$  activated by the empty set.
    First, we assume that $q \leq Q(v)$.
    Let $x_0 := w_1$, $x_{n+1} := w_2$, and $b_q$ be a node corresponding to $B_q(v)$, see Definition~\ref{def:bsets}.
    Because $w_1, w_2 \in B_q(v)$, we know that
    $v, b_q \in \ch{w_1} \cap \ch{w_2}$.
    By Lemma~\ref{lemma:ac_all_parents} 
    for all $i = 1, \dots, n$, we have
    $v, b_q \in \ch{x_i}$ and
    $x_i \in B_q(v) = \pa{v} \cap \pa{b_q}$.

    If $q = Q(v) + 1$, we are at the last stage of the algorithm and there is no $b_q$, but the same reasoning shows that for $i = 1, \dots, n$,
    $x_i \in B_q(v) = \pa{v}$.
    This concludes the proof.
\end{proof}

A useful lemma proven in \cite{derumigny2025MinimalTrails} (included without the proof in Lemma~\ref{lemma:generalisation_to_K}) shows that one property of a trail that implies another will not only hold for shortest trails but also for shortest trails consisting of nodes in a subset $K \subseteq \V$. However, this result cannot be directly applied to Lemma~\ref{lemma:bset_btrails}, because the property that for all $i=1,\dots,n$, $x_i \in B_q(v)$ concerns a node $v$ which is not on the trail.
Therefore, we prove the generalization of Lemma~\ref{lemma:generalisation_to_K}  in the corollary hereunder.

\begin{corollary}\label{cor:bset_trails}
    Let $\G$ be a DAG with no active cycles nor interfering v-structures, and let $v\in \V$.
    Let $q \in \{1, \dots, Q(v) + 1\}$ and let $w_1, w_2 \in B_q(v)$. Let $K$ be a set included in $\V$.
    Then, $w_1$ and $w_2$ are either independent or for any shortest trail activated by the empty set joining $w_1$ and $w_2$ consisting of nodes in $K$ must consist entirely of nodes contained in $B_q$.
\end{corollary}
\begin{proof}
    First, we assume that $q \leq B_q(v)$.
    Let $b_q$ be a node corresponding to $B_q$.
    Let $\G^*=(\V^*,\E^*)$ be the subgraph induced by the nodes in
    $\V^*:=\{v,w_1,w_2,b_q\}\cup K$.
    Note that $v$ and $b_q$ are children of both $w_1$ and $w_2$ in $\G^*$.
    Therefore, by Lemma~\ref{lemma:ac_all_parents}(ii), any shortest trail between $w_1$ and $w_2$ in $\G^*$ activated by the empty set must not contain $v$ nor $b_q$.
    This means that any shortest trail between $w_1$ and $w_2$ in $\G^*$ activated by the empty set must consist only of elements of $K$.
    
    Consider a shortest trail in $\G$
    \begin{equation}\label{eq:cor_bset_trail}
        w_1 \har x_1 \har \cdots \har x_n \har w_2
    \end{equation}
    consisting of nodes in $K$, i.e. $\{ x_i\}_{i=1}^n\subseteq K$.
    Therefore, it is a shortest trail activated by the empty set between $w_1$ and $w_2$ in $\G^*$.
    We now apply Lemma~\ref{lemma:bset_btrails}, since $w_1$ and $w_2$ belong to the B-set $B_q \cap \V^*$ corresponding to $v$ in the graph $\G^*$.
    Therefore, for all $i=1, \dots, n$,
    $x_i \in B_q$, completing the proof.

    If $q = Q(v) + 1$, then the proof is analogous to the previous case, but then with $\V^*:= (v,w_1,w_2)\cup K$.
\end{proof}

Informally, the lemma below states that if the set $\PossCandInd{\Ovk}$ is empty, then there is a node in $\BOvk \setminus \Ovk$ which is adjacent to a node in the set $\Ovk$.
\begin{lemma}\label{lemma:adset_not_empty}
    Under Assumption~\ref{assumpt:good_graphs}, let $w\in B(\Ovk)\setminus \Ovk$ such that $\notdsepbig{w}{\Ovk}{\emptyset}$,
    that is, $w \notin \PossCandInd{\Ovk}$.
    Then, $\ad{\Ovk}\cap B(\Ovk) \neq \emptyset$, where $\ad{\Ovk}$ is the set of nodes  adjacent to an element of $\Ovk$.
\end{lemma}
\begin{proof}
By assumption, we have
$\notdsepbig{w}{\Ovk}{\emptyset}$.
Therefore $w$ must be connected to $\Ovk$ by some trail activated by the empty set.
We pick a shortest trail from $w$ to $\Ovk$ activated by the empty set, as
\begin{align}
    w \har x_1 \har \cdots \har x_n \har o
    \label{trail:lemma:adset_not_empty}
\end{align}
where $o\in \Ovk$.

If $n = 0$, then $w\in \BOvk \setminus \Ovk$ is adjacent to $o$ and thus $w\in \ad{\Ovk}\cap B(\Ovk)$.

Now, assume that $n>0$. 
We will prove that $x_n\in \ad{\Ovk}\cap B(\Ovk)$.
Since we have a shortest trail between two nodes ($w$ and $o$) in $\BOvk$ with no chords, Lemma~\ref{lemma:bset_btrails} implies that all $x_i \in \BOvk$.
As a particular case, we have $x_n \in \BOvk$. 

If $x_n\in \Ovk$, then the trail
$w \har x_1 \har \cdots \har x_n$
would be a shorter trail from $w$ to $\Ovk$
than the trail in \eqref{trail:lemma:adset_not_empty}.
This is a contradiction, proving that $x_n\notin \Ovk$.
Therefore $x_n \in \BOvk\setminus \Ovk$.
Note that $x_n$ is adjacent to $o$, and thus $x_n \in \ad{\Ovk}\cap B(\Ovk)$.
This concludes the proof.
\end{proof}

\noindent By \cref{prop:characterization_poss_candidates}, we know that a node $w$ is not a possible candidate for partial order $\Ovk$, if $\notdsepbig{w}{\Ovk}{\emptyset}$ ($w \notin \PossCandInd{\Ovk}$) and $w$ and $\Ovk$ are not adjacent ($w \notin \PossCandIn{\Ovk} \sqcup \PossCandOut{\Ovk}$).
We will now prove an even stronger claim.
That is, a node $w$ is not a possible candidate to be added to a partial order $\Ovk$, if there exists an $o$ in $\Ovk$ such that:
\begin{itemize}
    \item $w$ and $o$ are not adjacent.
    \item There exists a trail between $w$ and $o$ activated by the empty set which does not contain any nodes in $\Ovk$.
\end{itemize}
The lemma below provides a clear intuition into how the algorithm grows a partial order.
For example, consider the a trail
$$
o \har w_1 \har w_2 \har \cdots \har w_n,
$$
with no converging connections where $o \in \Ovk$ and $\{w_i \}_{i=1}^n \subseteq \pa{v} \setminus \Ovk$.
In this case, we cannot add the node $w_i$ to $\Ovk$ for any $i \in \{2,\dots,n \}$ since it is connected to $o$ by an active trail $o \har w_1 \har \cdots \har w_i$ consisting of nodes in $\V \setminus \Ovk$.
Consequently, we must add node $w_1$ before adding node $w_i$.
If $w_1$ is added to $\Ovk$, then the same argument applies to the trail $w_1 \har w_2 \har \cdots \har w_n$, i.e. we must add $w_2$ next.
The recursion is clear; any node $w_i$ can only be added after $w_1,\dots,w_{i-1}$ have been added.
So, the algorithm ``walks'' over trails with no converging connections, adding elements of these trails one node at a time, and it is only allowed to make ``jumps'' whenever a node is d-separated by the empty set from the current partial order.

\begin{lemma}\label{lemma:adjacent_os}
     Under Assumption~\ref{assumpt:good_graphs}, let $o \in \Ovk$ and $w \in \PossCand{\Ovk}$.
    If there exists a trail
    \begin{align}\label{trail:lemma:adjacent_os}
        o \har x_1 \har \cdots \har x_n \har w,
    \end{align}
    with no converging connection such that $\forall i = 1, \dots, n$,
    $x_i \in \V \setminus \Ovk$,
    then $w$ and $o$ are adjacent.
\end{lemma}

\begin{proof}
We will employ an inductive argument, assuming that the lemma holds for all previous partial orders determined by the algorithm. 
By ``previous partial orders'' we mean all partial orders $O^p_w$ with $w<v$ and
$p \in \{1,\dots, | \pa{w} |-1 \}$,
and $O^p_v$ with $p\in \{1,\dots,k-1 \}$.

Without loss of generality we can assume that the trail \eqref{trail:lemma:adjacent_os} is a shortest trail between $o$ and $w$ with no converging connection and satisfying $\forall i = 1, \dots, n$,
$x_i \in \V \setminus \Ovk$.
Because $o$ and $w$ are parents of $v$ by construction, $\G$ contains the subgraph below.
\begin{figure}[H]
    \centering
\begin{tikzpicture}[scale=1, transform shape, node distance=1.2cm, state/.style={circle, font=\Large, draw=black}]

\node (oj) {$o$};
\node[right= of oj] (x1) {$x_1$};
\node[below right= of x1] (ot) {$v$};
\node[above right= of ot] (xn) {$x_n$};
\node[right= of xn] (oi) {$w$};

\begin{scope}[>={Stealth[length=6pt,width=4pt,inset=0pt]}]

\path [->] (oj) edge node {} (ot);
\path [->] (oi) edge node {} (ot);

\draw[loosely dotted, line width=0.5mm] (x1) -- (xn);
\draw[transform canvas={yshift=0.21ex},-left to,line width=0.25mm] (oj) -- (x1);
\draw[transform canvas={yshift=-0.21ex},left to-,line width=0.25mm] (oj) -- (x1);

\draw[transform canvas={yshift=0.21ex},-left to,line width=0.25mm] (xn) -- (oi);
\draw[transform canvas={yshift=-0.21ex},left to-,line width=0.25mm] (xn) -- (oi);

\end{scope}

\end{tikzpicture}
\end{figure}

Because \eqref{trail:lemma:adjacent_os} has no converging connection, we know that $\notdsepbig{w}{\Ovk}{\emptyset}$.
Thus, if $w\in \PossCand{\Ovk}$, then we must have $w\in \PossCandIn{\Ovk}$ or $w\in \PossCandOut{\Ovk}$.

In the base case where $k = 1$ and $\big| \Ovk \big| = 1$, we know that $\Ovk = \{o\}$. Therefore we directly know that $w$ and $o$ are adjacent (because $w\in \PossCandIn{\Ovk}$ or $w\in \PossCandOut{\Ovk}$, so $w$ must be connected to some node in $\Ovk$, and this must be $o$).

We now prove the induction step.
If $n = 0$, then $w$ and $o$ are adjacent, which concludes the proof. We now assume $n > 0$.
For this, we consider both cases depending on whether $w\in \PossCandIn{\Ovk}$ or $w\in \PossCandOut{\Ovk}$.

\underline{Case 1: $w\in \PossCandIn{\Ovk}$}. 
By definition of $\PossCandIn{\Ovk}$ (\cref{prop:characterization_poss_candidates}),
there exists an $\ot\in \Ovk$ such that $w\ra \ot$ satisfying the following restrictions:
\begin{enumerate}
    \item $\parentsdown{\ot}{w}\subseteq \Ovk$.
    \item $\dsepbig{w}{\Ovk\setminus ( \parentsdown{\ot}{w} \sqcup \{\ot\} )}{\parentsdown{\ot}{w} \sqcup \{\ot\}}$.
\end{enumerate}
First, note that if $\ot= o$, then $o$ and $w$ are adjacent, completing the proof.
Thus, we assume that $\ot\neq o$.

To satisfy the second restriction above,
any trail between $w$
and $\Ovk\setminus ( \parentsdown{\ot}{w} \sqcup \{\ot\} )$
must be blocked by $\parentsdown{\ot}{w} \sqcup \{\ot\}$.
If we assume that $o\in \Ovk\setminus ( \parentsdown{\ot}{w} \sqcup \{\ot\} )$, then the trail~\eqref{trail:lemma:adjacent_os} must be blocked by $\parentsdown{\ot}{w} \sqcup \{\ot\}$.
Since this trail~\eqref{trail:lemma:adjacent_os} contains no converging connections, there must be an $x_k \in ( \parentsdown{\ot}{w} \sqcup \{\ot\} )$ for some $k \in \{1, \dots, n\}$.
The first restriction combined with the definition of $\ot$ implies that $\parentsdown{\ot}{w} \sqcup \{\ot\} \subseteq \Ovk$,
and thus $x_k \in \Ovk$ which contradicts the assumption of Lemma~\ref{lemma:adjacent_os} that $x_i$ belongs to $\V \setminus \Ovk$ for every $i = 1, \dots, n$.

In the previous paragraph, we have proved that
$o \notin \Ovk \setminus ( \parentsdown{\ot}{w} \sqcup \{\ot\} )$.
Since $o\in \Ovk$, this implies that $o \in \parentsdown{\ot}{w} \sqcup \{\ot\}$.
Moreover, we assumed that $o\neq \ot$, and therefore $o\in \parentsdown{\ot}{w}$ which means that 
$o \ra \ot$ and $o<_{\displaystyle\ot} w$.
%
%
%
%
%
%
%
%
%
%

We have $\ot \in \ch{o}$ and $\ot \in \ch{w}$. Therefore, by Lemma~\ref{lemma:ac_all_parents}(i),
$\ot \in \ch{x_i}$ for all $i \in \{1,\dots, n\}$. 
Therefore, all $x_i$ must belong to
$\pa{\ot} = \parentsdown{\ot}{w} \sqcup \{w\}
\sqcup \parentsup{\ot}{w}$.
None of them is equal to $w$ since \eqref{trail:lemma:adjacent_os} is a shortest trail.
By assumption, none of the $x_i$ belong to $\Ovk$;
the first restriction states that
$\parentsdown{\ot}{w} \subseteq \Ovk$;
therefore all $x_i$ belong to $\parentsup{\ot}{w}$.
This means that
$w <_{\displaystyle \ot} x_i$ for all $i=1,\dots,n$.

Now, we have
$o
<_{\displaystyle \ot} w
<_{\displaystyle \ot} x_i$
for all $i\in \{1,\dots,n\}$.
This means that during the construction of $\smallerbig{\ot}$ in the algorithm we had $w\in \PossCand{O_{\ot}}$ for a partial order $O_{\ot}$ which contains $o$ but not $\{x_i\}_{i=1}^n$.
Therefore, by the induction hypothesis we obtain that $w$ and $o$ are adjacent, which finishes the proof for this case.

\medskip

\underline{Case 2: $w\in \PossCandOut{\Ovk}$}.
By Definition of $\PossCandOut{\Ovk}$ (\cref{prop:characterization_poss_candidates}),
there exists an $\ot\in \Ovk$ such that
$\ot \ra w$ satisfying the following conditions:
\begin{enumerate}
    \item $\parentsdown{w}{\ot}\subseteq \Ovk$.
    \item $\dsepbig{w}{\Ovk \setminus ( \parentsdown{w}{\ot} \sqcup \{ \ot \} )}{ \parentsdown{w}{\ot} \sqcup \{ \ot \} }$.
\end{enumerate}
If $\ot = o$ the proof is complete. We now assume $\ot \neq o$.

If $o \notin \parentsdown{w}{\ot} \sqcup \{ \ot \}$, then $o \in \Ovk \setminus ( \parentsdown{w}{\ot} \sqcup \{ \ot \} )$ and therefore \eqref{trail:lemma:adjacent_os} is a trail from $w$ to $\Ovk \setminus ( \parentsdown{w}{\ot} \sqcup \{ \ot \} )$.
By the second restriction above,
this trail must be blocked by $\parentsdown{w}{\ot} \sqcup \{ \ot \}$.
Because this trail has no converging connection there must be an
$i \in \{1, \dots, n\}$ such that
$x_i \in \parentsdown{w}{\ot} \sqcup \{ \ot \}
\subseteq \Ovk$
by the first restriction and the definition of $\ot$.
This is a contradiction since by the assumption of the lemma $x_i$ is in $\V \setminus \Ovk$.

Therefore we have shown that $o \in \parentsdown{w}{\ot} \sqcup \{ \ot \}$, which implies (by definition of this set) that $o$ and $w$ are adjacent, proving the lemma.
\end{proof}

Lemma~\ref{lemma:adjacent_os} immediately implies a very useful property of partial orders generated by our algorithm, which is proven in the corollary below.
\begin{corollary}
    Under Assumption~\ref{assumpt:good_graphs},
    let $o_i, o_j \in \Ovk$,
    such that $o_i$ (respectively $o_j$) is the $i$-th node
    (respectively $j$-th node) in the partial order $\Ovk$ and $i \neq j$.
    If there exists a trail
    \begin{align}\label{eq:cor_adjacent_os_trail}
        o_i \har x_1 \har \cdots \har x_n \har o_j,
    \end{align}
    with no converging connection such that
    $\forall m = 1, \dots, n$,
    $x_m \in \V \setminus \Ovk$,
    then $o_i$ and $o_j$ are adjacent.
   \label{cor:adjacent_os}
\end{corollary}
\begin{proof}
    Without loss of generality, we can assume that $i<j$.
    This means that $o_{j} \in \PossCand{O^{j-1}_v}$, with $o_i \in O^{j-1}_v$.
    Remark that $O^{j-1}_v \subseteq \Ovk$. 
    Therefore, for all $m = 1, \dots, n$,
    $x_m \in \V \setminus O^{j-1}_v$.
    Hence, by Lemma~\ref{lemma:adjacent_os},
    $o_i$ and $o_j$ are adjacent.
\end{proof}

We now prove a lemma which states that sets which are d-separated cannot be adjacent.
It is quite trivial but it will be useful in Lemma~\ref{lemma:o2_if_higher_os}.

\begin{lemma}
    \label{lemma:dsep_adjacent}
    Let $\G = (\V, \E)$ be a DAG and $X$, $Y$, $Z$ subsets of $\V$ such that $\dsepbig{X}{Y}{Z}$.
    Then $X$ and $Y$ cannot be adjacent, in the sense that $\forall x, y \in X \times Y$,
    $x \nothar y$.
\end{lemma}
\begin{proof}
    Let $x \in X$ and $y \in \V$ such that $x \har y$.
    $y$ is adjacent to $x \in X$ so the trail $x \har y$ is active given $Z$.
    This shows that $\notdsepbig{X}{\{y\}}{Z}$, which contradicts $\dsepbig{X}{Y}{Z}$.
\end{proof}

The lemma below states that under certain conditions an arc between a node $w\in \BOvk\setminus \Ovk$ and a node $o_i\in \Ovk$ implies the existence of another node $\ot \in \Ovk$ such that $w\ra \ot\in \E$ and $o_i\ra \ot\in \E$.

\begin{lemma}
\label{lemma:o2_if_higher_os}
    Following Definition~\ref{def:partial_order},
    let us write the partial order $\Ovk$ as $\Ovk = (o_1, \dots, o_k)$.
    Under Assumption~\ref{assumpt:good_graphs},
    let $i \in \{1, \dots, k\}$
    and $w \in \BOvk\setminus \Ovk$. 
    If $w \ra o_i$ and
    $\parentsup{o_i}{w} \cap \Ovk \neq \emptyset$,
    then there exists an
    $\ot\in \Ovk$ such that $\ot \neq o_i$ and $\G$ contains the subgraph below.
    \begin{figure}[H]
    \centering
    \begin{tikzpicture}[scale=1, transform shape, node distance=1.2cm, state/.style={circle, font=\Large, draw=black, minimum size=1cm}]
    \node (w) {$w$};
    \node[below = of w] (oi) {$o_i$};
    \node[right= of oi] (ot) {$\ot$};
    
    \begin{scope}[>={Stealth[length=6pt,width=4pt,inset=0pt]}]
    \path [->] (w) edge node[scale=0.8, above left] {} (oi);
    \path [->] (w) edge node[scale=0.8, above left] {} (ot);
    \path [->] (oi) edge node[scale=0.8, above left] {} (ot);
    \end{scope}
    \end{tikzpicture}
    \end{figure}
    \end{lemma}
\begin{proof}
Note that $w \in \pa{o_i} = \parentsdown{o_i}{w}
\sqcup \{ w \} \sqcup \parentsup{o_i}{w}$.
Since $\parentsup{o_i}{w} \cap \Ovk \neq \emptyset$,
let $o_j$ be its maximum element
according to $<_{\displaystyle o_i}$.
Consequently, $o_j \ra o_i$ and $w<_{\displaystyle o_i} o_j$.
%

Assume that there exists an $\ot\in \Ovk$ such that $\G$ contains the v-structure $o_i \ra \ot \la o_j$.
Thus, $\G$ contains the subgraph below.
\begin{figure}[H]
\centering
\begin{tikzpicture}[scale=1, transform shape, node distance=1.2cm, state/.style={circle, font=\Large, draw=black, minimum size=1cm}]
\node (w) {$w$};
\node[below = of w] (oi) {$o_i$};
\node[above right= of oi] (oj) {$o_j$};
\node[right= of oi] (ot) {$\ot$};
\begin{scope}[>={Stealth[length=6pt,width=4pt,inset=0pt]}]
\path [->] (w) edge node[scale=0.8, above left] {} (oi);
\path [->] (oj) edge node[scale=0.8, above left] {} (oi); 
\path [->] (oi) edge node[scale=0.8, above left] {} (ot); 
\path [->] (oj) edge node[scale=0.8, above left] {} (ot); 
\end{scope}
\end{tikzpicture}
\end{figure}
Observe that $o_j \in \pa{o_i}$ and $o_j \in \pa{\ot}$;
by definition of the B-sets, $o_j\in B(o_i,\ot)
= \pa{o_i} \cap \pa{\ot}$.
By Corollary \ref{cor:B_sets_implies_rightarrow},
since $w <_{\displaystyle o_i} o_j$,
we obtain $w \ra \ot$, giving us the desired subgraph and finishing the proof under the assumption of existence of $\ot$.

\medskip

There remains to prove the existence of such an $\ot$.
We consider two cases; when $i<j$ and $j<i$.

\underline{Case 1: $i<j$.}
In this case, since $o_j \ra v$, $o_j$ was added at the step $j-1$, and therefore
we must have $o_j \in \PossCand{O^{j-1}_v}$.
Because $i < j$, we obtain $o_i \in O^{j-1}_v$.
Therefore, $\notdsepbig{o_j}{O^{j-1}_v}{\emptyset}$ because of the arc $o_j\ra o_i$.
Hence, $o_j\notin \PossCandInd{O^{j-1}_v}$, and therefore $o_j$ must be in $\PossCandIn{O^{j-1}_v}$ or $\PossCandOut{O^{j-1}_v}$.
We consider both cases.
\begin{itemize}
    \item \underline{$o_j\in \PossCandIn{O^{j-1}_v}$}: 
    By \cref{prop:characterization_poss_candidates} there must be a node $\ot \in O^{j-1}_v$ such that $o_j\ra \ot$ satisfying following restrictions:
    \begin{enumerate}
        \item $\parentsdown{\ot}{o_j}\subseteq O^{j-1}_v$.
        \item $\dsepbig{o_j}{O^{j-1}_v\setminus ( \parentsdown{\ot}{o_j} \sqcup \{ \ot \} )}{ \parentsdown{\ot}{o_j} \sqcup \{ \ot \} }$.
    \end{enumerate}
    Remark that $w\in \parentsdown{o_i}{o_j}$ and $w \notin \Ovk \supseteq O^{j-1}_v$, and thus $\parentsdown{o_i}{o_j} \nsubseteq O^{j-1}_v$.
    This shows that $\ot \neq o_i$
    otherwise the first restriction could not be satisfied.
    
    By combining Lemma~\ref{lemma:dsep_adjacent} and the second restriction,
    no node in $O^{j-1}_v\setminus ( \parentsdown{\ot}{o_j} \sqcup \{ \ot \} )$ can be adjacent to $o_j$.
    Because we have the arc $o_j \ra o_i$ we can deduce that $o_i\notin O^{j-1}_v\setminus ( \parentsdown{\ot}{o_j} \sqcup \{ \ot \} )$.
    
    Since $o_i\in O^{j-1}_v$, this means that
    $o_i \in \parentsdown{\ot}{o_j} \sqcup \{ \ot \}$,
    and therefore $o_i \ra \ot$.
    Now, we have $o_j \ra \ot$ and $o_i\ra \ot$ which is the desired v-structure.
    
    \item \underline{$o_j\in \PossCandOut{O^{j-1}_v}$}: 
    By \cref{prop:characterization_poss_candidates} there must be a node $\ot \in O^{j-1}_v$ such that $\ot\ra o_j$ and the following restrictions are satisfied:
    \begin{enumerate}
        \item $\parentsdown{o_j}{\ot}\subseteq O^{j-1}_v$.
        \item $\dsepbig{o_j}{O^{j-1}_v\setminus( \parentsdown{o_j}{\ot} \sqcup \{ \ot \} )}{ \parentsdown{o_j}{\ot} \sqcup \{ \ot \} }$.
    \end{enumerate}
    If $\ot=o_i$, then $\G$ contains the cycle $o_i \ra o_j \ra o_i$, which is a contradiction, and thus $\ot \neq o_i$.

    Combining Lemma~\ref{lemma:dsep_adjacent} and the second restriction,
    no point in $O^{j-1}_v \setminus ( \parentsdown{o_j}{\ot} \sqcup \{ \ot \} )$ can be adjacent to $o_j$.
    Because we have the arc $o_j \ra o_i$ we can deduce that $o_i\notin O^{j-1}_v \setminus ( \parentsdown{o_j}{\ot} \sqcup \{ \ot \} )$.

    Since $o_i\in O^{j-1}_v$, this means that
    $o_i \in \parentsdown{o_j}{\ot} \sqcup \{ \ot \}$,
    and therefore $o_i \ra o_j$.
    This provides the cycle $o_i \ra o_j \ra o_i$ which
    gives a contradiction, showing that
    $o_j$ cannot be in $\PossCandOut{O^{j-1}_v}$.
\end{itemize}

\underline{Case 2: $j<i$.} 
In this case, since $o_i \ra v$, $o_i$ was added at the step $i-1$, therefore
we must have $o_i \in \PossCand{O^{i-1}_v}$.
Because $j < i$, we obtain $o_j \in O^{i-1}_v$.
Hence, $\notdsepbig{o_i}{O^{i-1}_v}{\emptyset}$ due to existence of the arc $o_j\ra o_i$. 
We get $o_i\notin \PossCandInd{O^{i-1}_v}$.
Thus, $o_i$ must be in $\PossCandIn{O^{i-1}_v}$ or $\PossCandOut{O^{i-1}_v}$.
Both cases are considered below.
\begin{itemize}
    \item \underline{$o_i\in \PossCandIn{O^{i-1}_v}$}: 
    By \cref{prop:characterization_poss_candidates} there must be a node $\ot \in O^{i-1}_v$ such that $o_i\ra \ot$ satisfying:
    \begin{enumerate}
        \item $\parentsdown{\ot}{o_i}\subseteq O^{i-1}_v$.
        \item $\dsepbig{o_i}{O^{i-1}_v\setminus( \parentsdown{\ot}{o_i} \sqcup \{ \ot \} )}{ \parentsdown{\ot}{o_i} \sqcup \{ \ot \} }$.
    \end{enumerate} 
    Note that $\ot= o_j$ creates the cycle $o_j \ra o_i \ra o_j$ which is a contradiction, and thus $\ot\neq o_j$.

    As before, combining Lemma~\ref{lemma:dsep_adjacent}, the second restriction, and the arc $o_j \ra o_i$ implies that $o_j\notin O^{i-1}_v\setminus ( \parentsdown{\ot}{o_i} \sqcup \{ \ot \} )$, and therefore $o_j \in \parentsdown{\ot}{o_i} \sqcup \{ \ot \} $ which means that $o_j\ra \ot$, giving us the desired v-structure.
    
    \item \underline{$o_i\in \PossCandOut{O^{i-1}_v}$}: 
    By \cref{prop:characterization_poss_candidates} there must be a node $\ot \in O^{i-1}_v$ such that $\ot\ra o_i$ satisfying:
    \begin{enumerate}
        \item $\parentsdown{o_i}{\ot}\subseteq O^{i-1}_v$.
        \item $\dsepbig{o_i}{O^{i-1}_v\setminus( \parentsdown{o_i}{\ot} \sqcup \{ \ot \} )}{ \parentsdown{o_i}{\ot} \sqcup \{ \ot \} }$.
    \end{enumerate}
    Remark that $w\in \parentsdown{o_i}{o_j}$ and
    $w \notin \Ovk \supseteq O^{j-1}_v$, and thus
    $\parentsdown{o_i}{o_j} \nsubseteq O^{j-1}_v$.
    This shows that $\ot \neq o_j$
    otherwise the first restriction could not be satisfied.

    Combining Lemma~\ref{lemma:dsep_adjacent}, the second restriction, and the existence of the arc $o_j \ra o_i$ implies that $o_j\notin O^{i-1}_v \setminus ( \parentsdown{o_i}{\ot} \sqcup \{ \ot \} )$.
    Since $o_j \in O^{i-1}_v$
    this means that $o_j \in \parentsdown{o_i}{\ot} \sqcup \{ \ot \}$, and thus $o_j \smallerbig{o_i} \ot$.
    Therefore, $\ot \in \parentsup{o_i}{w}$ since $w \smallerbig{o_i} o_j$.
    However, this is a contradiction since $o_j \neq \ot$ was chosen to be the maximum element in $\parentsup{o_i}{w} \cap \Ovk$.
    This shows that it cannot happen that $o_i$ is in $\PossCandOut{O^{i-1}_v}$.
\end{itemize}
\end{proof}

\section{Lemmas to construct sequences of nodes}

\subsection{Possible candidates by incoming arc}
\label{subsec:lemma:existence_trail_incoming_case}

\begin{assumption}
\label{assumpt:lemma:sequence_incoming_arc}
    Assumption~\ref{assumpt:good_graphs} holds.
    Furthermore, $w\in B(\Ovk)\setminus \Ovk$
    and $o \in \Ovk$ are nodes such that

    (i) $w \ra o \in \E$\;
    (ii) $\parentsdown{o}{w} \subseteq \Ovk$,\;
    (iii) $\parentsup{o}{w} \cap \Ovk = \emptyset$.
\end{assumption}

\begin{lemma}
\label{lemma:sequence_of_os_no_converging}
    Assume that Assumption~\ref{assumpt:lemma:sequence_incoming_arc} holds and that $w \notin \PossCandIn{\Ovk}$.
    Then there exists a trail from $w$ to an element of 
    $\Ovk \setminus \big( \parentsdown{o}{w} \sqcup \{o\} \big)$
    which is activated by 
    $\parentsdown{o}{w} \sqcup \{o\}$
    and contains no converging connections.
\end{lemma}

\begin{proof}
    Since $w \notin \PossCandIn{\Ovk}$,
    $w$ cannot be a possible candidate by the incoming arc $w \ra o$.
    This means that one of the two restrictions must be violated:
    \begin{enumerate}
        \item $\parentsdown{o}{w} \subseteq \Ovk$.
        \item $\dsepbig{w}{
        \Ovk \setminus \big( \parentsdown{o}{w} \sqcup \{o\} \big)}{\parentsdown{o}{w} \sqcup \{o\}}$.
    \end{enumerate}
    The first restrictions is satisfied by the assumption of the lemma.
    Therefore, the second restrictions must be violated, meaning that
    $\notdsepbig{w}{\Ovk \setminus \big( \parentsdown{o}{w} \sqcup \{o\} \big)}{\parentsdown{o}{w} \sqcup \{o\}}$.
    Hence, the set $\TRAILSbig{X}{Y}{Z}$ is not empty.
    Consequently, there exists a minimal trail 
    \begin{equation*}
        w \har \cdots \ra 
        c_1 \la \cdots \ra 
        c_2 \la \cdotslong \ra 
        c_C \la \cdots \har y
    \end{equation*}
    in $\TRAILSbig{X}{Y}{Z}$ according to $\smallerTRAIL$ with $X = \{w\}$, $Y := \Ovk \setminus \big( \parentsdown{o}{w} \sqcup \{o\} \big)$ and $Z := \parentsdown{o}{w} \sqcup \{o\}$.

    \medskip
    
    Therefore, $\TRAILSbig{X}{Y}{Z}$ is not empty.
    Lemma~\ref{lemma:minimal_trail_for_incoming_arc}\ref{fact:o7}
    implies that the existence of a minimal trail in $\TRAILSbig{X}{Y}{Z}$ containing no converging connections.
    This concludes the proof of Lemma~\ref{lemma:sequence_of_os_no_converging}.

\end{proof}


\begin{lemma}
\label{lemma:minimal_trail_for_incoming_arc}
Assume that Assumption~\ref{assumpt:lemma:sequence_incoming_arc} holds and take a minimal trail
\begin{equation}
    w \har \cdots \ra 
    c_1 \la \cdots \ra 
    c_2 \la \cdotslong \ra 
    c_C \la \cdots \har y
    \label{trail:sequence_of_os}
\end{equation}
in the set $\TRAILSbig{X}{Y}{Z}$ according to $\smallerTRAIL$ with
$X := \{w\}$,
$Y := \Ovk \setminus \big( \parentsdown{o}{w} \sqcup \{o\} \big)$ and
$Z := \parentsdown{o}{w} \sqcup \{o\}$.

\medskip

\noindent
Then, the following statements hold:
\begin{enumerate}[label = (\roman*)]
    \item  If the node $o$ is included in the trail, then it must be the first node, i.e. $o = c_1$.
    If this is the case, then $w \ra o$ is the first subtrail.
    \label{fact:o1}
    
    \item If $c_1 \in Z = \parentsdown{o}{w} \sqcup \{o\}$ then $\G$ contains one of the three subgraphs below presented in Figure~\ref{fig:trail:start}.
    \begin{figure}[H]
    \centering
        \begin{subfigure}{0.3\linewidth}
         \centering
         \begin{tikzpicture}[scale=1, transform shape, node distance=2cm, state/.style={circle, font=\Large, draw=black, minimum size=1cm}]

        \node (w) {$w$};
        \node[right of=w] (ob) {$c_1 = o$};

        \begin{scope}[>={Stealth[length=6pt,width=4pt,inset=0pt]}]
        
            \path [->] (w) edge node {} (ob);
        
        \end{scope}
        \end{tikzpicture}
        \caption{}
        \label{fig:trail:start1}
        \end{subfigure}
        \hfill
        \begin{subfigure}{0.3\linewidth}
        \centering
        \begin{tikzpicture}[scale=1, transform shape, node distance=2cm, state/.style={circle, font=\Large, draw=black, minimum size=1cm}]

        \node (xn-1) {$t^0_{n-1}$};
        \node[above right of=xn-1] (xn) {$t^0_n$};
        \node[below right of=xn] (c1) {$c_1$};
        \node[below right of=xn-1] (ob) {$o$};

        \begin{scope}[>={Stealth[length=6pt,width=4pt,inset=0pt]}]
        
            \path [->] (xn) edge node {} (xn-1);
            \path [->] (xn) edge node {} (c1);
            \path [->] (c1) edge node {} (xn-1);
            \path [->] (c1) edge node {} (ob);
        
        \end{scope}
        \end{tikzpicture}
        \caption{}
        \label{fig:trail:start2}
        \end{subfigure}
        \hfill
        \begin{subfigure}{0.3\linewidth}
        \centering
        \begin{tikzpicture}[scale=1, transform shape, node distance=2cm, state/.style={circle, font=\Large, draw=black, minimum size=1cm}]

        \node (xn) {$t^0_n$};
        \node[below right of=xn] (ob) {$o$};
        \node[above right of = ob] (c1) {$c_1$};

        \begin{scope}[>={Stealth[length=6pt,width=4pt,inset=0pt]}]
        
            \path [->] (xn) edge node {} (c1);
            \path [->] (xn) edge node {} (ob);
            \path [->] (c1) edge node {} (ob);
        
        \end{scope}
        \end{tikzpicture}
        \caption{}
        \label{fig:trail:start3}
        \end{subfigure}
        
        \caption{Subgraphs for which one must be included in $\G$ in case that $c_1 \in Z$.}
        
        \label{fig:trail:start}
    \end{figure}
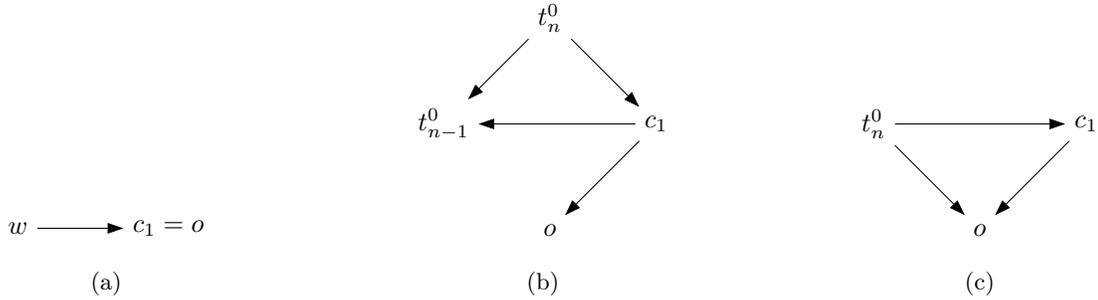
    \label{fact:o2}
    
    \item If $C > 1$, then without loss of generality we can assume that $c_{C-1} \la c_C$,
    in the sense that there exists a minimal trail 
    (according to $\smallerTRAIL$)
    in $\TRAILSbig{X}{Y}{Z}$ such that $c_{C-1} \la c_C$.
    \label{fact:o3}
    
    \item The node $y$ is not in $\pa{o}$.
    \label{fact:o4}
    
    \item There exists a node $\yt \in Y =  \Ovk \setminus \big( \pa{o} \sqcup \{o\} \big)$ such that $o \ra \yt$ and $\forall i=1,\dots,C$, $Z(c_i) \ra \yt$
    (whenever $Z(c_i) \neq \yt$).
    Furthermore, if $\yt \neq y$, then $y \ra \yt$.
    \label{fact:o5}
    
    \item The total number of converging nodes $C$ cannot be strictly larger than $1$.
    \label{fact:o6}
    
    \item The trail has no converging connections, i.e. $C=0$.
    \label{fact:o7}
\end{enumerate}
\end{lemma}

\noindent
\textit{Proof of Lemma~\ref{lemma:minimal_trail_for_incoming_arc}}.
First, note that $Y \sqcup Z = \Ovk$ have local relationships in the sense of Definition~\ref{def:propadj}, by Corollary \ref{cor:adjacent_os}.
Therefore, by Theorem~\ref{thm:minimal_trails_local_relationship}\ref{prop:subgraphs}, $\G$ contains the subgraph below with
$y \in Y
= \Ovk \setminus \big( \parentsdown{o}{w} \sqcup \{o\} \big)$.

\begin{figure}[H]
\centering
\begin{tikzpicture}[scale=1, transform shape, node distance=1cm, state/.style={circle, font=\Large, draw=black}]

\node (c1) {$c_{1}$};
\node[right = of c1] (c2) {$c_2$};
\node[right = 3cm of c2] (cC) {$c_C$};
\node[right = of cC] (o) {$y$};
\node[left = of c1] (t0n) {$t^0_n$};
\node[left = of t0n] (t0n-1) {$t^0_{n-1}$};
\node[left = of t0n-1] (t01) {$t^0_{1}$};
\node[left = of t01] (w) {$w$};

\node[above = of c1] (c1h1) {};
\node[above = of c1h1] (c1h2) {};
\node[above = of c1h2] (Yc1) {$Z(c_1)$};

\node[above = of c2] (c2h1) {};
\node[above = of c2h1] (c2h2) {};
\node[above = of c2h2] (Yc2) {$Z(c_2)$};

\node[above = of cC] (cCh1) {};
\node[above = of cCh1] (cCh2) {};
\node[above = of cCh2] (YcC) {$Z(c_C)$};

\begin{scope}[>={Stealth[length=6pt,width=4pt,inset=0pt]}]


    \path [-] (c1) edge[bend left=45] node {} (c2);
    \path [-] (cC) edge[bend left=45] node {} (o);

    \draw[transform canvas={yshift=0.21ex},-left to,line width=0.25mm] (c1) -- (c2);
    \draw[transform canvas={yshift=-0.21ex},left to-,line width=0.25mm] (c1) -- (c2);

    \draw[transform canvas={yshift=0.21ex},-left to,line width=0.25mm] (cC) -- (o);
    \draw[transform canvas={yshift=-0.21ex},left to-,line width=0.25mm] (cC) -- (o);

    \draw[loosely dotted, line width=0.5mm] (c2) -- (cC);

    \path [->] (c1) edge node {} (c1h1);
    \path [->] (c1h2) edge node {} (Yc1);
    \draw[loosely dotted, line width=0.5mm] (c1h1) -- (c1h2);

    \path [->] (c2) edge node {} (c2h1);
    \path [->] (c2h2) edge node {} (Yc2);
    \draw[loosely dotted, line width=0.5mm] (c2h1) -- (c2h2);

    \path [->] (cC) edge node {} (cCh1);
    \path [->] (cCh2) edge node {} (YcC);
    \draw[loosely dotted, line width=0.5mm] (cCh1) -- (cCh2);

    \draw[loosely dotted, line width=0.5mm] (t01) -- (t0n-1);
    \path [->] (t0n) edge node {} (c1);
    \draw[transform canvas={yshift=0.21ex},-left to,line width=0.25mm] (w) -- (t01);
    \draw[transform canvas={yshift=-0.21ex},left to-,line width=0.25mm] (w) -- (t01);
    \draw[transform canvas={yshift=0.21ex},-left to,line width=0.25mm] (t0n-1) -- (t0n);
    \draw[transform canvas={yshift=-0.21ex},left to-,line width=0.25mm] (t0n-1) -- (t0n);
    
\end{scope}
\end{tikzpicture}
\end{figure}

Each of these properties is proved, respectively in the following Sections~\ref{subsec:proof:lemma:minimal_trail_for_incoming_arc:o1},
\ref{subsec:proof:lemma:minimal_trail_for_incoming_arc:o2}, 
\ref{subsec:proof:lemma:minimal_trail_for_incoming_arc:o3}, 
\ref{subsec:proof:lemma:minimal_trail_for_incoming_arc:o4}, 
\ref{subsec:proof:lemma:minimal_trail_for_incoming_arc:o5}, 
\ref{subsec:proof:lemma:minimal_trail_for_incoming_arc:o6}
and
\ref{subsec:proof:lemma:minimal_trail_for_incoming_arc:o7}.

\subsubsection{Proof of Lemma~\ref{lemma:minimal_trail_for_incoming_arc}\ref{fact:o1}}
\label{subsec:proof:lemma:minimal_trail_for_incoming_arc:o1}



Naturally, if the trail \eqref{trail:sequence_of_os} contains the node
$o \in
Z = \parentsdown{o}{w} \sqcup \{o\}$, 
then it must correspond to a converging connection.
Otherwise, \eqref{trail:sequence_of_os} would be blocked by $Z$.

Consider the case when a node $c_i$ with $i \in \{2, \dots, C \}$ is equal to $o$.
Then, the trail
$$
w \ra o \la \cdots \ra c_{i+1} \la \cdotslong \ra c_C \la \cdots \har y
$$
would be a better trail than \eqref{trail:sequence_of_os} which is a contradiction.
Hence, if $o$ is located along the trail it must be equal to $c_1$.

In this case the subtrail
$
w \har t^0_1 \har \cdots \har t^0_n \ra o
$
takes the form $w \ra o$, since we picked a minimal trail.
This concludes the proof of \ref{fact:o1}.

\subsubsection{Proof of Lemma~\ref{lemma:minimal_trail_for_incoming_arc}\ref{fact:o2}}
\label{subsec:proof:lemma:minimal_trail_for_incoming_arc:o2}



Consider the subtrail
\begin{equation}
\label{trail:sequence_of_ws_start}
w \har t^0_1 \har \cdots \har t^0_n \ra c_1.
\end{equation}

If $o$ is in the trail~\eqref{trail:sequence_of_os},
then by \ref{fact:o1}, it is the first node along the trail, i.e. $\G$ contains the subgraph in Figure~\ref{fig:trail:start1}.

\medskip

Suppose that $o$ is not located along the trail.
Therefore, $c_1 \neq o$.
Because we assumed that
$c_1 \in Z
= \parentsdown{o}{w} \sqcup \{o\}$,
we know that $c_1$ belongs to $\parentsdown{o}{w}$ and therefore $c_1 \ra o$.
If $t^0_{n}$ is in $\pa{o}$, then $t^0_{n} \ra o$.
In this case, we find the subgraph in Figure~\ref{fig:trail:start3}.

\medskip

We will now show that
if $t^0_{n}$ is not in $\pa{o}$ and $o$ is not in \eqref{trail:sequence_of_os},
then $\G$ must contain the subgraph in Figure~\ref{fig:trail:start2}.
First, we define an integer
$$
p := \max\{i\in\{1,\dots,n \};\, t^0_i\in \pa{o} \}
$$
such that $t^0_p$ is the furthest node from $w$ in trail~\eqref{trail:sequence_of_ws_start} contained in the set $\pa{o}$.
By
Theorem~\ref{thm:minimal_trails}\ref{prop:no_chords},
the subtrail
$$
t^0_p \har \cdots \har t^0_n \ra c_1
$$
contains no chords.
Moreover, the nodes $t^0_p$ and $c_1$ are in $\pa{o}$, and the nodes $t^0_{p+1}, \dots, t^0_n$ are not in $\pa{o}$.

\medskip

Remark that $t^0_p \har \cdots \har t^0_n \ra c_1$ is a shortest trail activated by the empty set from $t^0_p$ to $c_1$ ending with a rightward pointing arrow consisting of nodes in $\V \setminus Z$.
Therefore, we may apply Lemma~\ref{lemma:generalisation_to_K} and Theorem~\ref{lemma:v1_trail_to_v2_parents_of_v3}
(with $v_1=t^0_p$, $v_2=c_1$ and $v_3=o$ in the notation of Theorem~\ref{lemma:v1_trail_to_v2_parents_of_v3})
to find that
$t^0_n \ra t^0_{n-1}$ and
$c_1 \ra t^0_{n-1}$.
Since we also know that $c_1 \ra o$,
we conclude that
to find that $\G$ must contain the subgraph in Figure~\ref{fig:trail:start2}.

\subsubsection{Proof of Lemma~\ref{lemma:minimal_trail_for_incoming_arc}\ref{fact:o3}} 
\label{subsec:proof:lemma:minimal_trail_for_incoming_arc:o3}



By Theorem~\ref{thm:minimal_trails_local_relationship}\ref{prop:adjacency},
we know that $c_{C-1}$ and $c_C$ are adjacent.
If $c_{C-1} \ra c_C$ then the proof is completed.
There only remains to study the case where
$c_{C-1} \ra c_C$.
From Theorem~\ref{thm:minimal_trails_local_relationship}\ref{prop:adjacency}, it follows that for all $i = 1, \dots, C$, the nodes $c_i$ and $c_{i+1}$ are adjacent.
By Lemma~\ref{prop:no_diverging_cs}, we know that for all $i = 2, \dots, C$, the trail $c_{i-1} \la c_i \ra c_{i+1}$ cannot be present.
Therefore, we get that for all $i = 1, \dots, C$, $c_i \ra c_{i+1}$.
From Theorem~\ref{thm:minimal_trails_local_relationship}\ref{prop:final_c_in_Z}, it follows that $c_C \in Z$.
Combining this with 
Corollary~\ref{cor:all_cs_in_Z}\ref{prop:rightarroww_all_in_Z},
we find that $\forall i=1, \dots, C-1$,
$c_i \in Z$.
In particular we have that $c_1 \ra c_2$
and $c_1, c_2 \in Z$.

\medskip

We will show that the arc $c_1 \ra c_2$ with $c_1,c_2 \in Z$ leads
the existence of another minimal trail, that satisfies $c_{C-1} \la c_C$.
Since $c_1 \in Z$ we can apply \ref{fact:o2}, to find that $\G$ contains one of three subgraphs in Figure~\ref{fig:trail:start}.
We consider each case separately.

\medskip

\noindent
\underline{Case 1: Subgraph~\ref{fig:trail:start1}.}\\
In this case $c_1 = o$.
Since $c_1 \ra c_2$, we have that $o \ra c_2$.
Moreover, because $c_2 \in Z = \parentsdown{o}{w} \sqcup \{o\}$, we know that $c_2 \ra o$, that leads to the cycle $c_2 \ra o \ra c_2$ which is a contradiction.

\medskip

\noindent
\underline{Case 2: Subgraph~\ref{fig:trail:start2}.}\\
Note that $o$ is not equal to $c_1$ (otherwise we would be in the previous case) and therefore $c_2 \neq o$. Since $c_2 \in Z = \parentsdown{o}{w} \sqcup \{o\}$, we get that $c_2 \ra o$.
Combining this and 
Figure~\ref{fig:trail:start2} with Theorem~\ref{thm:minimal_trails_local_relationship}\ref{prop:subgraphs},
we obtain that
$\G$ contains the subgraph below.
$$
\begin{tikzpicture}[scale=1, transform shape, node distance=1.2cm, state/.style={circle, font=\Large, draw=black}]
\node (v1) {$t^0_{n-1}$};
\node[above right= of v1] (xn) {$t^0_n$};
\node[right = 2cm of v1] (v2) {$c_1$};
\node[right = 6cm of v2] (v3) {$c_2$};

\node[above right = 1cm and 0.01cm of v2] (y1) {$t^{1}_1$};
\node[above right = of y1] (y2) {$t^{1}_2$};
\node[above left = 1cm and 0.01cm of v3] (yn) {$t^{1}_{n}$};
\node[above left= of yn] (yn-1) {$t^{1}_{n-1}$};
\node[below = of v2] (ob) {$o$};
\begin{scope}[>={Stealth[length=6pt,width=4pt,inset=0pt]}]

\path [->] (v2) edge node {} (v1);
\path [->] (xn) edge node {} (v1);
\path [->] (xn) edge node {} (v2);

\path [->] (v2) edge node {} (v3);
\path [->] (v2) edge node {} (y2);
\path [->] (v2) edge node {} (yn-1);
\path [->] (v2) edge node {} (yn);
\path [->] (y1) edge node {} (v2);
\path [->] (y1) edge node {} (y2);
\path [->] (yn-1) edge node {} (yn);
\path [->] (yn) edge node {} (v3);
\draw[loosely dotted, line width=0.5mm] (y2) -- (yn-1);

\path [->] (v2) edge node {} (ob);
\path [->] (v3) edge node {} (ob);
\end{scope}
\end{tikzpicture}
$$
Here, we have
\begin{align*}
    t^1_1&\in B(c_1, t^1_2),\\
    t^0_n&\in B(c_1, t^0_{n-1}).
\end{align*}
Since $\G$ does not contain any interfering v-structures, this means that $t^0_n\ra t^1_2$ or $t^1_1\ra t^0_{n-1}$.
These arcs provide us with the following respective trails
\begin{align*}
    w \har \cdots \har t^0_{n-1} \har t^0_n \ra t^1_2 \har \cdotslong \har y,\\
    w \har \cdots \har t^0_{n-1} \la t^1_1 \har t^1_2 \har \cdotslong \har y,
\end{align*}
which are both better than \eqref{trail:sequence_of_os} in the sense of $\smallerTRAIL$. Indeed, \eqref{trail:sequence_of_os} can be rewritten as
\begin{align*}
    w \har \cdots \har t^0_{n-1} \har t^0_n \ra c_1 \la t^1_1 \har t^1_2 \har \cdotslong \har y.
\end{align*}
We therefore get a contradiction.

\medskip

\noindent
\underline{Case 3: Subgraph~\ref{fig:trail:start3}.}\\
In this case, by Theorem~\ref{thm:minimal_trails_local_relationship}\ref{prop:subgraphs}, 
$\G$ contains the subgraph below, where $c_2 \ra o$ because $c_2 \in Z = \parentsdown{o}{w} \sqcup \{o\}$ and $c_2 \neq o$.

$$
\begin{tikzpicture}[scale=1, transform shape, node distance=1.2cm, state/.style={circle, font=\Large, draw=black}]
\node (v1) {$t^0_n$};
\node[right = 2cm of v1] (v2) {$c_1$};
\node[right = 6cm of v2] (v3) {$c_2$};

\node[above right = 1cm and 0.01cm of v2] (y1) {$t^{1}_1$};
\node[above right = of y1] (y2) {$t^{1}_2$};
\node[above left = 1cm and 0.01cm of v3] (yn) {$t^{1}_{n}$};
\node[above left= of yn] (yn-1) {$t^{1}_{n-1}$};
\node[below = of v2] (ob) {$o$};
\begin{scope}[>={Stealth[length=6pt,width=4pt,inset=0pt]}]

\path [->] (v1) edge node {} (v2);

\path [->] (v2) edge node {} (v3);
\path [->] (v2) edge node {} (y2);
\path [->] (v2) edge node {} (yn-1);
\path [->] (v2) edge node {} (yn);
\path [->] (y1) edge node {} (v2);
\path [->] (y1) edge node {} (y2);
\path [->] (yn-1) edge node {} (yn);
\path [->] (yn) edge node {} (v3);
\draw[loosely dotted, line width=0.5mm] (y2) -- (yn-1);

\path [->] (v1) edge node {} (ob);
\path [->] (v2) edge node {} (ob);
\path [->] (v3) edge node {} (ob);
\end{scope}
\end{tikzpicture}
$$
Here, we have $t^0_n\in B(c_1,o)$ and $t^1_1\in B(c_1,t^1_2)$.
By the same argument as above, this means that $t^0_n \ra t^1_2$ or $t^1_1 \ra o$.
The former arc results in a trail from $w$ to $y$, which is a better trail than \eqref{trail:sequence_of_os}, and therefore a contradiction.
Hence, we must have the arc $t^1_1 \ra o$.
This arc provides us with the trail
\begin{equation}
\label{trail:minimal_trail_o_start}
    w \ra o \la t^1_1 \har \cdotslong \har y
\end{equation}
which is better than the trail \eqref{trail:sequence_of_os}, that is
$$
w \har \cdots \har t^0_{n_t(0)} \ra c_1 \la t^1_1 \har \cdotslong \har y,
$$
unless $n_t(0) = 0$.
In this case, \eqref{trail:minimal_trail_o_start}
is also a minimal trail in $\TRAILSbig{X}{Y}{Z}$.

Note that the converging connections of the trail \eqref{trail:minimal_trail_o_start} are $o, c_2, \dots, c_C$. Combining the fact that $o \la c_2$ with 
Lemma~\ref{prop:no_diverging_cs}
we must have 
$$
o \la c_2 \la c_3 \la \cdots \la c_{C}.    
$$
So, we have proven that there exists a minimal trail in $\TRAILSbig{X}{Y}{Z}$ containing an arc $c_{C-1} \la c_C$, completing the proof of \ref{fact:o3}.

%

\subsubsection{
Proof of Lemma~\ref{lemma:minimal_trail_for_incoming_arc}\ref{fact:o4}}
\label{subsec:proof:lemma:minimal_trail_for_incoming_arc:o4}



By definition, $y \in Y
= \Ovk \setminus \big( \parentsdown{o}{w} \sqcup \{o\} \big)$.
Observe that
\begin{align*}
    Y \cap \pa{o}
    &= \big( \Ovk \setminus \big( \parentsdown{o}{w} \sqcup \{o\} \big) \big)
    \cap \big( \parentsdown{o}{w} \sqcup \{o\} \sqcup \parentsup{o}{w} \big) \\
    &= \big( \Ovk \setminus \big( \parentsdown{o}{w} \sqcup \{o\} \big) \big)
    \cap \parentsup{o}{w} \\
    &\subseteq \Ovk \cap \parentsup{o}{w}
    = \emptyset,
\end{align*}
by Assumption~\ref{assumpt:lemma:sequence_incoming_arc}, hence $y \notin \pa{o}$.


\subsubsection{Proof of Lemma~\ref{lemma:minimal_trail_for_incoming_arc}\ref{fact:o5}}
\label{subsec:proof:lemma:minimal_trail_for_incoming_arc:o5}




Suppose that $o$ is located on \eqref{trail:sequence_of_os}.
By \ref{fact:o1}, this means that $o = c_1$, and therefore $\G$ contains the trail
$$
o \la \cdots \ra c_2 \la \cdotslong \ra c_C \la \cdots \har y.
$$
that satisfies all conditions for Lemma~\ref{lemma:super_o_for_sequence}
(by applying Theorem~\ref{thm:minimal_trails}\ref{prop:nodes_along_subtrails}).
Remember that $w \ra o$ by Assumption~\ref{assumpt:lemma:sequence_incoming_arc}.
If $o \neq c_1$, then the trail
$$
o \la w \har \cdots \ra c_1 \la \cdots \ra c_2 \la \cdotslong \ra c_C \la \cdots \har y
$$
also satisfies the conditions of Lemma~\ref{lemma:super_o_for_sequence}
(by applying Theorem~\ref{thm:minimal_trails}\ref{prop:nodes_along_subtrails}).

\medskip

Therefore, in both cases we can apply Lemma~\ref{lemma:super_o_for_sequence} to find that there exists an $\yt$ in $\Ovk$ such that $o \ra \yt$, $y \ra \yt$ and for all $i=1,\dots,C$, $Z(c_i) \ra \yt$
(whenever this does not create a self-loop $\yt \ra \yt$, i.e. in the case where $\yt$ would be equal to $o$, $y$ or some $Z(c_i)$ for $i \in \{1, \dots, C\}$).
It remains to show that this node $\yt$ is in 
$Y =  \Ovk \setminus \big( \pa{o} \sqcup \{o\} \big)$,
Therefore, we only have to show that
$\yt$ cannot be in $\pa{o} \sqcup \{ o \}$.

\medskip

If $\yt = o$, then $\G$ contains the arc $y \ra o$.
This contradicts \ref{fact:o4} which states that $y \notin \pa{o}$.
If $\yt \in \pa{o}$
then $\yt \ra o$, and hence $\G$ contains the cycle $\yt \ra o \ra \yt$ which is a contradiction
(because we showed above that $o \ra \yt$).

\subsubsection{Proof of Lemma~\ref{lemma:minimal_trail_for_incoming_arc}\ref{fact:o6}}
\label{subsec:proof:lemma:minimal_trail_for_incoming_arc:o6}



Assume that \eqref{trail:sequence_of_os} has $C>1$ converging connection.
The end of the trail can have several different types of structures.
By \ref{fact:o3} we can assume that $c_{C-1} \la c_C$ without loss of generality.
If $c_C \ra y$, we obtain the subgraph $c_{C-1} \la c_C \ra c_{C+1}$, which is a contradiction by Lemma~\ref{prop:no_diverging_cs}.
Therefore $c_{C} \la y$.

\medskip

Remark that by Theorem~\ref{thm:minimal_trails_local_relationship}\ref{prop:final_c_in_Z}
the node $c_C$ is in the set $Z := \parentsdown{o}{w} \sqcup \{o\}$.
Note that, by \ref{fact:o1}, $c_C \neq o$,
and therefore $c_C \in \parentsdown{o}{w}$,
so $c_C \ra o$.
Consequently, the graph contains the trail
$y \ra c_C \ra o$.
This means that the node $\yt \in Y$
from \ref{fact:o5} cannot be equal to $y$.
Indeed, this would lead to the cycle
$y \ra c_C \ra o \ra y$ which is a contradiction.

\medskip

Furthermore, $c_C \in Z$, therefore $Z(c_C) = c_C \ra \yt$ by \ref{fact:o5}.
Combining the previous results with Theorem~\ref{thm:minimal_trails_local_relationship}\ref{prop:subgraphs} gives the subgraph below.
$$
\begin{tikzpicture}[scale=1, transform shape, node distance=1.2cm, state/.style={circle, font=\Large, draw=black}]
\node (v1) {$c_{C-1}$};
\node[right = 7cm of v1] (v2) {$c_C$};
\node[right = 3cm of v2] (v3) {$y$};
\node[below right= of v3] (ot) {$\yt$};
\node[below = of v2] (ob) {$o$};

\node[above right = 1cm and 0.01cm of v1] (x1) {$t^{C-1}_1$};
\node[above right = of x1] (x2) {$t^{C-1}_2$};
\node[above left = 1cm and 0.01cm of v2] (xn) {$t^{C-1}_{n}$};
\node[above left= of xn] (xn-1) {$t^{C-1}_{n-1}$};


\begin{scope}[>={Stealth[length=6pt,width=4pt,inset=0pt]}]

\path [->] (v2) edge node {} (v1);
\path [->] (v2) edge node {} (x2);
\path [->] (v2) edge node {} (xn-1);
\path [->] (v2) edge node {} (x1);
\path [->] (x1) edge node {} (v1);
\path [->] (x2) edge node {} (x1);
\path [->] (xn) edge node {} (xn-1);
\path [->] (xn) edge node {} (v2);
\draw[loosely dotted, line width=0.5mm] (x2) -- (xn-1);

\path [->] (v3) edge node {} (v2);

\path [->] (v2) edge node {} (ob);

\path [->] (v2) edge node {} (ot);
\path [->] (v3) edge node {} (ot);
\path [->] (ob) edge node {} (ot);
\end{scope}
\end{tikzpicture}
$$
Here, we have
$t^{C-1}_n \in B(c_C,t^{C-1}_{n-1})$
and $y \in B(c_C, \yt)$.
Since $\G$ does not contain interfering v-structures, we must have $t^{C-1}_n \ra \yt$ or $y \ra t^{C-1}_{n-1}$.
Both arcs provide a trail from $w$ to a node in $Y := \Ovk \setminus \big( \parentsdown{o}{w} \sqcup \{o\} \big)$ which is a better trail than \eqref{trail:sequence_of_os}.
Indeed, the trails
\begin{align*}
    w \har \cdotslong \la t^{C-1}_n \ra \yt,\\
    w \har \cdotslong \la t^{C-1}_{n-1} \la y,
\end{align*}
contain one fewer converging connection than \eqref{trail:sequence_of_os}, which is
\begin{equation*}
    w \har \cdotslong \la t^{C-1}_n \ra c_C \la y.
\end{equation*}
This leads to a contradiction because \eqref{trail:sequence_of_os} was assumed to be a minimal trail and the proof of \ref{fact:o6} is concluded.

\subsubsection{Proof of Lemma~\ref{lemma:minimal_trail_for_incoming_arc}\ref{fact:o7}}
\label{subsec:proof:lemma:minimal_trail_for_incoming_arc:o7}



By \ref{fact:o6}, the trail \eqref{trail:sequence_of_os} has either $0$ or $1$ converging connection.
If it has zero converging connection, then the existence of this trail completes the proof of \ref{fact:o7}.
Therefore we assume that \eqref{trail:sequence_of_os} has exactly one converging connection, i.e. $C=1$.
Furthermore, by Theorem~\ref{thm:minimal_trails_local_relationship}\ref{prop:final_c_in_Z}
we know that $c_C = c_1 \in Z := \parentsdown{o}{w} \sqcup \{o\}$.
This means that we can apply \ref{fact:o2} to find that $\G$ contains one of the three subgraphs in Figure~\ref{fig:trail:start}.
We consider each subgraph separately.
Furthermore, for each case we will consider two sub-cases; when $c_1 \ra y$ and when $c_1 \la y$, since $c_1$ and $y$ are adjacent by Theorem~\ref{thm:minimal_trails_local_relationship}\ref{prop:adjacency}.

\medskip

\noindent
\underline{Case 1: Subgraph~\ref{fig:trail:start1}.}\\
In this case the node $c_1$ is equal to $o$.
Since $y \notin \pa{o}$ by \ref{fact:o4}, the arc $c_1 \la y$ cannot be present.
Therefore, we must have $c_1 \ra y$.
Thus, by Theorem~\ref{thm:minimal_trails_local_relationship}\ref{prop:subgraphs} we know that $\G$ contains the subgraph below.

$$
\begin{tikzpicture}[scale=1, transform shape, node distance=1.2cm, state/.style={circle, font=\Large, draw=black}]

\node (ob) {$c_1 = o$};
\node[right = 6cm of ob] (oc) {$y$};
\node[above = of ob] (y1) {$t^0_1$};
\node[above right= of y1] (y2) {$t^0_2$};
\node[above = of oc] (yn) {$t^0_n$};
\node[above left= of yn] (yn-1) {$t^0_{n-1}$};

\begin{scope}[>={Stealth[length=6pt,width=4pt,inset=0pt]}]

\path [->] (ob) edge node {} (oc);
\path [->] (y1) edge node {} (ob);
\path [->] (y1) edge node {} (y2);
\path [->] (yn-1) edge node {} (yn);
\path [->] (yn) edge node {} (oc);
\path [->] (ob) edge node {} (y2);
\path [->] (ob) edge node {} (yn-1);
\path [->] (ob) edge node {} (yn);
\draw[loosely dotted, line width=0.5mm] (y2) -- (yn-1);

\end{scope}

\end{tikzpicture}
$$

The node $t^0_1$ is in $\V \setminus \Ovk$ by Theorem~\ref{thm:minimal_trails}\ref{prop:nodes_along_subtrails}, and it is also in $\pa{o}$.
This means that $t^0_1 \in \pa{o} \setminus \Ovk$.
By Assumption~\ref{assumpt:lemma:sequence_incoming_arc}, we have
$\parentsdown{o}{w} \subseteq \Ovk$
and $o \in \Ovk$.
Therefore,
$\pa{o} \setminus \Ovk
= \big( \parentsdown{o}{w} \sqcup \{w\}
\sqcup \parentsup{o}{w} \big) \setminus \Ovk
\subseteq \parentsup{o}{w} \sqcup \{w\}$.
Thus, $t^0_1 \in \parentsup{o}{w} \sqcup \{w\}$.
If $t^0_1 = w$, we find a shorter trail than \eqref{trail:sequence_of_os}, which is a contradiction.
Therefore
$t^0_1 \in \parentsup{o}{w}$
and so $w \smallerbig{o} t^0_1$.

\medskip

Because the parental order $\smallerbig{o}$ has been determined by our algorithm, it abides by the B-sets, see Corollary~\ref{cor:previous_orders_abide_by_bsets}.
Therefore, any B-set corresponding to the node $o$ which contains $t^0_1$ must also contain $w$.
Remark that $t^0_1\in B(o, t^0_2)$.
Consequently, we have that $w \in B(o, t^0_2)$.
This means that $w \ra t^0_2$ leads to the trail
$$
w \ra t^0_2 \ra \cdots \ra t^0_n \ra y
$$
which contains no converging connections.
Thus, this trail is better than the trail \eqref{trail:sequence_of_os} which is a contradiction.

\medskip

\noindent
\underline{Case 2: Subgraph~\ref{fig:trail:start2}.}\\
Note that both cases $c_1 \ra y$ and when $y \ra c_1$ must be considered. For both cases we have that $c_1 \in Z = \parentsdown{o}{w} \sqcup \{o\}$, and therefore $c_1 \ra o$.
First, let us assume that $c_1 \ra y$, then by Theorem~\ref{thm:minimal_trails_local_relationship}\ref{prop:subgraphs} we know that $\G$ contains the subgraph below.

$$
\begin{tikzpicture}[scale=1, transform shape, node distance=1.2cm, state/.style={circle, font=\Large, draw=black}]

\node (ob) {$o$};
\node[above = of ob] (oa1) {$c_1$};
\node[above left = of oa1] (xn) {$t^0_n$};
\node[below left = of xn] (xp) {$t^0_{n-1}$};
\node[right = 6cm of oa1] (oc) {$y$};

\node[above = of oa1] (y1) {$t^1_1$};
\node[above right = of y1] (y2) {$t^1_2$};
\node[above = of oc] (yn) {$t^1_{n}$};
\node[above left = of yn] (yn-1) {$t^1_{n-1}$};

\begin{scope}[>={Stealth[length=6pt,width=4pt,inset=0pt]}]

\path [->] (oa1) edge node {} (ob);
\path [->] (oa1) edge node {} (oc);

\path [->] (xn) edge node {} (xp);
\path [->] (xn) edge node {} (oa1);
\path [->] (oa1) edge node {} (xp);

\path [->] (y1) edge node {} (oa1);
\path [->] (y1) edge node {} (y2);
\path [->] (yn-1) edge node {} (yn);
\path [->] (yn) edge node {} (oc);
\path [->] (oa1) edge node {} (yn);
\path [->] (oa1) edge node {} (y2);
\path [->] (oa1) edge node {} (yn-1);
\draw[loosely dotted, line width=0.5mm] (y2) -- (yn-1);

\end{scope}
\end{tikzpicture}
$$

Here, we have that $t^0_n \in B(c_1, t^0_{n-1})$ and $t^1_1 \in B(c_1, t^1_2)$.
Since $\G$ does not contain any interfering v-structures, we must have $t^0_n \ra t^1_2$ or $t^1_1 \ra t^0_{n-1}$.
Both arcs result in a trail from $w$ to $Y$ without converging connections, and therefore lead to contradictions.
Indeed, we find the trails
\begin{align*}
    w \har \cdots \har t^0_n \ra t^1_2 \har \cdots \har y,\\
    w \har \cdots \har t^0_{n-1} \la t^1_1 \har \cdots \har y,
\end{align*}
which are better than \eqref{trail:sequence_of_os}.

Because the arc $c_1 \ra y$ leads to a contradiction, we can assume that $c_1 \la y$.
In this case the $\yt$ whose existence has been established from \ref{fact:o5} cannot be equal to $y$ since this would provide the cycle
$o \ra y \ra c_1 \ra o$, and therefore a contradiction.
Thus, $\G$ contains the subgraph below.
$$
\begin{tikzpicture}[scale=1, transform shape, node distance=1.2cm, state/.style={circle, font=\Large, draw=black}]

\node (ob) {$o$};
\node[above = of ob] (oa1) {$c_1$};
\node[above left = of oa1] (xn) {$t^0_{n}$};
\node[below left = of xn] (xp) {$t^0_{n-1}$};
\node[right = of oa1] (oc) {$y$};
\node[below right = of oc] (ot) {$\yt$};

\begin{scope}[>={Stealth[length=6pt,width=4pt,inset=0pt]}]

\path [->] (oa1) edge node {} (ob);
\path [->] (oc) edge node {} (oa1);

\path [->] (xn) edge node {} (oa1);
\path [->] (xn) edge node {} (xp);
\path [->] (oa1) edge node {} (xp);

\path [->] (oc) edge node {} (ot);
\path [->] (ob) edge node {} (ot);
\path [->] (oa1) edge node {} (ot);

\end{scope}
\end{tikzpicture}
$$

Here, we have $t^0_n \in B(c_1, t^0_{n-1})$ and $y \in B(c_1, \yt)$.
Similarly this means that $t^0_n \ra \yt$ or $y \ra t^0_{n-1}$.
Both arcs result in a trail from $w$ to $Y$ without converging connections, and therefore contradictions.
Indeed, we find the trails
\begin{align*}
    w \har \cdots \har t^0_n \ra \yt,\\
    w \har \cdots \har t^0_{n-1} \la y,
\end{align*}
where the node $\yt$ is in $Y$ by \ref{fact:o5}.
This gives us the existence of the trail as claimed.

\medskip

\noindent
\underline{Case 3: Subgraph~\ref{fig:trail:start3}.}\\
We must consider the two cases $c_1 \la y$ and $c_1 \ra y$.
First, let us assume that $c_1 \la y$, giving us the subgraph below.
Again, $\yt$ cannot be equal to $y$, since this would create a cycle.
Therefore, $\G$ contains the subgraph below.
$$
\begin{tikzpicture}[scale=1, transform shape, node distance=1.2cm, state/.style={circle, font=\Large, draw=black}]

\node (ob) {$o$};
\node[above = of ob] (oa1) {$c_1$};
\node[left = of oa1] (xn) {$t^0_n$};
\node[right = of oa1] (oc) {$y$};
\node[below right = of oc] (ot) {$\yt$};

\begin{scope}[>={Stealth[length=6pt,width=4pt,inset=0pt]}]

\path [->] (oa1) edge node {} (ob);
\path [->] (xn) edge node {} (oa1);
\path [->] (xn) edge node {} (ob);
\path [->] (oc) edge node {} (oa1);

\path [->] (oa1) edge node {} (ot);
\path [->] (ob) edge node {} (ot);
\path [->] (oc) edge node {} (ot);

\end{scope}
\end{tikzpicture}
$$
Here, we have $t^0_n \in B(c_1, o)$ and
$y \in B(c_1, \yt)$.
Therefore, $\E$ must contain $t^0_n \ra \yt$ or $y \ra o$.
The former arc results in a trail
$$
w \har \cdots \har t^0_n \ra \yt
$$
from $w$ to $Y$ without converging connections (and therefore a better trail than \eqref{trail:sequence_of_os}) and the latter would  contradict $y \notin \pa{o}$ (which we know by \ref{fact:o4}).
This means that $c_1 \ra y$.
Therefore, by combining subgraph~\ref{fig:trail:start3} with Theorem~\ref{thm:minimal_trails_local_relationship}\ref{prop:subgraphs} we obtain the subgraph below.
$$
\begin{tikzpicture}[scale=1, transform shape, node distance=1.2cm, state/.style={circle, font=\Large, draw=black}]

\node (ob) {$o$};
\node[above = of ob] (oa1) {$c_1$};
\node[above left = of oa1] (h) {};
\node[below left = of h] (xn) {$t^0_n$};
\node[right = 6cm of oa1] (oc) {$y$};

\node[above = of oa1] (y1) {$t^1_1$};
\node[above right = of y1] (y2) {$t^1_2$};
\node[above = of oc] (yn) {$t^1_{n}$};
\node[above left = of yn] (yn-1) {$t^1_{n-1}$};

\begin{scope}[>={Stealth[length=6pt,width=4pt,inset=0pt]}]

\path [->] (oa1) edge node {} (ob);
\path [->] (oa1) edge node {} (oc);

\path [->] (xn) edge node {} (oa1);
\path [->] (xn) edge node {} (ob);

\path [->] (y1) edge node {} (oa1);
\path [->] (y1) edge node {} (y2);
\path [->] (yn-1) edge node {} (yn);
\path [->] (yn) edge node {} (oc);
\path [->] (oa1) edge node {} (yn);
\path [->] (oa1) edge node {} (y2);
\path [->] (oa1) edge node {} (yn-1);
\draw[loosely dotted, line width=0.5mm] (y2) -- (yn-1);

\end{scope}
\end{tikzpicture}
$$
Here, we have $t^0_n \in B(c_1, o)$
and $t^1_1 \in B(c_1, t^1_2)$.
Therefore, we have $t^0_n \ra t^1_2$ or $t^1_1 \ra o$.
The arc $t^0_n \ra t^1_2$ provides the trail 
$$
w \har \cdots \har t^0_n \ra t^1_2 \har \cdots \har y
$$
between $w$ and $Y$ without converging connections, and thus a contradiction with the definition of \eqref{trail:sequence_of_os}.
The arc $t^1_1 \ra o$ provides us with the trail 
\begin{equation}
\label{trail:minimal_trail_o_single_node}
    w \ra o \la t^1_1 \har \cdots \har t^1_n \ra y
\end{equation}
which is better than the trail \eqref{trail:sequence_of_os} according to $\smallerTRAIL$, unless $n_t(0) = 0$.
In that case, \eqref{trail:minimal_trail_o_single_node} is also a minimal trail in $\TRAILSbig{X}{Y}{Z}$ with one converging node which is equal to $o$.
Therefore, we can apply the same argument as in Case 1 to the trail \eqref{trail:minimal_trail_o_single_node}, which leads to a contradiction.


\subsection{Possible candidates by outgoing arc}
\label{subsec:lemma:existence_trail_outgoing_case}

\begin{assumption}
\label{assumpt:lemma:sequence_outgoing_arc}
    Assumption~\ref{assumpt:good_graphs} holds.
    Furthermore, $w\in B(\Ovk)\setminus \Ovk$
    and $o \in \Ovk$ are nodes such that
    
    (i) $o \ra w$,\;
    (ii) $\parentsdown{w}{o}\setminus \Ovk = \emptyset$,\;
    (iii) There is no arc from $\BOvk \setminus \Ovk$ to $\Ovk$.
\end{assumption}

\begin{lemma}
\label{lemma:sequence_of_ws_no_converging}
    Assume that Assumption~\ref{assumpt:lemma:sequence_outgoing_arc} holds and that $w \notin \PossCandOut{\Ovk}$.
    Then there exists a trail
    from $w$ to a node in
    $\Ovk \setminus ( \parentsdown{w}{o} \sqcup \{ o \} )$ which is activated by 
    $\parentsdown{w}{o} \sqcup \{ o \}$
    and contains no converging connections. 
\end{lemma}

\begin{proof}
    By assumption we have $w \notin \PossCandOut{\Ovk}$.
    Therefore, $w$ is not a possible candidate by the outgoing arc $o \ra w$.
    By the definition of a possible candidate by outgoing arc (see~\eqref{eq:cond_outgoing_arc}),
    this means that one of the following conditions must be violated.
    \begin{enumerate}
         \item $\parentsdown{w}{o} \subseteq \Ovk$.
         \item $\dsepbig{w}{\Ovk \setminus (\parentsdown{w}{o} \sqcup \{o\})}{\parentsdown{w}{o} \sqcup \{o\}}$.
    \end{enumerate}
    The first condition is satisfied, because $\parentsdown{w}{o}\setminus \Ovk = \emptyset$ by Assumption~\ref{assumpt:lemma:sequence_outgoing_arc}.
    Therefore, the second restriction must be violated, i.e.
    $\notdsepbig{w}{\Ovk \setminus (\parentsdown{w}{o} \sqcup \{o\})}{\parentsdown{w}{o} \sqcup \{o\}}$.
    This means that there exists a trail from $w$ to $\Ovk\setminus (\parentsdown{w}{o} \sqcup \{o\})$ activated by $\parentsdown{w}{o} \sqcup \{o\}$.

    \medskip

    Consequently, the set $\TRAILSbig{X}{Y}{Z}$ with $X = \{w\}$, $Y = \Ovk \setminus (\parentsdown{w}{o} \sqcup \{o\})$ and $Z = \parentsdown{w}{o} \sqcup \{o\}$ is not empty, and thus we can pick a minimal trail in this set:
    \begin{equation}
    \label{trail:sequence_of_ws}
        w \har \cdots \ra c_1 \la \cdotslong \ra c_C \la \cdots \har y.
    \end{equation}
    By Lemma~\ref{lemma:minimal_trail_for_outgoing_arc}\ref{fact:w5}, such a trail has no converging connection, and this completes the proof of this lemma.
\end{proof}


\begin{lemma}
\label{lemma:minimal_trail_for_outgoing_arc}
    Assume that Assumption~\ref{assumpt:lemma:sequence_outgoing_arc} holds and take a minimal trail
    $$
    w \har \cdots \ra c_1 \la \cdotslong \ra c_C \la \cdots \har y
    $$
    in the set $\TRAILSbig{X}{Y}{Z}$ with
    $X = \{w\}$,
    $Y =  \Ovk \setminus ( \parentsdown{w}{o} \sqcup \{ o \} )$
    and
    $Z = \parentsdown{w}{o} \sqcup \{ o \}$.
    Then, the following statements hold:
    \begin{enumerate}[label = (\roman*)]
        \item $c_1 = Z(c_1)$,
        where $Z(c_i)$ denotes the closest descendant of $c_i$ in the sense of Definition~\ref{def:closest_descendant}.
        \label{fact:w1}
        
        \item The graph contains the subgraph below.
        \begin{figure}[H]
        
        \centering
        \begin{tikzpicture}[scale=1, transform shape, node distance=1.2cm, state/.style={circle, font=\Large, draw=black}]
        \node (v1) {$w$};
        \node[right = 6cm of v1] (v2) {$c_1$};
        
        \node[above right = 1cm and 0.01cm of v1] (x1) {$t^{0}_1$};
        \node[above right = of x1] (x2) {$t^{0}_2$};
        \node[above left = 1cm and 0.01cm of v2] (xn) {$t^{0}_{n}$};
        \node[above left= of xn] (xn-1) {$t^{0}_{n-1}$};
    
        \begin{scope}[>={Stealth[length=6pt,width=4pt,inset=0pt]}]
        \path [->] (v2) edge node {} (v1);
        \path [->] (v2) edge node {} (x2);
        \path [->] (v2) edge node {} (xn-1);
        \path [->] (v2) edge node {} (x1);
        \path [->] (x1) edge node {} (v1);
        \path [->] (x2) edge node {} (x1);
        \path [->] (xn) edge node {} (xn-1);
        \path [->] (xn) edge node {} (v2);
        \draw[loosely dotted, line width=0.5mm] (x2) -- (xn-1);
        
        \end{scope}
        \end{tikzpicture}
        \end{figure}
        
        \label{fact:w2}
        
        \item $c_1 \la c_2$.
        \label{fact:w3}
        
        \item For all $i=1,\dots,C$, we have that $c_i = Z(c_i)$ and $c_i \la c_{i+1}$.
        \label{fact:w4}
        
        \item The number of converging connections is equal to zero, i.e. $C=0$.
        \label{fact:w5}
    \end{enumerate}
   
\end{lemma}


\begin{proof}

    First, note that $Y \sqcup Z = \Ovk$ have local relationships in the sense of Definition~\ref{def:propadj}, by Corollary \ref{cor:adjacent_os}.
    Therefore, by Theorem~\ref{thm:minimal_trails_local_relationship}\ref{prop:subgraphs}, $\G$ contains the subgraph below with
    $y \in Y$.

    \begin{figure}[H]
    \centering
    \begin{tikzpicture}[scale=1, transform shape, node distance=1cm, state/.style={circle, font=\Large, draw=black}]
    
    \node (c1) {$c_{1}$};
    \node[right = of c1] (c2) {$c_2$};
    \node[right = 3cm of c2] (cC) {$c_C$};
    \node[right = of cC] (o) {$y$};
    \node[left = of c1] (t0n) {$t^0_n$};
    \node[left = of t0n] (t0n-1) {$t^0_{n-1}$};
    \node[left = of t0n-1] (t01) {$t^0_{1}$};
    \node[left = of t01] (w) {$w$};

    \node[above = of c1] (c1h1) {};
    \node[above = of c1h1] (c1h2) {};
    \node[above = of c1h2] (Yc1) {$Z(c_1)$};

    \node[above = of c2] (c2h1) {};
    \node[above = of c2h1] (c2h2) {};
    \node[above = of c2h2] (Yc2) {$Z(c_2)$};

    \node[above = of cC] (cCh1) {};
    \node[above = of cCh1] (cCh2) {};
    \node[above = of cCh2] (YcC) {$Z(c_C)$};
    
    \begin{scope}[>={Stealth[length=6pt,width=4pt,inset=0pt]}]
    
    
        \path [-] (c1) edge[bend left=45] node {} (c2);
        \path [-] (cC) edge[bend left=45] node {} (o);
    
        \draw[transform canvas={yshift=0.21ex},-left to,line width=0.25mm] (c1) -- (c2);
        \draw[transform canvas={yshift=-0.21ex},left to-,line width=0.25mm] (c1) -- (c2);
    
        \draw[transform canvas={yshift=0.21ex},-left to,line width=0.25mm] (cC) -- (o);
        \draw[transform canvas={yshift=-0.21ex},left to-,line width=0.25mm] (cC) -- (o);
    
        \draw[loosely dotted, line width=0.5mm] (c2) -- (cC);

        \path [->] (c1) edge node {} (c1h1);
        \path [->] (c1h2) edge node {} (Yc1);
        \draw[loosely dotted, line width=0.5mm] (c1h1) -- (c1h2);

        \path [->] (c2) edge node {} (c2h1);
        \path [->] (c2h2) edge node {} (Yc2);
        \draw[loosely dotted, line width=0.5mm] (c2h1) -- (c2h2);

        \path [->] (cC) edge node {} (cCh1);
        \path [->] (cCh2) edge node {} (YcC);
        \draw[loosely dotted, line width=0.5mm] (cCh1) -- (cCh2);

        \draw[loosely dotted, line width=0.5mm] (t01) -- (t0n-1);
        \path [->] (t0n) edge node {} (c1);
        \draw[transform canvas={yshift=0.21ex},-left to,line width=0.25mm] (w) -- (t01);
        \draw[transform canvas={yshift=-0.21ex},left to-,line width=0.25mm] (w) -- (t01);
        \draw[transform canvas={yshift=0.21ex},-left to,line width=0.25mm] (t0n-1) -- (t0n);
        \draw[transform canvas={yshift=-0.21ex},left to-,line width=0.25mm] (t0n-1) -- (t0n);
        
    \end{scope}
    \end{tikzpicture}
    \end{figure}

    \subsubsection{Proof of Lemma~\ref{lemma:minimal_trail_for_outgoing_arc}\ref{fact:w1}}
    
    

    Let us assume that $c_1 \neq Z(c_1)$. Since $Z = \parentsdown{w}{o} \sqcup \{o\}$, we have that $Z(c_i)\ra w$ for all $i=1,\dots,C$.
    In particular, we get that $Z(c_1)\ra w$, and therefore $\G$ contains the subgraph below.

    \begin{figure}[H]
    \centering
    \begin{tikzpicture}[scale=1, transform shape, node distance=1cm, state/.style={circle, font=\Large, draw=black}]
    
    \node (c1) {$c_{1}$};
    \node[left = of c1] (yn) {$t^0_{n_t}$};
    \node[left = of yn] (ym+1) {$t^0_{m+1}$};
    \node[left = of ym+1] (ym) {$t^0_{m}$};
    \node[left = of ym] (ym-1) {$t^0_{m-1}$};
    \node[left = of ym-1] (y1) {$t^0_{1}$};
    \node[left = of y1] (w) {$w$};

    \node[above = of c1] (c1h1) {$d^1_1$};
    \node[above = of c1h1] (c1h2) {$d^1_{n_Z}$};
    \node[above = of c1h2] (Yc1) {$Z(c_1)$};
    
    \begin{scope}[>={Stealth[length=6pt,width=4pt,inset=0pt]}]   
        \path [->] (c1) edge node {} (c1h1);
        \path [->] (c1h2) edge node {} (Yc1);
        \draw[loosely dotted, line width=0.5mm] (c1h1) -- (c1h2);

        \path [->] (y1) edge node {} (w);
        \path [->] (yn) edge node {} (c1);
        \path [->] (ym) edge node {} (ym-1);
        \path [->] (ym) edge node {} (ym+1);

        \draw[loosely dotted, line width=0.5mm] (y1) -- (ym-1);
        \draw[loosely dotted, line width=0.5mm] (ym+1) -- (yn);

        \path [->] (Yc1) edge node {} (w);
        
    \end{scope}
    \end{tikzpicture}
    \end{figure}

    The undirected cycle above is an active cycle, unless the appropriate chords are present in $\G$.
    Several chords can be excluded:
    \begin{itemize}
        \item The trails $c_1 \ra d_1 \ra \cdots \ra d_{n_t} \ra Z(c_1)$ and $w \har t^0_1 \har \cdots \har t^0_{n_t}$ do not contain any chords by Theorem~\ref{thm:minimal_trails}\ref{prop:no_chords}.

        \item $\forall l = 0, \dots, n_t - 1$, $t^0_l \ra c_1$ results in a shorter trail (and $t^0_{n_t} \ra c_1$ is not a chord).

        \item $\forall l = 1, \dots, n_t$, $t^0_l \ra Z(c_1)$ results in a trail with less converging connections not in $Z$.
        
        \item $\forall j= 1, \dots, n_Z$, $\forall l = 1, \dots, n_t$, $t^0_l \ra d^1_j$ results in a trail with shorter descendant paths.
        
        \item $\forall j = 0, \dots , n_Z +1$, $\forall l = m, \dots, n_t$,
        $d^1_j \ra t^0_l$ results in a cycle.

        \item $\forall j = 0, \dots , n_Z$, $\forall l = 0, \dots, m-1$,
        $d^1_j \ra t^0_l$ result in a trail with less converging connections.

        \item $\forall j = 1, \dots, n_Z$, $w \ra d^1_j$ results in a cycle.


        \item $\forall l = 0, \dots, m-1$, $c_1 \ra t^0_l$ results in a trail with less converging connections.

    \end{itemize}

    Therefore, the only remaining chords are $Z(c_1) \ra t^0_l$ with $l = 1, \dots, m-1$.
    It is evident that all such arcs must be present to prevent the appearance of an active cycle in $\G$, giving us the subgraph below.
    
    \begin{figure}[H]
    \centering
    \begin{tikzpicture}[scale=1, transform shape, node distance=1cm, state/.style={circle, font=\Large, draw=black}]
    
    \node (c1) {$c_{1}$};
    \node[left = of c1] (yn) {$t_{n_t}^0$};
    \node[left = of yn] (ym+1) {$t_{m+1}^0$};
    \node[left = of ym+1] (ym) {$t_{m}^0$};
    \node[left = of ym] (ym-1) {$t_{m-1}^0$};
    \node[left = of ym-1] (y1) {$t_{1}^0$};
    \node[left = of y1] (w) {$w$};
    
    \node[above = of c1] (c1h1) {$d^1 _1$};
    \node[above = of c1h1] (c1h2) {$d^1_{n_Z}$};
    \node[above = of c1h2] (Yc1) {$Z(c_1)$};
    
    \begin{scope}[>={Stealth[length=6pt,width=4pt,inset=0pt]}]   
        \path [->] (c1) edge[red] node {} (c1h1);
        \path [->] (c1h2) edge[red] node {} (Yc1);
        \draw[loosely dotted, line width=0.5mm, red] (c1h1) -- (c1h2);

        \path [->] (y1) edge node {} (w);
        \path [->] (yn) edge[red] node {} (c1);
        \path [->] (ym) edge[red] node {} (ym-1);
        \path [->] (ym) edge[red] node {} (ym+1);

        \draw[loosely dotted, line width=0.5mm] (y1) -- (ym-1);
        \draw[loosely dotted, line width=0.5mm, red] (ym+1) -- (yn);

        \path [->] (Yc1) edge node {} (w);
        \path [->] (Yc1) edge node {} (y1);
        \path [->] (Yc1) edge[red] node {} (ym-1);
        
    \end{scope}
    \end{tikzpicture}
    \end{figure}

    The subgraph above contains an undirected cycle with one converging connection (at $t^0_{m-1}$), coloured in red.
    Because there are no more chords which could be present, this undirected cycle must be of length smaller than 4, see Definition~\ref{def:active_cycle}.
    The undirected cycle consists of the nodes $c_1$, $Z(c_1)$,
    $t_{m-1}^0, t^0_{m}, \dots, t^0_{n_t}$ and 
    $d^1_1, \dots,  d^1_{n_Z}$;
    therefore it is of length $2 + n_t - (m-1) + 1 + n_Z = n_Z + n_t - m + 4$.
    This means that $n_Z + n_t - m + 4 \leq 3$, and therefore $n_Z + n_t - m \leq -1$.
    The equality can only hold if $n_Z = 0$ and $m = n_t + 1$.
    This is not possible because $t_m$ is a diverging connection, while $t^0_{n_t+1} = c_1$ is a converging connection.
    
    So, if $c_1 \neq Z(c_1)$, we have shown that $\G$ contains an active cycle, and therefore we have proven that $c_1 = Z(c_1)$.

    \subsubsection{Proof of Lemma~\ref{lemma:minimal_trail_for_outgoing_arc}\ref{fact:w2}}
    

    
    We know that $\G$ contains the trail
    $$
    w \har t^0_1 \har \cdots \har t^0_n \ra c_1.
    $$
    By \ref{fact:w1}, $c_1 \in Z = \parentsdown{w}{o} \sqcup \{o\}$.
    Because this is a shortest trail activated by the empty set ending with a rightward arrow ($t^0_n \ra c_1$) consisting of nodes in $\V \setminus Z$ and $c_1 = Z(c_1) \ra w$ by definition of the set $Z$, we can apply Lemma~\ref{lemma:generalisation_to_K} and Theorem~\ref{lemma:po_active_cycle1} (with $v_1=c_1$ and $v_2= w$ in the notation of Theorem~\ref{lemma:po_active_cycle1}) to find that $\G$ contains the subgraph as claimed.
    
    Furthermore, the length $n$ of this trail must be strictly larger than zero.
    If it were of length zero, then it would simply be the arc $w \ra c_1$.
    However, this would result in a cycle, as we have shown that the arc $c_1 \ra w$ must be present.

    \subsubsection{Proof of Lemma~\ref{lemma:minimal_trail_for_outgoing_arc}\ref{fact:w3}}
    
    

    By Theorem~\ref{thm:minimal_trails_local_relationship}\ref{prop:adjacency}, we know that $c_1$ and $c_2$ are adjacent.
    Therefore, it suffices to show that $c_1 \nrightarrow c_2$.
    
    Suppose that $c_1 \ra c_2$. Combining \ref{fact:w2} with Theorem~\ref{thm:minimal_trails_local_relationship}\ref{prop:subgraphs} leads to the conclusion that $\G$ contains the subgraph below.
    
    $$
    \begin{tikzpicture}[scale=1, transform shape, node distance=1.2cm, state/.style={circle, font=\Large, draw=black}]
    \node (v1) {$w$};
    \node[right = 6cm of v1] (v2) {$c_1$};
    \node[right = 6cm of v2] (v3) {$c_2$};
    
    \node[above right = 1cm and 0.01cm of v1] (x1) {$t^{0}_1$};
    \node[above right = of x1] (x2) {$t^{0}_2$};
    \node[above left = 1cm and 0.01cm of v2] (xn) {$t^{0}_{n}$};
    \node[above left= of xn] (xn-1) {$t^{0}_{n-1}$};
    
    \node[above right = 1cm and 0.01cm of v2] (y1) {$t^{1}_1$};
    \node[above right = of y1] (y2) {$t^{1}_2$};
    \node[above left = 1cm and 0.01cm of v3] (yn) {$t^{1}_{n}$};
    \node[above left= of yn] (yn-1) {$t^{1}_{n-1}$};
    \begin{scope}[>={Stealth[length=6pt,width=4pt,inset=0pt]}]
    
    \path [->] (v2) edge node {} (v1);
    \path [->] (v2) edge node {} (x2);
    \path [->] (v2) edge node {} (xn-1);
    \path [->] (v2) edge node {} (x1);
    \path [->] (x1) edge node {} (v1);
    \path [->] (x2) edge node {} (x1);
    \path [->] (xn) edge node {} (xn-1);
    \path [->] (xn) edge node {} (v2);
    \draw[loosely dotted, line width=0.5mm] (x2) -- (xn-1);

    \path [->] (v2) edge node {} (v3);
    \path [->] (v2) edge node {} (y2);
    \path [->] (v2) edge node {} (yn-1);
    \path [->] (v2) edge node {} (yn);
    \path [->] (y1) edge node {} (v2);
    \path [->] (y1) edge node {} (y2);
    \path [->] (yn-1) edge node {} (yn);
    \path [->] (yn) edge node {} (v3);
    \draw[loosely dotted, line width=0.5mm] (y2) -- (yn-1);
    
    \end{scope}
    \end{tikzpicture}
    $$
    
    Here, we have that $t^{0}_n \in B(c_i, t^{0}_{n-1})$ and $t^1_1 \in B(c_1, t^1_2)$.
    Since $\G$ does not contain any interfering v-structures, we must have $t^{1}_1 \ra t^{0}_{n-1}$ or $t^{0}_n \ra t^1_2$.
    However, both these arcs result in a better trail than \eqref{trail:sequence_of_ws}.
    Indeed, the trails
    \begin{align*}
        w \la \cdots \la t^0_{n-1} \la t^1_1 \har \cdotslong y,\\
        w \la \cdots \la t^0_n \ra t^1_2 \ra \cdotslong \har y,
    \end{align*}
    contain one fewer converging connection than \eqref{trail:sequence_of_ws}, which gives a contradiction as claimed.

    
    \subsubsection{Proof of Lemma~\ref{lemma:minimal_trail_for_outgoing_arc}\ref{fact:w4}}
    
    

    If $C = 1$, then the statement follows immediately by 
    \ref{fact:w1} and \ref{fact:w3}.
    If $C > 1$, by \ref{fact:w1} and \ref{fact:w3}, we have that $c_1 = Z(c_1)$ and $c_1 \la c_2$.
    By Theorem~\ref{thm:minimal_trails_local_relationship}\ref{prop:adjacency}, we know that for all $i = 1, \dots, C$. $c_i$ and $c_{i+1}$ are adjacent
    By Lemma~\ref{prop:no_diverging_cs}, there cannot be any $i$ such that $c_{i-1} \la c_i \ra c_{i+1}$.
    Therefore, $i \in \{1,\dots,C\}$, $c_i \la c_{i+1}$.
    Consequently, we can apply Corollary~\ref{cor:all_cs_in_Z}\ref{prop:leftarrow_all_in_Z}
    to find that for all $i \in \{1,\dots,C\}$, $c_i \in Z$, and therefore $c_i = Z(c_i)$.

    \subsubsection{Proof of Lemma~\ref{lemma:minimal_trail_for_outgoing_arc}\ref{fact:w5}}
    
    
    Assume that $C \geq 1$, 
    then by \ref{fact:w4} we get $c_C \la y$.
    Combining this with Theorem~\ref{thm:minimal_trails_local_relationship}\ref{prop:subgraphs}, gives that $\G$ contains the subgraph below with the convention $c_0 := w$ in the case that $C = 1$.
    
    $$
    \begin{tikzpicture}[scale=1, transform shape, node distance=1.2cm, state/.style={circle, font=\Large, draw=black}]
    \node (v1) {$c_{C-1}$};
    \node[right = 7.5cm of v1] (v2) {$c_C$};
    \node[right = of v2] (v3) {$y$};
    
    \node[above right = 1cm and 0.01cm of v1] (x1) {$t^{C-1}_1$};
    \node[above right = of x1] (x2) {$t^{C-1}_2$};
    \node[above left = 1cm and 0.01cm of v2] (xn) {$t^{C-1}_{n}$};
    \node[above left= of xn] (xn-1) {$t^{C-1}_{n-1}$};
    
    \begin{scope}[>={Stealth[length=6pt,width=4pt,inset=0pt]}]
    
    \path [->] (v2) edge node {} (v1);
    \path [->] (v2) edge node {} (x2);
    \path [->] (v2) edge node {} (xn-1);
    \path [->] (v2) edge node {} (x1);
    \path [->] (x1) edge node {} (v1);
    \path [->] (x2) edge node {} (x1);
    \path [->] (xn) edge node {} (xn-1);
    \path [->] (xn) edge node {} (v2);
    \draw[loosely dotted, line width=0.5mm] (x2) -- (xn-1);

    \path [->] (v3) edge node {} (v2);
    
    \end{scope}
    \end{tikzpicture}
    $$

    \noindent
    We consider two cases; when $\BOvk \neq \pa{v}$ and when $B(\Ovk) = \pa{v}$.

    \medskip

    \noindent
    \underline{Case 1:} 
    Let us assume that $\BOvk \neq \pa{v}$ and let $b_q$ be its corresponding node in the sense of Definition~\ref{def:bsets}.
    Such a $b_q$ always exists otherwise we would necessarily have $B(\Ovk) = \pa{v}$.
    Remark that the nodes $c_C$ and $y$ are in $Y \sqcup Z = \Ovk$, and therefore they are in $\BOvk$.
    Furthermore, if $C > 1$, then the node $c_{C-1}$ is in $Z \subseteq \Ovk \subset \BOvk$, and if $C = 1$, then $c_{C-1} = c_0 := w$ where $w$ is in $\BOvk$ by the assumptions of the lemma.

    \medskip

    By Definition~\ref{def:bsets}, any node in $\BOvk$ has an arc pointing towards both $v$ and $b_q$, giving us the subgraph below.

    $$
    \begin{tikzpicture}[scale=1, transform shape, node distance=1.2cm, state/.style={circle, font=\Large, draw=black}]
    \node (v1) {$c_{C-1}$};
    \node[right = 7.5cm of v1] (v2) {$c_C$};
    \node[right = of v2] (v3) {$y$};
    \node[below = of v2] (v) {$v$};
    \node[left = of v] (bq) {$b_q$};

    \node[above right = 1cm and 0.01cm of v1] (x1) {$t^{C-1}_1$};
    \node[above right = of x1] (x2) {$t^{C-1}_2$};
    \node[above left = 1cm and 0.01cm of v2] (xn) {$t^{C-1}_{n}$};
    \node[above left= of xn] (xn-1) {$t^{C-1}_{n-1}$};
    
    \begin{scope}[>={Stealth[length=6pt,width=4pt,inset=0pt]}]
    
    \path [->] (v2) edge node {} (v1);
    \path [->] (v2) edge node {} (x2);
    \path [->] (v2) edge node {} (xn-1);
    \path [->] (v2) edge node {} (x1);
    \path [->] (x1) edge node {} (v1);
    \path [->] (x2) edge node {} (x1);
    \path [->] (xn) edge node {} (xn-1);
    \path [->] (xn) edge node {} (v2);
    \draw[loosely dotted, line width=0.5mm] (x2) -- (xn-1);

    \path [->] (v3) edge node {} (v2);
    \path [->] (v1) edge node {} (v);
    \path [->] (v2) edge node {} (v);
    \path [->] (v3) edge node {} (v);
    \path [->] (v1) edge node {} (bq);
    \path [->] (v2) edge node {} (bq);
    \path [->] (v3) edge node {} (bq);

    \path [->] (v) edge node {} (bq);
    
    \end{scope}
    \end{tikzpicture}
    $$
    
    Here, we have that $t^{C-1}_n \in B(c_C, t^{C-1}_{n-1})$, $y \in B(c_C, v)$ and $y \in B(c_C, b_q)$.
    Since $\G$ does not contain any interfering v-structures, we must have $y \ra t^{C-1}_{n-1}$,
    or both $t^{C-1}_n \ra v$ and $t^{C-1}_n \ra b_q$.
    The arc $y \ra t^{C-1}_{n-1}$ results in a trail 
    $$
    w \har \cdots \har t^{C-1}_{n-1} \la y
    $$
    with fewer converging connections than \eqref{trail:sequence_of_ws}.
    By this contradiction, the arcs $t^{C-1}_n \ra v$ and $t^{C-1}_n \ra b_q$ must be present, and therefore
    $t^{C-1}_n \in \pa{v} \cap \pa{b_q} = B(v, b_q) = B_q$.
   Since $B_q = \BOvk$,
    we obtain $t^{C-1}_n \in \BOvk$.
    Moreover, by Theorem \ref{thm:minimal_trails}\ref{prop:nodes_along_subtrails} we know that $t^{C-1}_n \notin \Ovk$, and thus $t^{C-1}_n \in \BOvk \setminus \Ovk$.

    \medskip
    
    The arc $t^{C-1}_n \ra c_C$ is now an arc from a node in $\BOvk \setminus \Ovk$ to a node in $\Ovk$ which is not possible by the assumptions of Lemma~\ref{lemma:sequence_of_ws_no_converging}.
    Therefore, this contradiction completes the proof of the case.

    \medskip

    \noindent
    \underline{Case 2:} 
    If $\BOvk = \pa{v}$, then by a similar argument as in the first case we find that $\G$ must contain the subgraph
    $$
    \begin{tikzpicture}[scale=1, transform shape, node distance=1.2cm, state/.style={circle, font=\Large, draw=black}]
    \node (v1) {$c_{C-1}$};
    \node[right = 7.5cm of v1] (v2) {$c_C$};
    \node[right = of v2] (v3) {$y$};
    \node[below = of v2] (v) {$v$};
    
    \node[above right = 1cm and 0.01cm of v1] (x1) {$t^{C-1}_1$};
    \node[above right = of x1] (x2) {$t^{C-1}_2$};
    \node[above left = 1cm and 0.01cm of v2] (xn) {$t^{C-1}_{n}$};
    \node[above left= of xn] (xn-1) {$t^{C-1}_{n-1}$};
    
    \begin{scope}[>={Stealth[length=6pt,width=4pt,inset=0pt]}]
    
    \path [->] (v2) edge node {} (v1);
    \path [->] (v2) edge node {} (x2);
    \path [->] (v2) edge node {} (xn-1);
    \path [->] (v2) edge node {} (x1);
    \path [->] (x1) edge node {} (v1);
    \path [->] (x2) edge node {} (x1);
    \path [->] (xn) edge node {} (xn-1);
    \path [->] (xn) edge node {} (v2);
    \draw[loosely dotted, line width=0.5mm] (x2) -- (xn-1);

    \path [->] (v3) edge node {} (v2);
    \path [->] (v1) edge node {} (v);
    \path [->] (v2) edge node {} (v);
    \path [->] (v3) edge node {} (v);
    
    \end{scope}
    \end{tikzpicture}
    $$

    Remark that there are potential interfering v-structures at the nodes $y$ and $t^{C-1}_n$ to $c_C$.
    As in the previous case, we find that the arc $t^{C-1}_n \ra v$ must be present.
    This means that $t^{C-1}_n \in \pa{v} = \BOvk$, and therefore by Theorem~\ref{thm:minimal_trails}\ref{prop:nodes_along_subtrails} we have that $t^{C-1}_n \in \BOvk \setminus \Ovk$.
    Hence, we again find the arc $t^{C-1}_n \ra c_C$ from a node in $\BOvk \setminus \Ovk$ to a node in $\Ovk$ which is a contradiction.

    \medskip
    \noindent
    Thus, both cases  are not possible when $C > 0$, which completes the proof of Lemma~\ref{lemma:sequence_of_ws_no_converging}.
\end{proof}

\subsection{Auxiliary lemmas for the proofs in Sections \ref{subsec:lemma:existence_trail_incoming_case} and \ref{subsec:lemma:existence_trail_outgoing_case}}

\begin{lemma}
    Assume that $\G = (\V, \E)$ is a DAG with no interfering v-structures.
    Let $X,Y,Z \subseteq \V$ be three disjoint subsets and $Y \sqcup Z$ has local relationships (\cref{def:propadj}).
    Assume that $\TRAILSbig{X}{Y}{Z} \neq \emptyset$ and let $T$ a trail of the form (\ref{trail:lemma_XYZ_Cs}) be a minimal element of $\TRAILSbig{X}{Y}{Z}$
    with respect to the order $\smallerTRAIL$.
    
    \medskip
    
    \noindent
    Then for all $i=2, \dots, C$,
    the trail $c_{i-1} \la c_i \ra c_{i+1}$ can not be present in $\G$.
    \label{prop:no_diverging_cs}
\end{lemma}

\begin{proof}
    Suppose that there exists such a diverging connection.
    By Theorem~\ref{thm:minimal_trails_local_relationship}\ref{prop:subgraphs},
    $\G$ contains the subgraph below.
    $$
    \begin{tikzpicture}[scale=1, transform shape, node distance=1.2cm, state/.style={circle, font=\Large, draw=black}]
    \node (v1) {$c_{i-1}$};
    \node[right = 6cm of v1] (v2) {$c_i$};
    \node[right = 6cm of v2] (v3) {$c_{i+1}$};
    
    \node[above right = 1cm and 0.01cm of v1] (x1) {$t^{i-1}_1$};
    \node[above right = of x1] (x2) {$t^{i-1}_2$};
    \node[above left = 1cm and 0.01cm of v2] (xn) {$t^{i-1}_{n}$};
    \node[above left= of xn] (xn-1) {$t^{i-1}_{n-1}$};
    
    \node[above right = 1cm and 0.01cm of v2] (y1) {$t^{i}_1$};
    \node[above right = of y1] (y2) {$t^{i}_2$};
    \node[above left = 1cm and 0.01cm of v3] (yn) {$t^{i}_n$};
    \node[above left= of yn] (yn-1) {$t^{i}_{n-1}$};
    \begin{scope}[>={Stealth[length=6pt,width=4pt,inset=0pt]}]
    
    \path [->] (v2) edge node {} (v1);
    \path [->] (v2) edge node {} (x2);
    \path [->] (v2) edge node {} (xn-1);
    \path [->] (v2) edge node {} (x1);
    \path [->] (x1) edge node {} (v1);
    \path [->] (x2) edge node {} (x1);
    \path [->] (xn) edge node {} (xn-1);
    \path [->] (xn) edge node {} (v2);
    \draw[loosely dotted, line width=0.5mm] (x2) -- (xn-1);

    \path [->] (v2) edge node {} (v3);
    \path [->] (v2) edge node {} (y2);
    \path [->] (v2) edge node {} (yn-1);
    \path [->] (v2) edge node {} (yn);
    \path [->] (y1) edge node {} (v2);
    \path [->] (y1) edge node {} (y2);
    \path [->] (yn-1) edge node {} (yn);
    \path [->] (yn) edge node {} (v3);
    \draw[loosely dotted, line width=0.5mm] (y2) -- (yn-1);
    
    \end{scope}
    \end{tikzpicture}
    $$
    Remark that 
    \begin{align*}
        &t^{i-1}_{n}\in B(c_i,t^{i-1}_{n-1})
        = \pa{c_i} \cap \pa{t^{i-1}_{n-1}},\\
        &t^i_1\in B(c_i, t^i_2)
        = \pa{c_i} \cap \pa{t^{i}_{2}}.
    \end{align*}
    Since $\G$ does not contain any interfering v-structures, we must have
    $t^i_1 \in B(c_i,t^{i-1}_{n-1})$ or $t^{i-1}_n \in B(c_i, t^i_2)$.
    This means that $t^{i}_1\ra t^{i-1}_{n-1}$ or $t^{i-1}_n\ra t^i_2$.
    However, both arcs result in the existence of trails between $x$ and $y$ that have less converging connections than \eqref{trail:lemma_XYZ_Cs}, and therefore are better trails than \eqref{trail:lemma_XYZ_Cs} whcih contradict assumptions.
   Hence, there cannot be  diverging connection $c_{i-1} \la c_i \ra c_{i+1}$, concluding the proof.
\end{proof}

\medskip

\begin{lemma}
    \label{lemma:super_o_for_sequence}
    Under Assumption~\ref{assumpt:good_graphs}, let $Y,Z$ be two subsets such that $Y \sqcup Z = \Ovk$, and let $y \in Y$.
    Consider the subgraph below where:
    \begin{itemize}
        \item The trail 
        $$
        c_1 \la \cdots \ra c_2 \la \cdotslong \ra c_C \la \cdots \har y
        $$
        has $C$ converging connections corresponding to the nodes $\{ c_i \}_{i=1}^C$ with $C>1$.
        \item Each $c_i$ is either contained in $Z$ or it has a closest descendant $Z(c_i)$ in $Z$.
        \item All nodes on the trail and descendant paths not equal to $y$ or $Z(c_i)$ with $i \in \{1,\dots,C \}$ are in $\V \setminus \Ovk$.
    \end{itemize}
    \begin{figure}[H]
    \centering
    \begin{tikzpicture}[scale=1, transform shape, node distance=1.2cm, state/.style={circle, font=\Large, draw=black, minimum size=1cm}]
        \node(c1) {$c_1$};
        \node[right of=c1] (c1r) {};
        \node[right of=c1r] (c2l) {};
        \node[right of=c2l] (c2) {$c_2$};
        \node[right of=c2] (c2r) {};
        \node[right= 2.5cm of c2r] (cCl) {};
        \node[right of=cCl] (cC) {$c_C$};
        \node[right of=cC] (cCr) {};
        \node[right of=cCr] (cl) {};
        \node[right of=cl] (o) {$y$};

        \node[above of = c1] (c1a) {};
        \node[above of = c1a] (oc1u) {};
        \node[above of = oc1u] (oc1) {$Z(c_1)$};

        \node[above of = c2] (c2a) {};
        \node[above of = c2a] (oc2u) {};
        \node[above of = oc2u] (oc2) {$Z(c_2)$};

        \node[above of = cC] (cCa) {};
        \node[above of = cCa] (ocCu) {};
        \node[above of = ocCu] (ocC) {$Z(c_C)$};

        \begin{scope}[>={Stealth[length=6pt,width=4pt,inset=0pt]}]
            \path [->] (c1r) edge node {} (c1);
            \draw[loosely dotted, line width=0.5mm] (c1r) -- (c2l);
            \path [->] (c2l) edge node {} (c2);
            
            \path [->] (c2r) edge node {} (c2);
            \draw[loosely dotted, line width=0.5mm] (c2r) -- (cCl);
            \path [->] (cCl) edge node {} (cC);
            
            \path [->] (cCr) edge node {} (cC);
            \draw[loosely dotted, line width=0.5mm] (cCr) -- (cl);
            \draw[transform canvas={yshift=0.21ex},-left to,line width=0.25mm] (cl) -- (o);
            \draw[transform canvas={yshift=-0.21ex},left to-,line width=0.25mm] (cl) -- (o);

            \path [->] (c1) edge node {} (c1a);
            \path [->] (oc1u) edge node {} (oc1);
            \draw[loosely dotted, line width=0.5mm] (c1a) -- (oc1u);

            \path [->] (c2) edge node {} (c2a);
            \path [->] (oc2u) edge node {} (oc2);
            \draw[loosely dotted, line width=0.5mm] (c2a) -- (oc2u);

            \path [->] (cC) edge node {} (cCa);
            \path [->] (ocCu) edge node {} (ocC);
            \draw[loosely dotted, line width=0.5mm] (cCa) -- (ocCu);
        
        \end{scope}
    \end{tikzpicture}
    \end{figure}
    
    Then, there exists a node $\ot \in \Ovk$ such that $y \ra \ot$ and for all $i=1,\dots,C$, $Z(c_i) \ra \ot$
    whenever this does not result in the self-loop $\ot \ra \ot$.
    Indeed, the node $\ot$ may be equal to any node in $\Ovk$ including $y$ and $Z(c_i)$ with $i \in \{1, \dots, C \}$.
\end{lemma}

\begin{proof}
    Remark that 
    $\forall i=1, \dots, C$, $Z(c_i)\in Z \subseteq \Ovk$, and the node $y \in Y \subseteq \Ovk$.
    For convenience, we use the conventions $c_{C+1} := y$ and $Z(c_{C+1}) := y$.
    Therefore, the set
    $\{Z(c_i)\}_{i=1}^{C+1}$
    must have a highest node according to the ordered set $\Ovk$.
    We denote such a node by $o_{max}$ and pick $j \in \{1, \dots, C+1 \}$ such that $Z(c_j) = o_{max}$.
    This highest node $o_{max}$ must have been a possible candidate to some partial order $\Ot \subseteq \Ovk$ which contains all other nodes in the set, i.e.
    $\{Z(c_i): i = 1, \dots, C+1; i \neq j \}$.

    \medskip

    Consequently, the node $Z(c_j)$ must be a possible candidate to a set $\Ot \subseteq \Ovk$ which contains
    $\{ Z(c_i): i = 1, \dots, C+1; i \neq j \}$.
    We now prove that it cannot be a candidate by independence. 
    Observe that at least one of the following trails exist:
    \begin{align*}
        Z(c_j) \la \cdots \la c_j \la \cdots \ra c_{j+1} \ra \cdots \ra Z(c_{j+1}),\\
        Z(c_j) \la \cdots \la c_j \la \cdots \ra c_{j+1} \ra \cdots \ra Z(c_{j-1}).
    \end{align*}
    These are trails with no converging connections between $Z(c_j)$ and
    $\{Z(c_{j-1}), Z(c_{j+1}) \} \subseteq \Ot$
    consisting of nodes in $\V \setminus \Ovk \subseteq \V \setminus \Ot$.
    Therefore, we have that $\notdsepbig{Z(c_j)}{\Ot}{\emptyset}$, and therefore $Z(c_j) \notin \PossCandInd{\Ot}$, see Proposition~\ref{prop:characterization_poss_candidates}.
    This means that $Z(c_j)$ must be in $\PossCandIn{\Ot}$ or $\PossCandOut{\Ot}$.
    We consider both cases.

    \medskip
    
    \noindent
    \underline{Case 1:} Suppose that $Z(c_j) \in \PossCandIn{\Ot}$.
    Then, by Proposition~\ref{prop:characterization_poss_candidates}, there exists an $\ot\in \Ot$ such that $Z(c_j) \ra \ot$ satisfying
        \begin{enumerate}[label=\arabic*.]
            \item $\parentsdown{\ot}{Z(c_j)}\subseteq \Ot$,
            \item $\dsepbig{Z(c_j)}{\Ot \setminus (\parentsdown{\ot}{Z(c_j)} \sqcup \{\ot\} )}{\parentsdown{\ot}{Z(c_j)} \sqcup \{\ot\}}$.
        \end{enumerate}
    We will show that this implies that for all 
    $i\neq j$, $Z(c_i) \in \parentsdown{\ot}{Z(c_j)} \sqcup \{ \ot \}$.
    This means that each $Z(c_i)$ with $i \neq j$ points towards $\ot$ or is equal to $\ot$.
    Moreover, by the construction above we also know that $Z(c_j) \ra \ot$.
    This finishes the proof of Lemma~\ref{lemma:super_o_for_sequence}.

    \medskip
    Consider the nodes $Z(c_{j-1})$ and $Z(c_{j+1})$ (assuming that they exist).
    Suppose that the nodes $Z(c_{j-1})$ and $Z(c_{j+1})$ are not in $\parentsdown{\ot}{Z(c_j)} \sqcup \{\ot\}$.
    They are connected to $Z(c_j)$ by the trails
    \begin{align*}
        Z(c_j) \la \cdots \la c_j \la \cdots \ra c_{j+1} \ra \cdots \ra Z(c_{j+1}),\\
        Z(c_j) \la \cdots \la c_{j} \la \cdots \ra c_{j-1} \ra \cdots \ra Z(c_{j-1}),
    \end{align*}
    which contain no converging connections nor nodes in $\parentsdown{\ot}{Z(c_j)} \sqcup \{\ot\} \subseteq  \Ot \subseteq \Ovk$.
    Therefore, these trails are activated by 
    $\parentsdown{\ot}{Z(c_j)} \sqcup \{\ot\}$.
    Thus, they are trails from $Z(c_j)$ to 
    $\Ot\setminus (\parentsdown{\ot}{Z(c_j)} \sqcup \{\ot\})$ activated by $\parentsdown{\ot}{Z(c_j)} \sqcup \{\ot\}$.
    This means that
    $\notdsepbig{Z(c_j)}{\Ot\setminus (\parentsdown{\ot}{Z(c_j)} \sqcup \{\ot\})}{\parentsdown{\ot}{Z(c_j)} \sqcup \{\ot\}}$
    which contradicts the assumption that the second restriction is satisfied.
    Therefore, we must have $Z(c_{j-1}),Z(c_{j+1}) \in \parentsdown{\ot}{Z(c_j)} \sqcup \{\ot\}$.

    \medskip
    
    Now, the trails
    \begin{align*}
        Z(c_j) \la \cdots \la c_j \la  \cdots \ra c_{j+1} \la \cdots \ra c_{j+2}\ra \cdots \ra Z(c_{j+2}),\\
        Z(c_j) \la \cdots \la c_j \la  \cdots \ra c_{j-1} \la \cdots \ra c_{j-2}\ra \cdots \ra Z(c_{j-2})
    \end{align*}
    are activated by $\parentsdown{\ot}{Z(c_j)} \sqcup \{\ot\}$ since the converging connections at $c_{j-1}$ and $c_{j+1}$ have a descendant ($Z(c_{j-1})$ and $Z(c_{j+1})$, respectively) in $\parentsdown{\ot}{Z(c_j)} \sqcup \{\ot\}$.
    Thus, by the same argument $Z(c_{j-2})$ and $Z(c_{j+2})$ are in $\parentsdown{\ot}{Z(c_j)} \sqcup \{\ot\}$.

    \medskip

    We conclude this proof by induction. Indeed, the same argument can be repeated to show that for any $k$,
    $Z(c_{j+k}) \in \parentsdown{\ot}{Z(c_j)} \sqcup \{\ot\}$
    (resp. $Z(c_{j-k}) \in \parentsdown{\ot}{Z(c_j)} \sqcup \{\ot\}$)
    whenever $1 \leq j+k \leq C+1$
    (resp. $1 \leq j-k \leq C+1$).
    Therefore, we have proved that for all $i\neq j$, $Z(c_{i}) \in \parentsdown{\ot}{Z(c_j)} \sqcup \{\ot\}$, which completes the proof in this case.

    \medskip

    \noindent
    \underline{Case 2:}
    Suppose that $Z(c_j) \in \PossCandOut{\Ot}$.
    Then, there exists an $\ot \in \Ot$ with $\ot \ra Z(c_j)$ satisfying
        \begin{enumerate}[label=\arabic*.]
            \item $\parentsdown{Z(c_j)}{\ot}\subseteq \Ot$,
            \item $\dsepbig{Z(c_j)}{\Ot\setminus (\parentsdown{Z(c_j)}{\ot} \sqcup \{\ot\})}{\parentsdown{Z(c_j)}{\ot} \sqcup \{\ot\}}$.
        \end{enumerate}
    We will show that for all $i\neq j$, $Z(c_i) \in \parentsdown{Z(c_j)}{\ot} \sqcup \{\ot\}$. Therefore, $i\neq j$, $Z(c_i) \ra Z(c_j) \in \Ovk$, so $Z(c_j)$ is our desired $\ot$. This finishes the proof of \ref{lemma:super_o_for_sequence} in this case.

    \medskip

    First, consider the nodes $Z(c_{j+1})$ and $Z(c_{j-1})$ (assuming that they exist).
    If they are both in $\Ot \setminus (\parentsdown{Z(c_j)}{\ot} \sqcup \{\ot\})$,
    then
    \begin{align*}
        Z(c_j) \la \cdots \la c_j \la \cdots \ra c_{j+1} \la \cdots \ra Z(c_{j+1}),\\
         Z(c_j) \la \cdots \la c_j \la \cdots \ra c_{j-1} \la \cdots \ra Z(c_{j-1}),
    \end{align*}
    contain no converging connections nor nodes in $\parentsdown{Z(c_j)}{\ot} \sqcup \{\ot\} \subseteq \Ovk$.
    Therefore, they are activated by $\parentsdown{Z(c_j)}{\ot} \sqcup \{\ot\}$.
    This means that
    $\notdsepbig{Z(c_j)}{\Ot\setminus (\parentsdown{Z(c_j)}{\ot} \sqcup \{\ot\})}{\parentsdown{Z(c_j)}{\ot} \sqcup \{\ot\}}$
    which contradicts the assumption that the second restriction is satisfied.
    Therefore, we must have $Z(c_{j-1}),Z(c_{j+1}) \in \parentsdown{Z(c_j)}{\ot} \sqcup \{\ot\}$.

    \medskip

    Now, the trails
    \begin{align*}
    Z(c_j) \la \cdots \la c_j \la \cdots \ra c_{j+1} \la \cdots \ra c_{j+2}\ra \cdots \ra Z(c_{j+2}),\\
    Z(c_j) \la \cdots \la c_j \la \cdots \ra c_{j-1} \la \cdots \ra c_{j-2}\ra \cdots \ra Z(c_{j-2}),
    \end{align*}
    are activated by $\parentsdown{Z(c_j)}{\ot} \sqcup \{\ot\}$.
    Thus, by the same argument $Z(c_{j+2}), Z(c_{j-2}) \in \parentsdown{Z(c_j)}{\ot} \sqcup \{\ot\}$.

    \medskip

    Similarly as in the first case, the proof is finished by an induction argument, showing that
    for all $i\neq j$, $Z(c_{i}) \in \parentsdown{Z(c_j)}{\ot} \sqcup \{\ot\}$, as claimed above.
\end{proof}

\end{document}